\documentclass[twoside,leqno]{article}
\usepackage[letterpaper]{geometry}
\newcommand{\F}[1]{\mathcal{F}}

\usepackage{siamproceedings}
\usepackage{hyperref}

\usepackage{graphicx}

\usepackage{enumerate}
\usepackage{xcolor}

\usepackage{enumitem}

\usepackage{algorithm}
\usepackage{algpseudocode}

\usepackage{amsmath}
\usepackage{cleveref}

\usepackage[bbgreekl]{mathbbol}
\DeclareSymbolFontAlphabet{\mathbbl}{bbold}
\usepackage{amsmath}
\usepackage{amssymb}
\usepackage{ntheorem,cleveref}

\usepackage{thmtools, thm-restate}

\newtheorem{remark}{Remark}[section]
\newtheorem{observation}{Observation}[section]
\newtheorem{claim}{Claim}[section]
\newtheorem{algo}{Algorithm}[section]
\crefname{algo}{algorithm}{algorithms}
\Crefname{algo}{Algorithm}{Algorithms}

\crefname{problem}{problem}{problems}
\Crefname{problem}{Problem}{Problems}

\crefname{observation}{observation}{observations}
\Crefname{observation}{Observation}{Observations}
\crefname{invariant}{invariant}{invariants}
\Crefname{invariant}{Invariant}{Invariants}
\crefname{claim}{claim}{claims}
\Crefname{claim}{Claim}{Claims}

\newcommand{\ceil}[1]{\left\lceil #1 \right\rceil}
\newcommand{\floor}[1]{\left\lfloor #1 \right\rfloor}
\newcommand{\union}{\cup}
\newcommand{\inter}{\cap}

\newcommand{\card}[1]{\left|#1\right|}

\newcommand{\C}{\mathcal{C}}

\newcommand{\R}{\mathbb{R}}

\newcommand{\N}{\mathbb{N}}

\newcommand{\1}{\mathbbl{1}}
\newcommand{\vol}{\textbf{\textup{vol}}}
\newcommand{\invol}{\textbf{\textup{vol}}^{\textup{in}}}

\newcommand{\poly}{\mathrm{poly}}
\newcommand{\D}{\mathcal D}
\newcommand{\boundary}{\partial}
\renewcommand{\P}{\mathcal{P}}

\DeclareMathOperator*{\argmin}{arg\,min}
\newcommand{\PAR}[1]{\left( #1 \right)}
\newcommand{\lmin}{\lambda_{\min}}
\newcommand{\lmax}{\lambda_{\max}}
\newcommand{\eps}{\epsilon}
\newcommand{\Proba}{\mathbb P}

\newcommand{\pref}[1]{Property~(\ref{#1})}
\newcommand{\sref}[1]{Step~(\ref{#1})}
\newcommand{\alg}{\mathcal A}
\newcommand\elabel[1]{\label{#1}\addtocounter{equation}{1}\tag{\theequation}}
\newcommand{\proper}{proper }
\newcommand{\cross}{\times}

\newcommand{\justaflag}{}
\newcommand\flag[1]{%
    \leavevmode\marginpar{%
        \raisebox{\dimexpr-\totalheight+\ht\strutbox\relax}%
        [\dimexpr\ht\strutbox+17mm][\dp\strutbox]{\expandafter\includegraphics[width=0.01cm]{#1}}%
}}

\usepackage[colorinlistoftodos,prependcaption,textsize=tiny,textwidth=30mm]{todonotes}

\begin{document}

\title{Deterministic and Exact Fully-dynamic Minimum Cut\\of Superpolylogarithmic Size in Subpolynomial Time}

\author{Antoine El-Hayek\thanks{Institute of Science and Technology Austria, Klosterneuburg, Austria. email: \tt{antoine.el-hayek@ist.ac.at}}
\and Monika Henzinger\thanks{Institute of Science and Technology Austria, Klosterneuburg, Austria. email: \tt{monika.henzinger@ist.ac.at}}
\and Jason Li\thanks{Carnegie Mellon University. email: \tt{jmli@cs.cmu.edu}}
}

\date{}

\maketitle

\begin{abstract}

We present an exact fully-dynamic minimum cut algorithm that runs in $n^{o(1)}$ deterministic update time when the minimum cut size is at most $2^{\Theta(\log^{3/4-c}n)}$ for any $c>0$, improving on the previous algorithm of Jin, Sun, and Thorup~(SODA 2024) whose minimum cut size limit is $(\log n)^{o(1)}$. Combined with graph sparsification, we obtain the first $(1+\epsilon)$-approximate fully-dynamic minimum cut algorithm on \emph{weighted} graphs, for any $\epsilon\ge2^{-\Theta(\log^{3/4-c}n)}$, in $n^{o(1)}$ randomized update time.

Our main technical contribution is a deterministic local minimum cut algorithm, which replaces the randomized LocalKCut procedure from El-Hayek, Henzinger, and Li~(SODA 2025).

\end{abstract}

\section{Introduction}

We consider the minimum cut problem in an undirected, unweighted graph: find a set of edges of minimum size whose removal disconnects the graph. This is a basic problem in combinatorial optimization with a rich body of work, and its study has led to the discovery of many fundamental concepts in graph algorithms, such as randomized contractions~\cite{karger1996new} and graph sparsification~\cite{DBLP:conf/soda/Karger94}. The randomized running time was famously settled by Karger~\cite{karger2000minimum}, who also posed as an open problem whether a fast deterministic algorithm exists. That problem proved notoriously difficult and resisted progress for 20~years, until a recent sequence of works~\cite{kawarabayashi2015deterministic,li2020deterministic,li2021deterministic,DBLP:conf/soda/HenzingerLRW24} finally obtained a deterministic algorithm that is optimal up to polylogarithmic factors.

However, on dynamic graphs undergoing edge insertions and deletions, this problem is much less understood. For a long time, the state of the art algorithm for $(1+o(1))$-approximate minimum cut was due to Thorup~\cite{thorup2007fully}, which is Monte-Carlo randomized and runs in $\tilde{O}(\sqrt n)$ worst-case update time. Recently, El-Hayek, Henzinger, and Li~\cite{DBLP:conf/soda/HenzingerLRW24} improved the running time to $n^{o(1)}$ amortized update time~\cite{DBLP:conf/soda/El-HayekH025}, also by a Monte-Carlo randomized algorithm. In comparison, the fastest deterministic approximation algorithm is the $\sqrt{2+o(1)}$-approximate algorithm of Thorup and Karger~\cite{thorup2000dynamic}, which runs in $O(\lambda^2\textup{polylog}(n))$ time where $\lambda$ is the minimum cut of the graph, but is only capable of outputting the (approximate) cut-size of the minimum cut. For larger approximations,  Goranci, Räcke, Saranurak and Tan~\cite{expanderhierarchy}, then van den Brand, Chen, Kyng, Liu, Meierhans, Gutenberg, and Sachdeva~\cite{DBLP:conf/focs/Brand0KLMGS24} gave an $n^{o(1)}$-approximate algorithm that runs in $n^{o(1)}$ worst-case deterministic update time, in unweighted and weighted graphs respectively. For exact minimum cut, Jin, Sun, and Thorup~\cite{DBLP:conf/soda/JinST24} obtained $n^{o(1)}$ worst-case deterministic update time when the minimum cut-size is $(\log n)^{o(1)}$, while Goranci, Henzinger, Nanongkai, Saranurak, Thorup, and Wulff-Nilsen~\cite{DBLP:conf/soda/GoranciHNSTW23} obtained the first sublinear-time algorithm for all cut-sizes, which takes $\tilde{O}(n)$ worst-case randomized time or $\tilde{O}(m^{30/31})$ amortized deterministic time, which was further improved by de Vos and Christiansen to $O(m^{11/12})$ update time with a deterministic algorithm. Overall, the randomized running times are notably better than the deterministic ones, which suggests that deterministic minimum cut is notoriously difficult even in the dynamic setting.

In this paper, we significantly advance the state of the art for exact algorithms, obtaining an algorithm in $n^{o(1)}$ amortized deterministic update time for minimum cut-sizes up to $2^{\Theta(\log^{3/4-c}n)}$ for any $c>0$. This is a significant leap over the $(\log n)^{o(1)}$ cut-size limit of Jin, Sun, and Thorup~\cite{DBLP:conf/soda/JinST24}, at a cost of amortized instead of worst-case update time. Most notably, due to the increased minimum cut-size that our algorithm can handle, we can use the standard approach of sparsifying the graph to reduce the minimum cut-size to $O(\frac{\log n}{\epsilon^2})$ and then applying our algorithm on the resulting graph. This leads to the first 
$(1+o(1))$-approximate dynamic minimum cut algorithm even on \emph{weighted} graphs in $n^{o(1)}$ update time.

To construct our deterministic algorithm, we replace the randomized LocalKCut procedure of~\cite{DBLP:conf/soda/El-HayekH025} with an entirely new algorithm based on tree packing and color coding, based on the techniques of Lokshtanov, Saurabh, and Surianarayanan~\cite{lokshtanov2022parameterized}. This new deterministic LocalKCut procedure is one of our main technical contributions. Furthermore, to increase the minimum cut-size bound to $2^{\Theta(\log^{3/4-c}n)}$ for any $c>0$, we adapt an idea from the static setting~\cite{DBLP:conf/soda/HenzingerLRW24} to the dynamic setting.

\textbf{Our main result and its applications.} We next describe our main result as well as its applications. Since our main result is about \emph{proper} cuts, we need to generalize the minimum cut problem to the minimum \emph{proper} cut problem. 

\begin{definition}[Minimum Proper Cut]
Let $G=(V,E)$ be a graph, and $F\subset E$ be a subset of edges. 
$F$ is said to be a \emph{proper} cut of $G$ if the removal of $F$ strictly increases the number of connected components of $G$, and is minimal for this property.    
Equivalently, a set $U\subset V$ of nodes is said to be a proper cut if it is connected and has at least one outgoing edge.
\end{definition}
Thus, a minimum proper cut is the smallest non-zero cut in a potentially disconnected graph. 
Our data structure will work recursively, where in each level of the recursion only a subset of the cuts is analyzed. To do so, certain edges have been removed from the graph, which might lead to a disconnected subgraph on this level. Instead of starting a new data structure for each connected component, which is what is usually done in this case, we maintain \emph{one} data structure per level, representing the subgraph on this level. This is crucial for achieving our improvements, but in turn we need for each level an algorithm that is able to maintain a minimum \emph{proper} cut. 

More specifically, our main technical contribution is a deterministic algorithm that solves the dynamic minimum proper cut problem for specific parameter ranges:
\begin{restatable}[Bounded Dynamic Minimum Proper Cut]{problem}{mainproblem}\label{pb:boundeddynamic}
Let $G$ be an unweighted fully dynamic graph with $n$ nodes under edge updates, and let $\lmin < \lmax \in \N$ be two values with $\lmax \le 1.2 \lmin = 2^{O(\log^{3/4-c} n)}$ for some $c>0$. Assume that the minimum proper cut of $G$ is at least $\lmin$ at all times.
    The \emph{Bounded Dynamic Minimum Proper Cut} problem is to maintain a data structure that, depending on the value of the minimum proper cut $\lambda$, gives the following output:
    \begin{enumerate}[noitemsep]
    \item If $\lmin \le  \lambda \le \lmax$, outputs the (exact) value of the minimum proper cut. Moreover, a partition achieving the minimum proper cut should be stored implicitly.
    \item If $ \lambda > \lmax$, outputs ``$ \lambda > \lmax$''.
\end{enumerate}
\end{restatable}

\begin{theorem}\label{thm:boundeddyn}
    There exists a deterministic algorithm that solves \Cref{pb:boundeddynamic} in $n^{o(1)}$ update time.
\end{theorem}

Before we delve into the main techniques we developed to prove this theorem, we first discuss its implications. 
The first one is an exact algorithm for dynamic minimum cut for $\lambda < 2^{\Theta(\log^{3/4-c} n)}$ for any $c>0$:

\begin{theorem}[Fully Dynamic Exact Mincut]\label{thm:exact}
    There exists a deterministic algorithm that solves the minimum cut problem \emph{exactly} on a fully dynamic unweighted graph in $n^{o(1)}$ update time, as long as the value of the minimum cut is $\lambda < 2^{\Theta(\log^{3/4-c} n)}$ for some $c>0$. 
    It outputs the value of the minimum cut and stores a partition achieving the minimum cut implicitly.
\end{theorem}

To prove \Cref{thm:exact} from \Cref{thm:boundeddyn}, simply run a connectivity algorithm and additionally $O(\log n)$ many instances of the algorithm from \Cref{thm:boundeddyn}, one per range where $\lmin = 1.2^i$ and $\lmax = 1.2^{i+1}$, for all ranges, i.e., all values of $i\in[0,\lceil 2\log n \rceil ]$. 
Note that each instance is built for the whole graph, whether it is connected or disconnected, there is \emph{not} one instance per connected component and range -- in fact, we ensure, as we discuss below, that the input for every instance is connected at all times, by delaying updates to the different instances to ensure connectivity guarantees. 
Whenever the connectivity algorithm returns that the graph is disconnected, simply output that the minimum cut is $0$. Otherwise, we need to find the minimum proper cut, which is equal to the minimum cut.
To ensure that the minimum proper cut is at least $\lmin$ for every instance, we give first the initial graph to the instance with smallest index $i=0$.
If that instance returns ``$\lambda > \lmax$'', this ensures that $\lambda \ge 1.2^{i+1}$ which is the $\lmin$ of the next instance. We thus give the graph to the next instance, and do that iteratively until one instance returns a value for the minimum proper cut that is in its range.
We then store for each instance an integer ``Last update time'', that is either $0$ if it has been instantiated, or $-1$ otherwise.
To handle an update, we proceed similarly, updating the connectivity algorithm first, and if it returns that the graph is connected, updating or instantiating all the instances in increasing order of indices, and stopping when one instance returns ``$\lambda > \lmax$''.
During that phase, when considering an instance, we know that before updating it, either the instance has not been instantiated, or the graph it is considering -- which might be out of date -- satisfies $\lambda \ge \lmin$. 
In the first case, since we know from the instance with lower index that $\lambda \ge \lmin$ in the current graph, we can instantiate the instance with the current graph. 
In the latter case, we know, since the graph is connected (from the connectivity algorithm),
that handling all edge insertions then all edge deletions since the last update to the instance, will bring the instance up to date, and ensure that $\lambda\ge \lmin$ throughout.
We then update the variable ``Last update time'' accordingly.
To answer the query, we report the value of the one algorithm that outputs a value inside its range and is up to date.
We update the instances in this way as we have to guarantee that the minimum proper cut in any instance is at least $\lmin$ at all time.

\paragraph{A note on the maximum value of $\lambda$, $\lmin$ and $\lmax$:} It is worth discussing briefly why we require $\lmin = 2^{O(\log^{3/4-c} n)}$ for any $c>0$, as opposed to having an algorithm valid for all values of $\lambda$ that runs in time $\poly (\lmin) \cdot n^{o(1)}$. As we discuss below, our algorithm is in fact a $(1+\delta)$-approximate algorithm, for $\delta = 2^{-O(\log^{3/4-c} n)}$. Said differently, our algorithm can only be exact for values of $\lmax$ that are  less than $1/\delta$. We cannot decrease our approximation ratio, as our algorithm is recursive, and the graph size (number of edges in the graph)  and recourse for each recursive call depends on $\delta$ and $\lmax$: more precisely, the recourse per level is $\tilde O(\frac{\rho\lmax}{\delta^3})$ with $\rho = 2^{O(\sqrt{\log n})}$ for any $c>0$, and after a recursive depth of $r=\sqrt[4]{\log n}$, the goal is to have a recourse of $2^{O(\log^{1-c}n)} = n^{o(1)}$. 
Thus, making $\delta$ too small or $\lmax$ too high would make the overall recourse too high. 

Note also that our algorithm relies on the dynamic expander decomposition of Goranci, Räcke, Saranurak and Tan~\cite{expanderhierarchy}. This algorithm maintains an expander hierarchy with expander parameter $\phi=2^{-\Theta(\log^{3/4}n)}$ and recourse $\rho=2^{\Theta(\log^{1/2}n)}$. Our cluster decomposition achieves then the same depth as their expander hierarchy, a depth of $O(\sqrt[4]{\log n})$, which is a direct consequence of the choice of the value of $\phi$, as the number of edges is multiplied by $\phi$ from a level to the next, and $m\phi^{a\sqrt[4]{\log n}}=1$ for a constant $a$ that is large enough. Importantly, their recourse per level is truly smaller than $\frac 1 \phi$, and hence, after $O(\sqrt[4]{\log n})$ many levels, the recourse is $\rho^{O(\sqrt[4]{\log n})} = n^{o(1)}$.
One could hope for other dynamic expander decompositions with different tradeoffs for $\phi$ and $\rho$. In that case, our algorithm would adapt very effectively: For example, let $\gamma>0$ be a real value and assume that we have an algorithm that maintains an expander decomposition for $\phi = n^{-O(\gamma)}$.
The depth of recursion would then be $O(\frac 1 \gamma)$, and assume the recourse per level is $\rho$ for a new value of $\rho$. Then our recourse per level would be $\tilde O(\frac{\rho\lmax}{\delta^3})$, and over $O(\frac 1 \gamma)$ levels, this would add up to  $\left(\frac{\rho\lmax}{\delta}\right)^{O(\frac 1 \gamma)}$. Each update requiring $\left(\frac{\rho\lmax}{\phi\delta}\right)^{O(1)}$ time to be processed, this would yield a total update time of $\left(\frac{\rho\lmax}{\delta}\right)^{O(\frac 1 \gamma)}\frac 1 \phi = \left(\frac{\rho\lmax}{\delta}\right)^{O(\frac 1 \gamma)}n^{\gamma}$. This could be leveraged to get different tradeoffs between running time and largest acceptable $\lmax$.

The second implication of our algorithm is an approximation algorithm for the dynamic minimum cut problem, where there are no restrictions on the size of $\lambda$ nor on the edges being unweighted. Note that this is the first algorithm that maintains a $(1+\eps)$-approximate minimum cut in \emph{weighted} graphs:

\begin{theorem}[Fully Dynamic Approximate Mincut]\label{thm:weighted}
    Let $G_w$ be a weighted fully dynamic graph with $n$ nodes under edge updates, and let $\eps \ge 2^{-\Theta(\log^{3/4-c} n)}$ for any $c>0$.
    There exists a dynamic algorithm that maintains a $(1+\eps)$-approximation of the value of the minimum cut in $n^{o(1)}$ update time. Moreover, a partition achieving the minimum cut is stored implicitly.

    The algorithm is randomized and works against an oblivious adversary, and is correct with high probability.
\end{theorem}

Let us quickly discuss how we obtain \Cref{thm:weighted} from \Cref{thm:boundeddyn}.
Our first step is to sparsify the graph as to ensure that the minimum cut  value is reduced to at most $O(\frac {\log n} {\eps ^2}) < 2^{\Theta(\log^{3/4-c} n)}$ for some $c>0$ using the edge-sampling approach of Karger~\cite{kargersparse}.
It samples each edge with a probability $p$ aptly chosen as described in the following lemma:

\begin{restatable}[Theorem 2.1 and Corollary 2.1 of~\cite{kargersparse}]{theorem}{kargersparse}\label{thm:kargersparsify}
    Let $G$ be any graph with minimum cut $\lambda$, let $\eps >0$,  and let $p \ge \frac {54 \ln n} {\eps^2 \lambda}$.
    Let $G(p)$ be the graph obtained from $G$ by sampling each edge independently with probability $p$.
    Then with high probability, the minimum cut $S$ in $G(p)$ will correspond to a $(1+\eps)$-minimum cut of $G$.
    Moreover, also with high probability, the value of \emph{any} cut $S$ in $G(p)$ has value no less than $(1-\eps)$ and no more than $(1+\eps)$ times its expected value. 
\end{restatable}

We adapt this sparsifying technique to weighted graphs, and sparsify the graph for different values of estimates of $\lambda$ $\lambda_0, \dots, \lambda_{\log_{1.1}(n^2\cdot \frac U L)}$ where $\lambda_i = L\cdot 1.1^i$ for every $i \in [\log_{1.1}(n^2\cdot \frac U L)]$.
We then run the algorithm from \Cref{thm:boundeddyn} on each sparsified graph, and return the value of the algorithm with the correct estimate. 

\textbf{Organization.}
The paper is organized as follows: in \Cref{sec:prelims}, we formally define the problem and introduce definitions and notations. 
In \Cref{sec:technical overview}, we give a technical overview of our work.
In \Cref{sec:localkcut}, we present our deterministic LocalKCut. 
In \Cref{sec:fragmenting}, we present our fragmenting algorithm. 
In \Cref{sec:hierarchy}, we present the Cluster decomposition and Hierarchy.
In \Cref{sec:mirrorcuts}, we discuss how to maintain mirror cuts, a building block of our data structure.
In \Cref{sec:static}, we present a static algorithm for the minimum proper cut problem for a specific range, and make this algorithm dynamic in \Cref{sec:dynamic}.
In \Cref{sec:datastructure}, we present the last details needed for our data structure.
And finally, in \Cref{sec:weighted}, we discuss how to deal with weighted graphs.

\section{Preliminaries}\label{sec:prelims}
\paragraph{Problem Definition.}
We are given an undirected unweighted graph $G=(V,E)$, where $V$ is the set of vertices, and $E$ is the set of edges. Every vertex subset $\emptyset \subsetneq S \subsetneq V$ is also called a \emph{cut} of $G$. We denote by $\vol(S)$ the sum of the degrees of the nodes in $S$.
We define, for every $T \subseteq V$, a cost associated with $T$, that we call the \emph{boundary-size} of $T$, or equivalently, the \emph{(cut-)size} of $T$:
$$\boundary T = \card{E(T, V\setminus T)}$$

Typically, we will use the term ``cut-size'' when $T$ is a candidate to be the minimum cut, while we will prefer the term ``boundary-size'' when $T$ is a cluster in our cluster decomposition (which is the internal data structure that we introduce later on), and thus $\boundary T$ is a value that we do not wish to compare (yet) to the cut-sizes of other cuts.

A \emph{minimum proper cut} of $G$ is a set $MC(G) \subsetneq V$ that minimizes the cut-size, while still being non-zero: $MC(G) \in \argmin_{\varnothing \subsetneq T \subsetneq V, \boundary T \neq 0}\PAR{\boundary T}$. 
Note that there are dynamic  algorithms with polylogarithmic time per operation that maintain the connected components of a graph (See e.g. Holm, de Lichtenberg, and Thorup~\cite{DBLP:journals/jacm/HolmLT01}). However, we do not rely on them. Our algorithm detects connected components itself.

We study the problem in the \emph{dynamic setting under edge updates}, that is, at every time step, an edge can be either inserted or deleted.

\begin{remark}
    We allow parallel edges. 
All our running times that include the number of edges count their multiplicities. For example, in a graph with two nodes and 3 parallel edges between them, we have that $m=3$.
\end{remark}

\paragraph{Definitions and Notations.}

For any sets $A,B \subseteq V$, with $A\inter B = \varnothing$, we use $E(A,B)$ to denote the set of edges going from $A$ to $B$ and $w(A,B)$ to denote  $\card{E(A,B)}$.

For any subset $C \subseteq V$, $G/C$ is the graph where the set $C$ is contracted to a single node, adding parallel edges if necessary. No self-loops are allowed.

Let $H=(V_H, E_H)$ with $V_H \subseteq V$ and $E_H \subseteq E$, be a subgraph of $G$.
We define  the boundary   $\boundary_H T$ of a set $T \subseteq V_H$ to be the boundary of $T$ in $H$. If $E_H$ is not defined explicitly, by convention we mean $H$ is the subgraph of $G$ induced by $V_H$. We also identify $V_H$ and $H$.
     
A \emph{cluster} is a set of vertices that are ``highly connected''. It is defined in more detail below.

\begin{definition}[Crossing and uncrossing]
 We say that a cut $S$ \emph{crosses} cluster $C$ if both $S\inter C$ and $C \setminus S$ are nonempty.
A cut $S'$ \emph{uncrosses a cut $S$ in a cluster $C$} if $S$ crosses $C$ and $S'$ either equals $S\union C$ or $S \setminus C$.
\end{definition}

Our algorithm only uncrosses a cut $S$ in a cluster $C$ when $S\inter C$ is not $(1-\delta)$-boundary sparse (see Definition~\ref{def:sparse}) in $C$, which will imply that $S'$ is a $(1+O(\delta))$-approximation of $S$.

\begin{definition}
For any subset $U \subseteq V$, $G[U]$ is the subgraph induced by $U$, and $G[U]^r$ for $r \in \R$ is the subgraph induced by $U$ where for every 
boundary edge $\{v, x\}, v\in U, x \notin U$ we add $\ceil{r}$ self-loops to $v$.
\end{definition}
\begin{definition}
    The conductance of a cut $U \subsetneq V$ in the graph $G$ is $$\phi_G(U) = \frac{w(U, V\setminus U)}{\min\{\vol_G(U), \vol_G(V\setminus U)\}}.$$

A graph $G$ is a $\phi$-expander if every cut in $G$ has conductance at least $\phi$.
\end{definition}
\begin{definition}[$(\alpha, \phi)$-Expander, Definition 4.1 of~\cite{expanderhierarchy}]
    For a graph $G = (V, E)$ and parameters $\alpha, \phi \in (0,1)$ and $H \in \R$,
a subgraph $U\subset V$ is $(\alpha, \phi)$-boundary-linked expander with slack $H$ in $ G$ if the graph $ G[U ]^{\frac \alpha \phi}$ is a $\phi/H$-expander.
\end{definition}
We will maintain a partition of the graph into vertex-disjoint expanders, called \emph{expander decomposition}.
An \emph{inter-expander} edge is an edge whose endpoints belong two different expanders. We will further maintain a \emph{pre-cluster decomposition} and a \emph{cluster decomposition} that is a refinement of the expander decomposition. In the pre-cluster decomposition, each \emph{cluster} is a set of vertices and any edge  that lies between two different clusters is an \emph{inter-cluster} edge.
A cluster decomposition is a refinement of the pre-cluster decomposition, and any edge that lies between two clusters of the cluster decomposition but inside a cluster of the pre-cluster decomposition is a ``fragmented'' edge.

To compute and maintain the expander decomposition under edge insertions and deletions we use the following result from Goranci, Räcke, Saranurak and Tan~\cite{expanderhierarchy}.
In this theorem, the term  \emph{contracted graph} refers to the graph obtained from the original graph by merging all nodes in an expander into one (for each expander).
\emph{Recourse} is the number of changes made to the contracted graph, which correspond to the number of edges that switch from being 
intra-cluster edges 
to inter-cluster edges or vice-versa.

\begin{theorem}[Dynamic Expander Decomposition, Lemma 7.3 of~\cite{expanderhierarchy}]\label{thm:expanderdecomposition}
Let $\phi' = 2^{-\Theta (\log^{3/4}n)}$. Suppose a graph $G$ initially contains $m$ edges and undergoes a
sequence of at most $O(\frac {m\phi}{\rho})$ adaptive updates such that $V (G) \le n$ and $E(G) \le m$ always hold.
Then there exists an algorithm that maintains an $(\alpha, \phi')$-expander decomposition with slack $38^h = 2^{-O(\log^{1/2}n)}$ and its contracted graph with the following properties:
\begin{enumerate}[noitemsep]
 \item update time: $\Tilde O (\frac {\psi 38^{2h}}{\phi ^2})$
\item preprocessing time: $\Tilde O(\frac m \phi)$
\item initial volume of the contracted graph (after preprocessing): $\Tilde O(\phi m)$
\item amortized recourse (number of updates to the contracted graph): $O(\rho)$
\end{enumerate}
\end{theorem}

Since the amortized recourse is $O(\rho)$, and the algorithm can handle $O(\frac {m\phi} \rho)$ updates, this shows that overall the number of edges that are initially or became inter-cluster edges during the given sequence of updates is $O(m\phi)$.

\paragraph{Choice of parameter values.}
Throughout the paper we assume $c>0$, $\delta = 2^{O(\log^{3/4 -c } n)}\le 0.04$. We set $\lmin=2^{\Theta (\log^{3/4-c}n)}$ and $\lmax\le 1.2 \lmin$, $H=38^h = 2^{-O(\log^{1/2}n)}$, $\phi =\frac {\phi'} {38^h}= 2^{-\Theta (\log^{3/4}n)}$, $\rho= 2^{\Theta(\log ^{1/2}n)}$ and $ \alpha=\frac 1 {\poly \log n}.$

\section{Technical Overview}
\label{sec:technical overview}

\subsection{Overview of \cite{DBLP:conf/soda/El-HayekH025}}

As our result builds up on the work of El-Hayek, Henzinger and Li~\cite{DBLP:conf/soda/El-HayekH025}, we start by giving an overview of their techniques first.
The core of their algorithm is a fully-dynamic subroutine that takes as input two parameters, $\lmin$ and $\lmax$, and is able to maintain the value of the minimum cut if that value falls in the range $[\lmin, \lmax]$. If it falls outside of this range, then it simply outputs ``$\lambda<\lmin$'' or ``$\lambda > \lmax$'' accordingly.

Let $G$ be a graph. The starting point of the algorithm is the expander decomposition of Goranci, Räcke, Saranurak and Tan~\cite{expanderhierarchy}.
Then, the key point is to realize that any cut $S$ of cut-size at most $\lmax$ in $G$  leads to a cut of cut-size at most $\lmax$ in the graph induced by any of the expanders in the expander decomposition of $G$. 
Recall that for any $\phi$-expander $C$, we have $\frac {w(C\inter S, C\setminus S)}{\min \{\vol(C\inter S), \vol(C\setminus S)\}} \ge \phi$, which implies an upper bound of $\lmax/\phi$ on the volume of the smaller side of the cut in the expander. 
Finding these cuts of small volume can be efficiently done using the algorithm from Nalam and Saranurak~\cite{DBLP:journals/jacm/HolmLT01}, the so-called \emph{LocalKCut algorithm}.
Thus, refining the expander decomposition by repeatedly cutting expanders along cuts of cut-size at most $\lmax$ and of volume at most $\frac \lmax \phi$ ensures that  if the minimum cut has value at most $\lmax$, it does not cut any cluster obtained from that refined decomposition.

The problem with this approach, however, is that cutting along all such cuts can lead to a decomposition that is ``too fine-grained'': In the extreme case, each vertex forms its own cluster, meaning that the contracted graph is identical to the original graph.
This is problematic, as the algorithm then contracts each cluster in the decomposition and calls itself recursively on the contracted graph, that, for efficiency reasons, needs to be ``sufficiently smaller'' than the original graph.
To guarantee that the contracted graph is indeed ``small enough'', the idea of~\cite{DBLP:conf/soda/El-HayekH025}  is to filter the cuts returned by the LocalKCut algorithm and
to only cut along cuts that are $(1-\delta)$-boundary sparse, a concept introduced by Henzinger, Li, Rao, and Wang~\cite{DBLP:conf/soda/HenzingerLRW24}:

\begin{definition}[Boundary-Sparse. Def. 2.3 of~\cite{DBLP:conf/soda/HenzingerLRW24}]\label{def:sparse}
    For a set $C \subsetneq V$ and parameter $\delta \leq 1$, a cut $U \subsetneq C$ is $(1-\delta)$-boundary sparse in $C$ iff 
    $$
    w(U, C\setminus U)< (1-\delta) \min\{w(U, V\setminus C), w(C\setminus U, V\setminus C)\}
    $$
\end{definition}

Intuitively, the cut $U\subsetneq C$ is boundary sparse\footnote{Boundary sparseness indicates \emph{sparseness in comparison to the boundary}.} in $C$ if the number of edges connecting it to the rest of $C$ is a $(1-\delta)$ factor smaller than the number of edges that connect it to $V\setminus C$ and, for symmetry, $C\setminus U$  also fulfills this property. 
Thus, if the minimum cut $S$ is such that $S\inter C = U$, one can ``uncross'' $S$ into $S\setminus C$ or $S\union C$, whichever results in a cut of smaller cut-size, and the resulting near-minimum cut will only gain a $(1+O(\delta))$ factor in cut-size.
Therefore, from the expander decomposition, if one decomposes the expanders further along cuts that have volume at most $\frac \lmax \phi$, cut-size at most $\lmax$ and are $(1-\delta)$-boundary sparse, then in most cases the minimum cut can be approximated by a cut that does not cross any cluster. 
The only case where this does not happen is when the minimum cut is contained in at most two expanders. 
In that case, one can show that this cut can be retrieved by looking at the minimum cut of small volume in the so-called \emph{mirror clusters}, where a mirror cluster of a cluster $C$ is obtained from $G$  by contracting all nodes outside of $C$ into one node.

Therefore, the algorithm maintains the following invariant for its cluster decomposition: at no point in time does there exist a cut in a cluster that (1) has  volume at most $\frac \lmax \phi$ in the cluster and (2) cut-size at most $\lmax$ and (3) is $(1-\delta)$-boundary sparse.
It also maintains the mirror clusters and the cuts of small volume in it.

\subsection{Our techniques}
To prove \Cref{thm:boundeddyn} we (1) derandomize the LocalKCut algorithm and (2) improve the approximation ratio of the algorithm in~\cite{DBLP:conf/soda/El-HayekH025} which we described above.

Throughout this paper, $\lmax$ is set to be at most $2^{O(\log^{3/4-c} n)}$ with $c>0$. As discussed below, this is because the recourse of our algorithm is dependent on $\lmax$. Over $r=\sqrt[4]{\log n}$ many levels of recursion, this is at most $\lmax^{O(\sqrt[4]{\log n})}$. However to achieve an $n^{o(1)}$ update time, this value  has to be $n^{o(1)}$, which in turn requires that $\lmax$ is at most $2^{O(\log^{3/4-c} n)}$ with $c>0$.

\subsubsection{A deterministic LocalKCut algorithm}
The input of the LocalKCut  algorithm is a vertex $v$, a volume $\nu$ and a cut-size $\lmax$, and a cluster $C$ and its goal is to return all cuts in $C$ that contain $v$, are of volume (in $G[C]$) at most $\nu$, are of cut-size at most $\lmax$, and are connected.
The LocalKCut algorithm by Nalam and Saranurak~\cite{kragerlocal} heavily relies on randomization, as at its core, the algorithm starts at $v$, and randomly explores the graph checking at each step if the set of all vertices visited until that point form a cut that satisfies the desired conditions or not. 
They show that with a ``high enough'' probability, any cut that satisfies the desired conditions is found in this way, and thus, running this algorithm enough times guarantees that all such cuts are found whp.

We, thus, have to develop a whole new technique to derandomize it. For that, we rely on the \emph{greedy tree packing technique} developed by Karger~\cite{karger2000minimum}.
The tree packing technique works as follows: start with the graph itself, assign to each edge a weight of $0$, and create an empty set of trees, called the \emph{tree packing}.
Then iteratively find a minimum spanning tree of the graph, add it to the tree packing, and increase the weight of every edge in that tree by 1. 

Karger showed that if we choose enough trees in the tree-packing, then for every minimum cut $S$, there must be a tree that \emph{$2$-respects} $S$, that is, the edges of the minimum cut intersect the edges of the tree at most twice. We show an immediate generalization of this result: for any $\beta$-approximation of the minimum cut, there exists a tree that $2\beta$-respects the cut for a suitable integer constant $\beta$.
Assume for the moment that any
cut in $G$ of cut-size at most $\lmax$ is a $\beta$-approximation of the minimum cut.

We discuss below how to guarantee this condition for the graphs on which LocalKCut is executed.

Now we implement a deterministic version of LocalKCut inspired by techniques of Lokshtanov, Saurabh, and Surianarayanan~\cite{lokshtanov2022parameterized} for approximate minimum $k$-cut as follows:
A \emph{red-blue edge coloring} is a coloring of the edges where every edge is either red or blue, but not both.
We assume for the moment that we are given a family $\mathcal{F}_1$ of red-blue edge colorings and also a family $\mathcal{F}_2$ of green-yellow edge colorings and discuss later how we find and maintain such colorings efficiently. The first family $\mathcal{F}_1$ consists of a collection of red-blue edge colorings  such that for any two disjoint edge sets $F_r,F_b$ with $|F_r|\le 2\beta$ and $|F_b|\le\nu$, there exists a coloring in $\mathcal{F}_1$ with all edges in $F_r$ red and all edges in $F_b$ blue. The second family $\mathcal{F}_2$ consists of a collection of green-yellow edge colorings such that for any two disjoint edge sets $F_g,F_y$ with $|F_g|\le 2\beta-1$ and $|F_y|\le\lambda_{\max}$, there is a coloring  in $\mathcal{F}_2$ with all edges in $F_g$ green and all edges in $F_y$ yellow. We use the term \emph{pair of colorings} to denote a pair of colorings consisting of one coloring of $\mathcal{F}_1$ and one of $\mathcal{F}_2$.

The algorithm maintains during the sequence of edge updates a greedy tree packing of $G$, and the two edge coloring families $\F__1$ and $\F__2$.

Then, whenever LocalKCut is started on a vertex $v$ it performs, for every pair of colorings and every tree $T$, a breadth-first search {\em using only the blue edges that belong to $T$ and the green edges that do not belong to $T$.} Red edges and yellow edges are ignored. If the BFS visits nodes of total volume more than $\nu$, LocalKCut stops the current BFS, and moves to the next pair of colorings. If the BFS cannot visit any new vertices and is still below the volume threshold of $\nu$, LocalKCut checks whether the resulting cut induced by the visited vertices fulfills the desired conditions on cut-size and volume (i.e.~has edge capacity at most $\lmax$ and volume at most $\nu$) and if so, it adds the cut to the collection of cuts to be returned. It then moves on to the next pairs of colorings. 

We claim that this deterministic version of LocalKCut finds every cut $S$ that contains $v$, is connected, has volume at most $\nu$, and cut-size at most $\lmax$ using the following argument and setting $\beta \ge 8$, as long as $S$ is a $\beta$-approximation of the minimum cut in $C$.

As, by assumption, $S$ is $\beta$-approximation of the minimum cut in $C$, we are guaranteed that there exists a tree $T$ in our tree packing that $2\beta$-respects $S$, i.e., at most $2\beta$ of the edges of $T$ cross the cut induced by $S$. Choose $F_r$ to be these cut edges and choose $F_b$ to be all the edges of $T$ with both endpoints in $S$. By the conditions that $S$ fulfills, it follows that $|F_r| \le 2 \beta$ and $|F_b| \le \nu$. Hence there exists a blue-red coloring in $\mathcal{F}_1$ where all edges in $F_r$ are red, and all edges in $F_b$ are blue. 
As the edges in $F_b$ are edges of $T$ (that spans the cluster $C$) and they are connected by at most $2\beta$ (red) tree edges to the rest of $T$, they can form   at most $2\beta$ distinct connected components. As $S$ is connected, these components can be connected by at most $2\beta -1$ edges of $G$. Let $F_g$ be a minimum-cardinality set of edges that connect these components.  Then $|F_g| \le 2\beta-1$ and all these edges have both endpoints in $S$. Furthermore, let $F_y= E(S,V\setminus S).$ Note that $|F_y| \le \lmax$. 
Therefore, there exists a green-yellow coloring in $\mathcal{F}_2$ where all the edges in $F_y$ are yellow and all edges in $F_g$ are green.
As a result it holds that $F_b \cup F_g$ spans $S$ and all tree edges leaving $S$ are red and and all non-tree edges leaving $S$ are yellow. 
Recall that the algorithm explores the graph in BFS fashion using only the blue tree edges and the green non-tree edges for all trees in our tree packing and all pairs of edge colorings in $\mathcal{F}_1 \cross \mathcal{F}_2$. Thus, when our algorithm described above runs with that particular tree and these particular colorings it is guaranteed to find $S$ as $F_b \cup F_g$ spans $S$ and none of the tree edges leaving $S$ are blue and none of the non-tree edges leaving $S$ are green. 

In the above argument we used the fact that every cut in $G$ of cut-size at most $\lmax$ is a $\beta$-approximation of the minimum cut. 
Note that our algorithm in~\cite{DBLP:conf/soda/El-HayekH025} only runs LocalKCut on either clusters or mirror clusters, and not on arbitrary graphs $G$.
Therefore, if we can ensure that the minimum cut-size is at least $\frac \lmax \beta$ in every cluster and mirror cluster, and then run our deterministic LocalKCut algorithm on every cluster and mirror cluster, we are guaranteed to find all cuts that we need.

To ensure this property we note that if a cut in a cluster has cut-size smaller than $\frac \lmax \beta$ with $\beta > 3$, and if the graph has minimum cut at least $\lmin$, then the cut must be $(1-\delta)$-boundary sparse in the cluster. 
Note that the correctness proof for the cluster decomposition~\cite{DBLP:conf/soda/El-HayekH025} (and we will we will also base our correctness proof on it) does not place any conditions on cutting clusters other than requiring that the cuts are $(1-\delta)$-boundary-sparse. Thus, our algorithms is allowed to cut the cluster along any cut of cut-size smaller than $\frac \lmax \beta$, and do so repeatedly until no such cut remains.
After recursively cutting all these small cuts the cluster has the desired property and we can run our deterministic LocalKCut routine on it.

The question remains on how to efficiently (1) find such cuts and (2) how many colorings are needed and how to find and maintain them dynamically.

(1) We need determine new cuts that fulfill the requirements of our LocalKCut algorithm after each edge update and if an update splits a custer, we need  to run our LocalKCut algorithm recursively on the new clusters. To do so we maintain a LocalKCut data structure that allows to quickly output all desired cuts after an edge update.
However, after an update  a cluster $C$ might be split into two clusters that are not connected with each other in $C$, %
rendering the LocalKCut data structure for this cluster useless since it only works for $\beta$-approximations of the minimum cut in that cluster, which is now $0$. We could create a separate LocalKCut data structure for each cluster, creating new LocalKCut data structures whenever a cluster is split into two, but this proves to be too slow.

Instead, our solution is to not run and maintain one LocalKCut data structure per cluster, but instead build one LocalKCut data structure \emph{for all clusters}, i.e., for the whole graph. Thus, to split a cluster it suffices to simply delete from the data structure the edges that connect the two new clusters.
This also works fine with our constructions in general, although one needs to be careful.
For example, the greedy tree packing for our LocalKCut data structure now needs to be replaced by a greedy \emph{forest} packing, which is a greedy tree packing on each of the connected components.
The resulting LocalKCut data structure maintains a subgraph $G'$ of $G$ and implements the following operations:

\begin{definition}
    Let $G$ be a graph that is updated by a sequence of edge insertions and deletions. A LocalKCut data structure parametrized by $ \nu \in \N, s\in \N, \beta \in \N$ maintains a subgraph $G'$ of $G$ and implements the following operations:
    
(i) LocalKCut query($v$): Return all connected cuts in $G'$ containing $v$ of volume (in $G'$) at most $\nu$ and cut-size in $G'$ at most $s$, that are a $\beta$-approximation of the minimum proper cut of $G'$.

(ii) DeleteEdge($e$): Remove the edge $e$ from $G'$.

(iii) InsertEdge($e$): Add the edge $e$ to $G'$.
\end{definition}

For our LocalKCut data structure the time for each LocalKCut query is $\tilde O(\beta^2 (\lmax\nu  )^ {O(\beta)}) $. The amortized time for every DeleteEdge and InsertEdge operations is $\tilde{O}\left(\beta^4( \lambda_{\max}\nu)^{O(\beta)}\right)$. Using $\beta =O(1)$ and $\lambda, \nu = n^{o(1)}$ we get that every operation takes amortized time $n^{o(1)}$.

(2) Let us now discuss how to find the red-blue and yellow-green colorings. 
For that, we use the following lemma which was shown by
Chitnis, Cygan, Hajiaghayi, Pilipczuk and Pilipczuk~\cite{chitnis2016designing} and that guarantees the existence of a family of colorings of a set of elements and gives a fast algorithm for finding them. In our case each element is an edge.
\begin{restatable}[Lemma 1.1 of~\cite{chitnis2016designing}]{lemma}{howmanycolorings}\label{lem:colorings}
Given a set $U$ of cardinality $n$, and integers $0 \le a,\, b \le n$, one can in time $2^{
O(\min(a,b) \log(a+b))}n \log n$ construct
a family $\mathcal F$ of at most $2^{
O(\min(a,b) \log(a+b))} \log n$ subsets of $U$, such that the following holds: for any sets $A, B \subseteq U$,
$A \cap B = \emptyset$, $|A| \le a,\, |B| \le b$, there exists a set $S \in\mathcal F$ with $A \subseteq S$ and $B \cap S = \emptyset$.
\end{restatable}

In our case, $U$ is the set of edges, and $a$ and $b$ are respectively $2\beta$ and $\nu$ for the red-blue coloring, and $2\beta-1$ and $\lmax$ for the green-yellow coloring. 
Plugging in $\lmax, \nu = n^{o(1)}$ and $\beta = O(1)$ we have that at most $n^{o(1)}$ many of these colorings are necessary to ensure the properties that we need.

Another feature of this construction is that it is remarkably robust against updates: indeed, if an edge is deleted, the colorings do not have to change, and the properties we require still hold. 
If another edge is then inserted, then it can take over the colors of the deleted edge, and the properties still hold, as it is a set-theoretic result, and not a graph-theoretic one. 
Therefore, when computing colorings, one can do so for $2m$ many edges, and assign only the first $m$ colors of a coloring to the existing edges, sparing the $m$ other ones for potential edge insertions. 
These colorings will thus be valid until the number of edges in the graph doubles. When this happens we can afford to compute new family of colorings. 
Overall, we spend $m^{1+o(1)}$ time to compute these colorings, that we can amortize over at least $m$ updates, for an update time of at most $n^{o(1)}$.

\subsubsection{Getting a better approximation ratio}

In~\cite{DBLP:conf/soda/El-HayekH025}, the approximation ratio is $\left(1+\Theta(\frac 1 {\sqrt{\log n}})\right)$. 
As a consequence, if we want to use this algorithm to compute an exact answer, the minimum cut can only be as large as $O(\sqrt{\log n})$, as otherwise the approximation ratio is not small enough to guarantee that the answer is exact.

The reason for such a weak approximation ratio is that cutting along  $(1-\delta)$-boundary sparse cuts is done in an arbitrary order.
As a result, we only achieve an upper bound on the total number of inter-cluster edges created between any two rebuilds of the expander decomposition of roughly $M2^{O(\frac 1 \delta)}$, where $M$ is the number of inter-expander edges in the expander decomposition, which is typically $m\phi = m 2^{-\Theta(\log^{3/4} n)}$.
We need $\phi$ to dominate $2^{O(\frac 1 \delta)}$ to ensure that the total number of inter-cluster edges can still be bounded by  %
$m 2^{-\Theta(\log^{3/4} n)}$
and, thus, in~\cite{DBLP:conf/soda/El-HayekH025} $\delta$ cannot be truly smaller than $\Omega(\frac 1 { {\log^{3/4} n}})$. 
In this work we aim for $\delta= 2^{-O(\log ^{3/4-c}n)}$ with $c>0$, and thus need a cluster decomposition with at most $O(\frac  M {\poly (\delta)})$ many inter-cluster edges. 
Since $M=O(m\phi)$, the $\phi$ factor dominates the $\frac 1 {\poly(\delta)}$ factor, as $\phi = 2^{-\Theta(\log^{3/4} n)}$ and $\delta= 2^{-O(\log ^{3/4-c}n)}$.
We are, thus, able to show an approximation ratio of $1 + 2^{-O(\log ^{3/4-c}n)}$, which implies that our new algorithm gives an exact solution for $\lmax=2^{O(\log^{3/4-c} n)}$ with $c>0$, improving on the result of~\cite{DBLP:conf/soda/El-HayekH025}.

We now explain how to achieve a better approximation ratio.
Henzinger, Li, Rao, and Wang, who presented the first near-linear time deterministic \emph{static} exact minimum cut algorithm in weighted graphs~\cite{DBLP:conf/soda/HenzingerLRW24}, also cut clusters along $(1-\delta)$-boundary sparse cut. They are able to get the desired $O(\frac M {\poly (\delta)})$ upper bound on the number of inter-cluster edges by performing the cuts in a different manner than~\cite{DBLP:conf/soda/El-HayekH025}: If the cluster has high boundary-size (say, at least $3\lmax$), they cut the cluster along any boundary sparse cuts;
if, however, the cluster has small boundary-size, then they call a subroutine, which we name \emph{fragmenting}, which carefully selects the order in which the $(1-\delta)$-boundary sparse cuts are performed, and outputs which edges need to be cut to get the decomposition -- that is, which edges need to become inter-cluster edges.
Fragmenting runs in time polynomial in the volume of the cluster that is decomposed, but this is not a problem for them as the clusters they consider have only volume polylogarithmic in $n$.

We cannot run this algorithm in a black-box manner for the following reasons: (a) The cluster we want to fragment might have large volume as they might be the expanders of the expander decomposition, and (b) their algorithm is static and not dynamic.
Thus, we instead adapt their fragmenting algorithm  to our needs, solving both problems simultaneously: we give a static deterministic algorithm that runs in $n^{o(1)}$ time on clusters of boundary-size $O(\lmax)=n^{o(1)}$ and has output of size $\tilde O(\frac \lmax{\poly (\delta)})=n^{o(1)}$.
This is fast enough so that we can simply rerun it after every update on each affected cluster. Running the algorithm from scratch after every update makes it dynamic: Since its output is of size $\tilde O(\frac \lmax {\poly (\delta)})$, which corresponds to the cut-size of the cuts along which the cluster has to be partitioned, the recourse is $\tilde O(\frac \lmax {\poly (\delta)})$. 
Its static running time of $n^{o(1)}$ translates then in a  $n^{o(1)}$ time per update on clusters of boundary-size at most $O(\lmax)$.

To speed up the algorithm we rely on two observations: the first one is that, for clusters $C$ of small boundary-size, that is, $\boundary C = O(\lmax)$, the running time of the fragmenting algorithm is dominated by the time needed to find the set of all $(1-\delta)$-boundary sparse cuts in the cluster.
We can replace this step by running our version of LocalKCut instead, only looking at cuts of small volume.
Since any $(1-\delta)$-boundary sparse cut must have at least one boundary edge, it thus suffices to run LocalKCut from every node incident to a boundary edge. This takes $n^{o(1)}$ time per LocalKCut query. As the number of queries is $O(\lmax)=n^{o(1)}$ exploring only volume $\nu = n^{o(1)}$ in $C$, this leads to a $\lmax n^{o(1)}=n^{o(1)}$ overall running time for the fragmenting algorithm.

The second observation is that the number of edges that our fragmenting algorithm cuts, that is, the number of inter-cluster edges that it creates, is $\tilde O(\frac {\boundary C} {\delta ^2})$, as we can show that it creates at most $\tilde O(\frac 1 {\delta^2})$ many clusters of cut-size at most $\boundary C=O(\lmax)$ each. This number increases the recourse of our cluster decomposition.
As long as $\delta = 2^{-O(\log ^{3/4-c}n)}$ and $\lmax = 2^{O(\log ^{3/4-c}n)}$, the additive increase in the recourse is still dominated by $\frac 1 \phi$, and thus we can afford such a recourse increase for the cluster decomposition.
We can therefore run our fragmenting algorithm
from scratch on every cluster that is affected by an update of the dynamic expander decomposition as both the running time and the recourse stay sufficiently bounded. 

 Overall, to maintain our cluster decomposition, we run three dynamic algorithms back-to-back: the dynamic expander decomposition, the pre-cluster decomposition that cuts along $(1-\delta)$-boundary-sparse cuts arbitrarily as long as the clusters have boundary-size at least $\Omega(\lmax)$, and the fragmenting algorithm that only runs on clusters of boundary-size $O(\lmax)$ and that chooses which  $(1-\delta)$-boundary-sparse cuts to cut and in which order.
The expander decomposition has recourse $\rho = 2^{O(\sqrt{\log n})}$, the pre-cluster decomposition has recourse $O(\frac 1 \delta)$ for each update caused by the expander decomposition, while the fragmenting algorithm has recourse $\tilde O(\frac \lmax {\delta^2})$ for each update caused by any of the prior two decomposition algorithms. 
Overall we thus get a recourse which is a product of these three factors, i.e., of $\tilde O(\frac {\rho \lmax} {\delta^3}) = 2^{O(\log^{3/4-c} n)}$ per level of recursion. As we show that the recursive depth is at most $O(\sqrt[4]{\log n})$, this yields a total recourse of $2^{O(\log^{1-c} n)}$, which we can process in $n^{o(1)}$ time per change, thus $n^{o(1)}$ overall time. 

Hence, we get an algorithm that in $n^{o(1)}$ time per update maintains a cluster hierarchy, that is, a collection of graphs associated to a cluster decomposition each, where each graph is obained from the previous graph by contracting each cluster. As in~\cite{DBLP:conf/soda/El-HayekH025}, we show that any minimum proper cut is well-approximated by a cut of small volume in a graph of this hierarchy. We thus maintain additionally, for each vertex $v$, the best proper cut of small volume in which $v$ is contained if there exists one small enough, store them in a heap data structure, and simply return the cut of smallest cut-size among all these cuts to obtain our result and return an approximate minimum cut.

\section{A deterministic LocalKCut data structure}
\label{sec:localkcut}

In this section, we present a deterministic LocalKCut data structure that achieves the following bounds:

\begin{theorem}\label{thm:deterministic-dynamic}
For any parameters $\lambda_{\max},\nu, \beta\in \N_+$ with $\lmax < \nu$, there is a deterministic, fully dynamic LocalKCut data structure that outputs, on each LocalKCut query specified by a vertex $v\in V$, a collection $\mathcal S$ of cuts $S\subseteq V$ containing $v$ such that
 \begin{enumerate}
 \item Each cut $S\subseteq V$ satisfies $\partial S\le\lambda_{\max}$, $\textbf{\textup{vol}}(S)\le\nu$, and $G[S]$ is connected.
 \item $\mathcal S$ contains all $\beta$-approximate mincuts $S\subseteq V$ satisfying condition~(1).
 \end{enumerate}
The data structure runs in $\tilde O( m \beta^2 (\lambda_{\max}\nu)^{O(\beta)})$ preprocessing time,  $\tilde{O}\left(\beta^4( \lambda_{\max}\nu)^{O(\beta)}\right)$ amortized update time and $\tilde O(\beta^2 (\lmax\nu  )^ {O(\beta)}) $ query time.
\end{theorem}

The goal is to set $\beta=O(1)$, $\lmax =n^{o(1)}$ and $\nu = n^{o(1)}$ so that the preprocessing time is $m^{1+o(1)}$ and the update and query time are $n^{o(1)}$.

\subsection{Description of the algorithm and proof of correctness}
The full description of the algorithm can be found in \Cref{alg:detlocalkcut}, but we present it step by step here. We begin with some preliminaries from~\cite{thorup2000dynamic}.

A \emph{tree packing} is an assignment of weights to spanning trees of $G$ so that each edge $e$ gets \emph{load}
\[ \ell(e)=\sum_{T\ni e}w(T)\le1 .\]
The value of the tree packing is $\sum_Tw(T)$. We let $\tau$ denote the value of the maximum tree packing of $G$. The classic Nash-Williams' theorem relates $\tau$ to the mincut $\lambda$ of the graph.
\begin{theorem}[Nash-Williams~\cite{nash1961edge}]
$\lambda/2\le\tau\le\lambda$.
\end{theorem}

A tree packing is \emph{$(1+\epsilon)$-approximate} if $\sum_Tw(T)\ge\tau/(1+\epsilon)$. In order to use the same parameters as prior work, we slightly abuse the notation in this section: The $\epsilon$ in this section is independent from the $\epsilon$ from all other sections. \
\begin{theorem}[Generalisation of \cite{karger2000minimum}]\label{thm:respects}
For any $\beta$-approximate mincut and any $(1+\epsilon)$-approximate tree packing, there is a tree in the packing that $\lfloor2(1+\epsilon)\beta\rfloor$-respects the $\beta$-approximate mincut.
\end{theorem}

\begin{proof}
Since the tree packing $\mathcal T$ is $(1+\epsilon)$-approximate, we have $\sum_Tw(T)\ge\tau/(1+\epsilon)$.
If we sample a random tree $T\in\mathcal T$ with probability proportional to $w(T)$, then the probability that the sampled tree contains a given edge $e$ is
\[ \frac{\sum_{T\ni e}w(T)}{\sum_Tw(T)} \le \frac1{\sum_Tw(T)}\le\frac1{\tau/(1+\epsilon)}. \]
For a given $\beta$-approximate mincut, the expected number of edges in the sampled tree is at most
\[ \beta\lambda\cdot\frac1{\tau/(1+\epsilon)}=(1+\epsilon)\beta\lambda/\tau ,\]
which is at most $2(1+\epsilon)\beta$ by Nash-Williams' theorem. It follows that there exists a tree $T\in\mathcal T$ that $\lfloor2(1+\epsilon)\beta\rfloor$-respects the $\beta$-approximate mincut.
\end{proof}

Our goal now is to have good approximate tree packings. An easy way to get an approximate tree packing is by looking at  \emph{Greedy Tree Packings}, popularized by Karger and Thorup~\cite{thorup2000dynamic, thorup2007fully}. 
Let $H$ be a connected unweighted graph.
A greedy tree packing with $k$ trees is a packing obtained using the following procedure: start with $H$ with weights $0$ on all edges. Take a minimum spanning tree on $H$, add it to the packing, then increase the weights of all edges included in that tree. Repeat the instructions in the former sentence until $k$ trees are in the packing.

\begin{theorem}[Lemma 5 of~\cite{thorup2007fully}]
    A greedy tree packing with at least $6\lambda \log m /\epsilon^2$ trees is $(1+\epsilon)$-approximate.
\end{theorem}

We modify this slightly and consider a \emph{Greedy Forest Packing}, where we take the exact same procedure as for greedy tree packings, but the graph $G$ does not have to be connected, and the minimum spanning tree is replaced by a minimum spanning forest, that is, a forest where each tree is a minimum spanning tree of its connected component in $G$.

The idea here is that a greedy forest packing is essentially a greedy tree packing on each of the connected components of $G$, which has the nice properties discussed above that we will harness.

We first discuss how to maintain a greedy forest packing dynamically:

\begin{theorem}[Theorem 8 of \cite{holm2001poly}]\label{thm:dynforest}
    There is a fully dynamic deterministic Minimum Spanning Forest algorithm that for a graph with $n$ vertices, starting with no edges, maintains a minimum spanning forest in $O(\log^4 n)$ amortized time per edge insertion or deletion.
\end{theorem}

As discussed in Section 3.2 of~\cite{thorup2000dynamic}, by adding priorities to the edges, one can ensure that there is a unique minimum spanning forest in each graph.
This leads to the following lemma, which we prove similarly to~\cite{thorup2000dynamic}:

\begin{lemma}
    Each insertion or deletion of an edge changes at most $i$ edges in the $i$-th forest.
\end{lemma}

\begin{proof}
    We will consider only edge deletions, as the edge insertions are treated symmetrically.

    Let $(u,v)$ be the deleted edge.
    \begin{itemize}
        \item If the edge deletion disconnects a component, then the edge load drops to $0$.
        On each connected component of $G$, the tree packing does not change, as otherwise the forest packing before the deletion was not greedy.
        We have thus only $1$ edge change on the $i$-th forest, for each $i$.
        \item If the edge deletion does not disconnect a component:  For any edge $e$, we define $\ell_i(e)$ to be the number of forests with index smaller or equal to $i$ that contain $e$. This is the \emph{load} of $e$ after forst $i$. 
    We distinguish $\ell_i^{old}(e)$ from $\ell_i^{new}(e)$, the former being the load of $e$ before the edge deletion, and the latter being the load of $e$ after the edge deletion.
    
    Then we prove by induction that, for every $j \ge 1$:
        $$
        \ell_j^{old} (e) \le \ell_j^{new(e)} \qquad \forall e \in E
        $$
        that is, the edge deletion only increases the loads for all $j$.

    This, coupled to the fact that for each $j$, the total load of the edges apart from $(u,v)$ is equal to $j\cdot(n-\text{number of components of } G)$, proves that at most $j$ edges $e$ satisfy $\ell_j^{old}(e)<\ell_j^{new}(e)$.
    \begin{itemize}
        \item Base case: for $j=1$, if an edge is deleted within a component, this edge is replaced by another one in the acyclic matroid.
        \item Induction step: Assume we have that $        \ell_j^{old} (e) \le \ell_j^{new(e)} \quad \forall e \in E $, and we want to show that $\ell_{j+1}^{old} (e) \le \ell_{j+1}^{new(e)} \quad \forall e \in E $.
        There are three cases:
        \begin{itemize}
            \item Either $f=(u,v)$ is in the $j+1$-th forest after the edge deletion.
            Then $\ell_{j+1}^{new}(f) = \ell_{j}^{new}(f)+1 \ge \ell_{j}^{old}(f)+1 \ge \ell_{j+1}^{old}(f)$,
            \item or  $f=(u,v)$ was not in the $j+1$-th forest before the edge deletion. Then $\ell_{j+1}^{old}(f) = \ell_{j}^{old}(f) \le \ell_{j}^{new}(f) \le \ell_{j+1}^{new}(f)$,
            \item or $f=(u,v)$ was in the $j+1$-th forest before the deletion and is not in the forest after the deletion.
            This means that before the edge deletions, $f$ was in the minimum spanning forest in $G$ with weights $\ell_j^{old}$, but is not anymore after the deletions with weights $\ell_j^{new}$.
            Since, by induction, none of the weights have decreased between $\ell_j^{old}$ and $\ell_j^{new}$, this means that the load on $f$ has strictly increased: $\ell_{j}^{new} \ge \ell_{j}^{old}+1$.
            Hence: $\ell_{j+1}^{new}(f) = \ell_{j}^{new}(f) \ge \ell_{j}^{old}(f)+1 = \ell_{j+1}^{old}(f)$.
        \end{itemize}
    \end{itemize}
    \end{itemize}
\end{proof}

This, in turn, ensures that we can maintain a dynamic greedy forest packing with $k$ trees in $\tilde O (k^2)$ time per update operations

\begin{theorem}[Direct Generalization of Theorem~4 of \cite{thorup2000dynamic}]\label{thm:dynamic-greedy-tree-packing}
For any parameter $\lambda_{\max}$, there is a fully dynamic algorithm that maintains a greedy forest packing of $\tilde{O}(\lambda_{\max}/\epsilon^2)$ trees in $\tilde{O}(\lambda_{\max}^2/\epsilon^4)$ amortized time per update. Whenever the graph has edge connectivity $\le\lambda_{\max}$, the tree packing is $(1+\epsilon)$-approximate. Whenever the edge connectivity of the graph is $> \lambda_{\max}$, the tree packing has no guarantees.
\end{theorem}

Our dynamic algorithm calls \Cref{thm:dynamic-greedy-tree-packing} with parameters $\lambda_{\max}$ and $\epsilon=\frac 1 {3\beta}$. Remember that we want to find all cuts containing $v$ that have volume smaller than $\nu$, cut-size smaller than $\lmax$ and are $\beta$-approximations of the minimum cut. By \Cref{thm:respects}, this cut $2\beta$-respects a tree in the forest packing. Therefore, it suffices to implement the following dynamic data structure, and run the queries on each tree currently in the packing.

\begin{lemma}\label{lem:given-tree}
For any parameters $\lambda_{\max},\nu, \beta\in \N_+$, there is a fully dynamic algorithm that outputs, on each query specified by a vertex $v\in V$ and a spanning tree $T$ of the connected component of $v$, the collection of all cuts $S\subseteq V$ satisfying $\partial S\le\lambda_{\max}$, $\textbf{\textup{vol}}(S)\le\nu$, $G[S]$ is connected, and $S$ $2\beta$-respects the tree $T$.
The algorithm runs in $\tilde{O}(\lambda_{\max}\nu)^{O(\beta)}$ time per update and query.
\end{lemma}

More precisely, the proof of Theorem~\ref{thm:deterministic-dynamic} proceeds as follows: Given a query vertex $v\in V$, our dynamic algorithm runs \Cref{lem:given-tree} on the queries $(v,T)$ for each $T$ currently in the packing, and takes the union of all sets $\mathcal S$ returned by the queries. The final running time is $\tilde{O}(\lambda_{\max})$ many calls, one per tree, each taking $\tilde{O}(\lambda_{\max}\nu)^{O(\beta)}$ time, which is $\tilde{O}(\lambda_{\max}\nu)^{O(\beta)}$ overall.

It remains to prove \Cref{lem:given-tree}. The algorithm is inspired by the fixed-parameter tractable (FPT) algorithm for approximate minimum $k$-cut~\cite{lokshtanov2022parameterized}.  The algorithm maintains two families of $2$-colorings of edges of the graph. The first family consists of red-blue edge colorings such that for any disjoint sets $F_r,F_b$ of edges with $|F_r|\le2\beta$ and $|F_b|\le\nu$, there is a coloring with all edges in $F_r$ red and all edges in $F_b$ blue. We discuss the size and maintenance of this family later on. The second family consists of green-yellow edge colorings such that for any disjoint sets $F_g,F_y$ of edges with $|F_g|\le 2\beta -1$ and $|F_y|\le\lambda_{\max}$, there is a coloring with all edges in $F_g$ green and all edges in $F_y$ yellow. 

\textit{Algorithm for answering a LocalKCut Query.}
The algorithm iterates over all pairs of colorings, one red-blue coloring and one green-yellow coloring. For each pair of colorings, the algorithm runs a depth-first search in the subgraph $H$ consisting of blue edges in the tree $T$ and green edges from the graph that are not in $T$.
If this depth-first search for this pair of colorings encounters more than $\nu$ edges, then it terminates early, spending time  $O(\nu)$.
Let $S\subseteq V$ be the set of vertices visited by the search. If $S$ satisfies  $\partial S\le\lambda_{\max}$, $\textbf{\textup{vol}}(S)\le\nu$, $G[S]$ is connected, and $S$ $2\beta$-respects the tree $T$, then add it to the final collection $\mathcal S$ then proceeds to the next pair of colorings.

\sloppy

\textit{Correctness.}
We claim that any set $S$ satisfying the above requirements is added to $\mathcal S$ for some pair of colorings. To prove this, it suffices to show that $S$ is precisely the connected component containing $v$ in the subgraph $H$ on some pair of colorings, so that the depth-first search visits exactly the vertices in $S$. 
Since $|S|\le\textbf{\textup{vol}}(S)\le\nu$, there are at most $|S|-1\le\nu-1$ many tree edges in $G[S]$, and since $S$ $2\beta$-respects $T$, there are at most $2\beta$ tree edges in $\partial S$. So there is a red-blue coloring with all tree edges in $G[S]$ colored blue, and all tree edges in $\partial S$ colored red. Moreover, since $S$ $2\beta$-respects $T$, the tree edges in $G[S]$ form at most $2\beta$ connected components.
Let  $1 \le j<2\beta$ be the number of connected components formed by the tree edges in $G[S]$. Then these connected components are connected in $G$ by exactly $j-1$ non-tree edges $\mathcal E$  because $G[S]$ is connected. There is a green-yellow coloring such that the edges $\mathcal E$ (if $\mathcal E \not= \emptyset$)  are green  and all edges in $\partial S$ are yellow (of which there are at most $\lambda_{\max}$). 
In the iteration where these two colorings are considered, by construction $S$ is the connected component containing $v$ in the graph $H$, and we are done.

It remains to bound the number of colorings. We use the lemma below.

\howmanycolorings*

At the beginning of the algorithm, and at every point where the number of edges double compared to the last recomputation, we set $U$ to be a universe of elements of size at most $2m$ (where $m$ is the number of edges at this point in time) and associate each current edge of the graph with a unique element in $U$. On each edge insertion, associate the new edge to an unassigned element in $U$. If there are no unassigned elements remaining, then recompute $U$ to be double the size. This way, the time spent recomputing is amortized among the edge insertions.

For the family of red-blue colorings, we set $a=2\beta$ and $b=\nu$ and map each set $S\in\mathcal F$ to the coloring that colors all edges with associated element in $S$ red and all other edges blue. This family has size $(2\beta)^{O(\log\nu)}\log n=\nu^{O(\beta)}\log n$. For the family of green-yellow colorings, we set $a=2\beta-1$ and $b=\lambda_{\max}$ and map each set $S\in\mathcal F$ to the coloring that colors all edges with associated element in $S$ green and all other edges yellow. This family has size $(2\beta)^{O(\log\lambda_{\max})}\log n=\lambda_{\max}^{O(\beta)}\log n$.

The pseudocode of the algorithm is given next.

\begin{algo}[Deterministic LocalKCut]\label{alg:detlocalkcut}
    Input:
    \begin{enumerate}
        \item A dynamic unweighted graph $G$
        \item A maximum cut-size $\lmax \in \N_+$
        \item A maximum volume $\nu \in \N_+$, $\nu > \lmax$
        \item An approximation ratio $\beta \in \N_+$, $\beta < \lmax$.
    \end{enumerate}
    Preprocessing and maintained variables:
    \begin{enumerate}
        \item \label{detlocalkcutpre1} Set $\epsilon = \frac 1 {3\beta}$
        \item \label{detlocalkcutpre2} Start an instance of \Cref{thm:dynamic-greedy-tree-packing} -- a greedy forest packing with $6 \lmax \log m /\epsilon ^2$ forests.
        \item \label{detlocalkcutpre3} Create red-blue colorings of $2m$ edges with $a=2\beta$ and $b=\nu$ using \Cref{lem:colorings}.
        \item \label{detlocalkcutpre4} Create yellow-green colorings of $2m$ the edges with $a=2\beta$, $b=\lmax$ using \Cref{lem:colorings}.
    \end{enumerate}
    Processing an edge insertion:
    \begin{enumerate}
        \item \label{detlocalkcutins1} Feed it to \Cref{thm:dynamic-greedy-tree-packing}, updating the greedy forest packing.
        \item \label{detlocalkcutins2} Give it an unused color for each of the colorings. If there are no unused colors, create new colorings of double the size.
    \end{enumerate}
     Processing an edge deletion:
    \begin{enumerate}
        \item Give it to \Cref{thm:dynamic-greedy-tree-packing}, updating the greedy forest packing.
        \item Free its colors in each coloring to make it available for a future edge insertion.
    \end{enumerate}
    Processing a request on node $v$:
    \begin{enumerate}
        \item $\mathcal S \leftarrow \varnothing$
        \item For each forest, red-blue and yellow-green colorings:
        \begin{enumerate}
            \item Do a DFS starting at $v$, exploring edges that are blue if they are in the forest, yellow if not in the tree, capped at $\nu$ volume visited. 
            \item If the DFS terminates before the cap, let $S$ be the set of explored vertices. If $\boundary S \le \lmax$, add $S$ to $\mathcal S$.
        \end{enumerate}
        \item return $\mathcal S$.
    \end{enumerate}
    
\end{algo}

\subsection{Running time analysis}

\begin{lemma}
    The preprocessing time of \Cref{alg:detlocalkcut} is $\tilde O( m \beta^2 (\lambda_{\max}\nu)^{O(\beta)})$.
\end{lemma}

\begin{proof}
    Step~\ref{detlocalkcutpre1} of preprocessing takes constant time. Step~\ref{detlocalkcutpre2} takes $O(m \cdot \log^4m \cdot 6\lmax \log m /\epsilon^2)$. Indeed, we initialize the forests by inserting each edge one by one in the first instance of \Cref{thm:dynforest}, then on to the second one, and so on until the last one. Each edge inserted takes $O(\log^4 n)$ time, and there are $O(m)$ many edge, each inserted into $6\lmax \log m /\epsilon^2$ many trees. Setting $\epsilon$ to $\frac 1 {3 \beta}$ we get that Step~\ref{detlocalkcutpre2} takes $\tilde O(m \lmax \beta^2)$.
    Step~\ref{detlocalkcutpre3} takes, by \Cref{lem:colorings}, $2^{O(\min(2\beta, \nu)\log (2\beta + \nu))} m\log m = \tilde O( m \nu ^ {O(\beta)}) $, where we use $\beta < \nu$.  Step~\ref{detlocalkcutpre4} takes, by \Cref{lem:colorings}, $2^{O(\min(2\beta, \lmax)\log (2\beta + \lmax))}m\log m = \tilde O( m\lmax ^ {O(\beta)}) $, where we use $\beta < \lmax$. 
\end{proof}

\begin{lemma}
    The amortized update time of \Cref{alg:detlocalkcut} on an update is $\tilde{O}\left(\beta^4( \lambda_{\max}\nu)^{O(\beta)}\right)$.
\end{lemma}

\begin{proof}
    Let us first consider an edge insertion. 
    Step~\ref{detlocalkcutins1} takes $O( \lmax^2 \beta^4)$ by \Cref{thm:dynamic-greedy-tree-packing}, where $\epsilon=\frac {1}{3\beta}$. Step~\ref{detlocalkcutins2} takes $O(1)$ time if we do not compute new colorings, but $\tilde O( m(\lmax \nu) ^ {O(\beta)}) $ if we need to compute new ones, as discussed in the previous proof. However, if it is the case, it means the number of edges have been doubled since the last time colorings were computed, and thus we can amortize this recomputation over the last $\frac m 2$ edges. This yields an amortized update time of $\tilde O((\lmax \nu) ^ {O(\beta)}) $.

    The case for edge deletions is similar (and easier as we do not need to recompute any colorings).
\end{proof}

\begin{lemma}
    The query time of \Cref{alg:detlocalkcut} is $\tilde O(\beta^2 (\lmax\nu  )^ {O(\beta)}) $.
\end{lemma}

\begin{proof}
    For every forest in the forest packing, red-blue coloring and yellow-green coloring, we run a DFS that can explore nodes of at most $O(\nu)$ volume. Hence, every DFS takes $O(\nu)$ time. There are $\tilde O(\beta^2 \lmax)$ many forests, $2^{O(\min(2\beta, \nu)\log (2\beta + \nu))} = \tilde O(\nu ^ {O(\beta)}) $ many red-blue colorings, and $2^{O(\min(2\beta, \lmax)\log (2\beta + \lmax))}= \tilde O(\lmax ^ {O(\beta)}) $ many yellow-green colorings, which means we run overall $\tilde O(\beta^2 (\lmax\nu  )^ {O(\beta)}) $ many DFS. 
\end{proof}

\section{Fragmenting a Cluster}
\label{sec:fragmenting}

In this section, we discuss how to fragment a cluster. Note that this algorithm is static. This algorithm and its analysis closely follow Lemma 6.4 of~\cite{DBLP:conf/soda/HenzingerLRW24}.
For this section, we assume that the data structure we present in \Cref{sec:datastructure} supports our cluster decomposition.

\begin{theorem}[Fragmenting, inspired by Lemma 6.4 of~\cite{DBLP:conf/soda/HenzingerLRW24}]\label{thm:fragmenting}
    Let $G$ be a graph, and $C$ a cluster in that graph. Let $\lmax, \lmin, \nu, \delta, \beta$ be parameters and assume that:
    \begin{enumerate}
        \item The minimum cut of $G[C]$ is at least $\lmin \ge \frac \lmax \beta$
        \item $\boundary C \le 6 \lmax$.
        \item an instance of the LocalKCut data structure (\Cref{thm:deterministic-dynamic}) is running on the intracluster edges of $G$ with parameters $\lmax, \nu, \beta$.
    \end{enumerate}
    Then, there exists an algorithm that decomposes $C$ into a disjoint union of clusters $C_1, \dots, C_k$ for some $k \in \N$, such that:
    \begin{enumerate}
        \item \label{fragprop1} For any set $S \subseteq C$ with $\vol(S) \le \nu$ and $\boundary_{G[C]}S \le \lmax$, there exists a partition $\mathcal A$ of $\{C_1, \dots, C_k\} $ such that for each element $A \in \mathcal A$ (which is a subset of $\{C_1, \dots, C_k\} $), the set $S \inter \left(\union_{A' \in A} A'\right)$ is not $(1-\delta)$-boundary sparse in $\union_{A' \in A} A'$.
        \item \label{fragprop2} There are at most $O(\delta^{-2}\log^{O(1)}\card C)$ many clusters, and each cluster satisfies $\boundary C_i \le \boundary C$.
    \end{enumerate}
    The algorithm runs in $\tilde O( \delta^{-2}\beta^2(\nu \lmax)^{O(\beta)})$ time.
\end{theorem}

\begin{definition}
    Let $A$ be a set, $\C$ a set of subsets of $A$ and $S\subseteq A$ a subset of $A$. Then the intersection $\C\inter S$ of $\C$ and $S$ is the set of subsets of $S$: $\C \inter S = \{T\inter S: T\in \C\}$.
\end{definition}

\begin{algo}\label{alg:fragmenting}
    \item Input:
    \begin{enumerate}
        \item a graph $G$ with a partition $\P$
        \item a cluster $C \in \P$
        \item access to a LocalKCut instance on $G\setminus \{e\in E(G): \textup{label}(e) = \textup{``intercluster''}\}$
        \item A maximum cut value $\lmax$, a maximum volume $\nu$, an approximation ratio $\beta$, with $\lmax \le \nu$.
    \end{enumerate}
    \item Algorithm:
    \begin{enumerate}
        \item \label{frag1} Run a LocalKCut query on every vertex $v$ in the set $C\inter \mathcal N(\mathcal N(C))$ of vertices in $C$ incident to boundary edges of $C$, which outputs $\C_v=\{U_1, \dots, U_\ell\}$, a collection of cuts returned for $v$.
        \item \label{frag2} Let $\C$ be the collection of cuts obtained by taking the union of the collection of cuts returned by LocalKCut: $\C=\union_{v\in \mathcal N(C)} \C_v$.
        \item \label{frag3} \textsc{Trim}$(C, \C)$
    \end{enumerate}
    \item  \textsc{Trim}$(C', \C')$
    \begin{enumerate}
        \item \label{trim0} Remove from $\C'$ all cuts that are not $(1-\delta)$-boundary sparse in $C'$ (we call this the \emph{pruning} step).
        \item \label{trim1} If there is a $D\in \C'$ that crosses $C'$ such that $\boundary_{G{[C']}}D \le 0.4\lmax$, remove $C'$ from $\P$, add $C'\setminus D$ and $D$ to $\P$ and recursively call \textsc{Trim}$(C'\setminus D, \C'\inter (C'\setminus D))$ and \textsc{Trim}$(D, \C'\inter D)$.
        \item \label{trim2} Otherwise, if $\C' \neq \varnothing$, choose the set $D\in \C'$ with minimum $\min\{ \card {D}, \card{C'\setminus D}\}$ (in other words, $D$ is the most unbalanced -- cardinality wise -- cut in $\C'$). Remove $C'$ from $\P$, add $C'\setminus D$ and $D$ to $\P$ and recursively call \textsc{Trim}$(C'\setminus D, \C'\inter (C'\setminus D))$ and \textsc{Trim}$(D, \C'\inter D)$.
    \end{enumerate}
\end{algo}

\subsection{Correctness}
\begin{proof}[Proof of \pref{fragprop1}]
    Consider a set $S \subseteq C$ with $\vol(S) \le \nu$ and $\boundary _{G[C]}S \le \lmax$. 
    
    {\em Case 1:} If $S=\varnothing$ or $S=C$, then \pref{fragprop1} is trivial because $\varnothing$ and $C$ are always non-$(1-\delta)$-boundary-sparse in $C$, and we can set $\mathcal A = \{\{C_1, \dots, C_k\}\}$.

    {\em Case 2:}
    If $w(S, V\setminus C) = 0$, then similarly, $S$ is trivially non-$(1-\delta)$-boundary-sparse in $C$, and we can set $\mathcal A =\{ \{C_1, \dots, C_k\}\}$.

    {\em Case 3:}
    Thus we are left to consider the case $S \neq \varnothing$,$S\subsetneq C$, and $w(S, V\setminus C) >0$. 
    Since $\boundary _{G[C]}S \le \lmax$ and $w(S, V\setminus C) >0$, there exists an edge $(u,v)$ with $u\in S\inter C, v \in V\setminus C$.  Querying LocalKCut on $v$, as \textsc{Trim} does in the first step, finds $S$, by \Cref{thm:deterministic-dynamic}. 
    Therefore, we have that $S \in \C$, the collection of cuts found at the root instance \textsc{Trim}$(C, \C)$. 
    Moreover, the property $S\inter C' \in \C'$ holds recursively along the recursion tree as long as Step~\ref{trim0} of \textsc{Trim} does not apply. 
    
    We now prove \pref{fragprop1} for $S'=S\inter C' \in \C'$ by bottom-up induction on the recursion tree of the \textsc{Trim} subroutine, starting at the calls where $S'=S\inter C'$ is pruned from $\C'$ in Step~\ref{trim0}. 
    As base cases, we consider 
    instances of \textsc{Trim} with first parameter $C'$ where $S$ is removed from $\C'$ in the initial pruning step of \textsc{Trim}.
    Since $S'$ is removed from $\C'$, it is non-$(1-\delta)$-boundary sparse in $C'$, that is, we can set  $\mathcal A = \{\{C_i \subseteq C'\}, \{C_1, \dots, C_k\}\setminus \{C_i \subseteq C'\}\}$.

    For the induction step, we consider sets $S'$,  where $S'$ survives the initial pruning step of \textsc{Trim}.
    Then $\C' \neq \varnothing$ after Step~\ref{trim0}.
    This implies that $C'$ is replaced by $D$ and $C'\setminus D$ for some $D \subsetneq C'$.
    By induction on the recursive instances \textsc{Trim}$(C'\setminus D, \C'\inter (C'\setminus D))$ and \textsc{Trim}$(D, \C'\inter D)$, there exist two partitions $\P_D$ and $\P_{C'\setminus D}$ of $\{C_1, \dots C_k\}$, the first of $\{C_1, \dots C_k\} \inter D\subsetneq \{C_1, \dots C_k\}$ and the latter of $\{C_1, \dots C_k\} \inter C'\setminus D\subsetneq \{C_1, \dots C_k\}$, where for each collection $A \in \P_D\union  \P_{C'\setminus D}$, such that
    $S \inter \left(\union_{A' \in A} A'\right)$ is not $(1-\delta)$-boundary sparse in $\union_{A' \in A} A'$.
    Thus, we simply set $\mathcal A = \P_D \union \P_{C'\setminus D}$.
\end{proof}

The rest of this subsection is devoted to prove \pref{fragprop2}, focusing on the \textsc{Trim} subroutine.
We start by showing that the boundary-size of the clusters only gets smaller as we run \textsc{Trim}.

\begin{observation}\label{observation 6.5}
    If \textsc{Trim} decomposes $C'$ into $D$ and $C'\setminus D$ in \sref{trim1}, then $\boundary D \le \boundary C' - 0.04\lmin$ and $\boundary (C'\setminus D) \le \boundary C' - 0.04 \lmin$.
\end{observation}

\begin{proof}
    Recall that $\lmax = 1.2 \min$. 
    In \sref{trim1} we have that
    \begin{multline*}
        w(D, V\setminus C') = \boundary D - w(D, C'\setminus D) \ge \lambda - 0.4\lmax \ge \lmin -0.4 \lmax = 0.52 \lmin\\
        \Rightarrow \boundary (C'\setminus D) = \boundary C' - w(D, V\setminus C') + w(D, C'\setminus D) \le \boundary C' - 0.52 \lmin + 0.4 \lmax \le \boundary C' - 0.04\lmin
    \end{multline*}
    and by symmetry, we also obtain $\boundary D \le \boundary C' - 0.04 \lmin$.
    In other words, the boundary-size decreases by at least $0.04\lmin$ on both recursive instances.
\end{proof}

\begin{observation}\label{observation 6.6}
    If \textsc{Trim} decomposes in \sref{trim2} into $D$ and $C'\setminus D$, then $\boundary D \le 0.4 \delta \lmin$ and $\boundary (C'\setminus D) \le 0.4 \delta \lmin$.
\end{observation}

\begin{proof}
    If the algorithm decomposes in \sref{trim2}, then since $D$ is $(1-\delta)$-boundary sparse in $C'$, 
    \begin{multline*}
        \boundary (C'\setminus D) = \boundary C' - w(D, V\setminus C') + w(D, C'\setminus D)\\
        \le \boundary C' -\frac 1 {1-\delta} w(D, C'\setminus D) + w(D, C'\setminus D) \le \boundary C' -\delta w(D, C'\setminus D) \le \boundary C' - 0.4 \delta \lmax
    \end{multline*}
    where we used that $w(D, C'\setminus D) \ge 0.4\lmax$ since we are not in \sref{trim1}. By symmetry, $\boundary D \le \boundary C' - 0.4\delta \lmax \le  \boundary C' - 0.4\delta \lmin$.
\end{proof}

The two observations together show that the boundary-size of the cluster $C'$ \textsc{Trim} is called on decreases by at least $\min\{0.04 \lmin, 0.4\delta\lmin\} $ on both recursive instances $D$ and $C' \inter D$.

We now upper bound the number of clusters by tracking the following three parameters of a cluster $C'$ as \textsc{Trim} is called on it:
\begin{enumerate}
    \item The boundary-size $\boundary C'$,
    \item The maximum cut-size within $C'$ of a cut in $\C'$ after the pruning step (Step~\ref{trim0} of \textsc{Trim}), defined by
    $$
    c(C') = \max_{\substack{D\in \C'\\ \text{after}\\\text{pruning}}} \boundary _{G[C']}D
    $$
    which we define as $0$ if $\C' = \varnothing$,
    \item The cluster cardinality $\card{C'}$.
\end{enumerate}

We will give an upper bound on the number of clusters parametrized by those three parameters, by casing on the value of the maximum cut-size within $C'$ and a double induction on the other two parameters.

We claim that the first two parameters are at most $6\lmax$ and $\lmax$ respectively.
The first parameter is always at most $6\lmax$ since originally $\boundary C \le 6 \lmax$ by assumption, and we have shown in the two observations above that boundary-size never increases upon recursion. 
The second parameter is at most $\lmax$ originally (since it is the parameter given to LocalKCut) and it never increases upon recursion since the value $\boundary _{G[X]}D$ decreases as $X$ gets replaced by smaller (inclusion-wise) and smaller clusters upon recursion and $\C'$ can only lose elements as one walks down the recursion tree.

Let $f(b,c,d)$ be the maximum number of clusters in the final decomposition 
starting from a cluster $C'$ with $\boundary C' \le b$, $c(C') \le c$ and $\card{C'} \le d$. 
Our goal is to show that 
$$
f(b,c,d) \le O\left(\left(\frac b {\delta \lmin}\right) ^2 (\log d)^{O( b/ \lmin)}\right).
$$

We recursively bound $f(b,c,d)$ by stepping through each case of the \textsc{Trim} algorithm. 
To do so, we first make two observations.

\begin{observation}\label{observation 6.7}
    If \textsc{Trim} decomposes in \sref{trim1} into $D$ and $C'\setminus D$, then the maximum number of clusters in the final decomposition of each recursive instance is upper bounded by $f(b-0.04\lmin, c, d)$, i.e., we have $f(b,c,d) = 2f(b-0.04\lmin, c, d)$.
\end{observation}

\begin{proof}
    This observation follows directly from \Cref{observation 6.5} as it shows that $\boundary D, \boundary (C'\setminus D) \le \boundary C' - 0.04\lmin \le b-0.04 \lmin$.
\end{proof}

\begin{observation}\label{observation 6.8}
    If \textsc{Trim} decomposes in \sref{trim2}, let $D' = D$ if $\card D = \min\{\card D, \card{ C'\setminus D}$ and $D'=C'\setminus D$ otherwise. 
    Then we must have $\card {D'} \le \card {C'} /2$. 
    Moreover, all sets $U \in \C'$ that cross $D'$ must also cross $C'\setminus {D'}$.
\end{observation}

\begin{proof}
    The first statement is trivial.
    For the second statement, recall that $D \in \C'$ when this step is executed.
    Suppose for contradiction that there exists $U \in \C'$ that crosses $D'$ but not $C'\setminus D'$.
    Note that this implies that $U \ne D'$ and $U \ne C' \setminus D'$.
    Furthermore it either implies that $U \subsetneq D'$ or $C' \setminus D' \subsetneq U$. The latter is equivalent to $C'\setminus U \subsetneq D'$. 
    This contradicts that $D$ is chosen in \sref{trim2} to minimize $\min\{\card D, \card{C'\setminus D}\}$, as $U$ would have smaller cardinality.
\end{proof}

\begin{observation}\label{observation 6.9}
    If \textsc{Trim} decomposes in \sref{trim2}, we obtain the recursive bound
    $$
    f(b,c,d) \le \max\{1,\quad 2f(b-0.04\lmin, c, d) + f(b,c,d/2),  \quad f(b, 0.6c, d)+f(b-0.4\delta\lmin, c,d)\}
    $$
\end{observation}

\begin{proof}
    Suppose that \textsc{Trim} decomposes in \sref{trim2} into $D'$ and $C'\setminus D'$, with $D'$ defined as in \Cref{observation 6.8}. 
    By that observation, we have that $\card {D'} \le \card {C'} /2$ and moreover, all sets $U\in \C'$ that cross $D'$ must also cross $C'\setminus D'$.
    We case on whether $c(D') \ge 0.6 c(C')$, i.e., whether there exists $U \in \C'$ that crosses $D'$ with $\boundary_{G[D']}U \ge 0.6 c(C')$.
    \begin{enumerate}
        \item If such a set $U$ exists, then from $\boundary _{G[D']}U + \boundary _{G[C'\setminus D']}U\le \boundary _{G[C']}U$, we obtain $\boundary _{G[C'\setminus D']}U\le  \boundary _{G[C']}U - \boundary _{G[D']}U \le  c(C') - 0.6 c(C') \le 0.4 c(C')$.
        Since $U$ also crosses $C' \setminus D'$, it follows that the recursive instance $C'\setminus D'$ must go through \sref{trim1}. 
        Let $U_1, U_2$ be the decomposition of $C'\setminus D'$.
        Then by \Cref{observation 6.5}, $\boundary U_1\le \boundary (C'\setminus D') -0.04\lmin \le \boundary C'-0.04 \lmin$ and
        $\boundary U_2 \le \boundary (C'\setminus D') -0.04\lmin \le \boundary C'-0.04 \lmin$. 
        Also by \Cref{observation 6.8}
        $|D'| \le |C'|/2 \le d/2$.
By induction on clusters $D', U_1, U_2$, the total number of clusters at the end is at most
        $$
        2f(b-0.04\lmin, c,b) + f(b,c,d/2)
        $$
        \item Otherwise, we have $c(D') < 0.6 c(C') \le c/2$. By \Cref{observation 6.6}, we have that $\boundary (C'\setminus D') \le \boundary C' - 0.4\delta \lmin$.
        By induction on $D'$ and $C'\setminus D'$, the total number of clusters at the end is at most 
        $$
        f(b,0.6 c, d) + f(b-0.4\delta \lmin, c,d)
        $$
    \end{enumerate}
    Overall, if the algorithm decomposes by \sref{trim2}, we obtain the recursive bound

  $$
    f(b,c,d) \le \max\{1,\quad 2f(b-0.04\lmin, c, d) + f(b,c,d/2),  \quad f(b, 0.6c, d)+f(b-0.4\delta\lmin, c,d)\}
    $$
\end{proof}

Next, we state the base cases for our induction that bounds $f(b,c,d)$.
\begin{enumerate}
    \item Since we always have $\boundary C' \ge \lambda \ge \lmin$, we have the vacuous bound $f(b,c,d) = 0$ for $b< \lmin$.
    \item If $\C' = \varnothing$, no further recursion is necessary. Thus, $f(b,0,d)=1$.
    \item If $\card {C'} = 1$, the cluster cannot be cut any further. Thus, $f(b,c,1) = 1$.
\end{enumerate}

We now bound $f(b,c,d)$ using a case analysis depending on the value of the second parameter. 
As shown by \Cref{observation 6.7} and \Cref{observation 6.9}, the value of the second parameter never increases. 
Furthermore, if it drops to at most $0.4 \lmax$ on a set $C'$, then the algorithm always executes \sref{trim1} on subsets of $C'$. 
Thus, we can bound this setting first.

\paragraph{Case 1.} If $c(C') \le 0.4 \lmax$, then \textsc{Trim} either takes \sref{trim1} or does nothing (in the case $\C'=\varnothing)$, so
$$
f(b,c,d) \le \max\{1, 2f(b-0.04\lmin, c,d)\} \quad \text{for} \quad c\le 0.4\lmax
$$

Every recursive call with first parameter $D$ still has $c(D)\le 0.4\lmin $, so all recursive calls made in the subtree of the call with first parameter $C'$ is made in \sref{trim1}. 
We bound $f(b,c,d)$ by bounding the number of recursive calls of \textsc{Trim} as each call increases the number of clusters in $\mathcal{P}$ by 1.
The number of recursive calls in \sref{trim1} doubles for at most $ \ceil{\frac {b-\lmin} {0.04 \lmin}} \le \frac {b}{ {0.04 \lmin} }\le 25 \frac b \lmin$ levels of recursion and $2^x \ge 1$ for all positive $x$, so it holds that $f(b,c,d) \le 2^{25b/\lmin}$.

\paragraph{Case 2.} For $c \in (0.4\lmax, 0.61 \lmax]$, \textsc{Trim} either takes \sref{trim1} or \sref{trim2}, or does nothing.
Note that by \Cref{observation 6.7} and \Cref{observation 6.9}, in both cases we have 
 $$
    f(b,c,d) \le \max\{1,\quad 2f(b-0.04\lmin, c, d) + f(b,c,d/2),  \quad f(b, 0.6c, d)+f(b-0.4\delta\lmin, c,d)\}
    $$

We show the following bound by double induction on $b$ and $d$.

\begin{equation} \label{eq3}
f(b,c,d) \le \frac b {0.4 \delta \lmin} (2+2\log d)^{25b/\lmin}.
\end{equation}

The base case is trivially true if either $b<\lmin$ or $d=1$. 
For the induction step, note that $0.6 c \le 0.6\cdot 0.61 \lmax \le 0.4 \lmax$, so that a call on a set with second parameter $0.6c$ guarantees that \sref{trim1} is executed on this call and we can apply the bound from Case 1.
Now by \Cref{observation 6.9} it holds that:
\begin{alignat*}{2}
    f(b,c,d) &\le \max \{ &&1\\
    & &&2f(b-0.04 \lmin, c, d) + f(b,c,d/2), \elabel{eq2}\\
    & &&f(b,0.6c,d)+f(b-0.4\delta \lmin,c,d)\} \elabel{eq4}\\
    \intertext{We apply the induction hypothesis (\Cref{eq3}) to both terms of \Cref{eq2} and the second term of \Cref{eq4}, while we apply the bound of Case 1 to the first term of \Cref{eq4}:}
    f(b,c,d) &\le \max \{&&1,\\
     & && 2\cdot \frac b {0.4 \delta \lmin} (2+2\log d)^{25(b-0.04\lmin)/\lmin} + \frac b {0.4 \delta \lmin} (2+2\log \frac d 2)^{25b/\lmin},\elabel{eq5}\\
     & && 2^{25b/\lmin} + \frac {b-0.4\delta \lmin}{0.4 \delta \lmin}(2+2\log d)^{25b/\lmin}\}\elabel{eq6}\\
     \intertext{We develop the exponent in the first term of \Cref{eq5}, and develop the logarithm in its second term. We upper bound $2$ by $2+2\log d$ in \Cref{eq6}:}
     &\le \max \{&&1,\\
     & && 2\cdot \frac b {0.4 \delta \lmin} (2+2\log d)^{25b/\lmin-1} + \frac b {0.4 \delta \lmin} (2+2\log d )^{25b/\lmin-1}\cdot (2+2\log d-2),\\
     & && (2+2\log d)2^{25b/\lmin} + \frac {b-0.4\delta \lmin}{0.4 \delta \lmin}(2+2\log d)^{25b/\lmin}\}\\
     &=&& \frac b {0.4 \delta \lmin}(2+2\log d)^{25b/\lmin}.
\end{alignat*}
where the last equality comes from the fact that $\frac b {0.4 \delta \lmin}(2+2\log d)^{25b/\lmin} \ge \frac{b}{\delta \lmin} 2 ^{25b/\lmin} \ge \frac 1 \delta 2^{25} \ge 1$.
\paragraph{Case 3.} For $c \in [0.61 \lmax, \lmax]$, \textsc{Trim} either takes \sref{trim1} or \sref{trim2}, or does nothing.
Note that by \Cref{observation 6.7} and \Cref{observation 6.9}, in both cases we have 
 $$
    f(b,c,d) \le \max\{1,\quad 2f(b-0.04\lmin, c, d) + f(b,c,d/2),  \quad f(b, 0.6c, d)+f(b-0.4\delta\lmin, c,d)\}
    $$ 
    
    We show the bound
\begin{equation}\label{eq10}
f(b,c,d) \le \left(\frac b {0.4 \delta \lmin}\right)^2 (2+2\log d)^{25b/\lmin}
\end{equation}
by induction on $b$ and $d$.

The base case is trivially true if either $b < \lmin$ or $d=1$. 
For the induction step, note that $0.6c \le 0.6 \lmax$, which implies that a call on a set with second parameter $0.6c$ guarantees that the bound from Case 2 can be applied to this call.
Now by induction on $b \ge \lmin$ and $d>1$, it holds that

\begin{alignat*}{2}
    f(b,c,d) &\le \max \{ &&1\\
    & &&2f(b-0.04 \lmin, c, d) + f(b,c,d/2), \elabel{eq12}\\
    & &&f(b,0.6c,d)+f(b-0.4\delta \lmin,c,d)\} \elabel{eq14}\\
    \intertext{We apply the induction hypothesis (\Cref{eq10}) to both terms of \Cref{eq12} and the second term of \Cref{eq14}, while we apply the bound of Case 2 to the first term of \Cref{eq14}:}
    &\le \max \{&&1,\\
     & && 2\cdot \left(\frac b {0.4 \delta \lmin}\right)^2 (2+2\log d)^{25(b-0.04\lmin)/\lmin} + \left(\frac b {0.4 \delta \lmin} \right)^2(2+2\log \frac d 2)^{25b/\lmin},\elabel{eq15}\\
     & && \frac {b}{0.4 \delta \lmin }(2+2\log d)^{25b/\lmin} + \left(\frac {b-0.4\delta \lmin}{0.4 \delta \lmin}\right)^2(2+2\log d)^{25b/\lmin}\}\elabel{eq16}\\
     \intertext{We develop the exponent in the first term of \Cref{eq15}, and develop the logarithm in its second term. We upper bound $b-0.4\delta \lmin$ by $b$ in \Cref{eq16}:}
     &\le \max \{&&1,\\
     & && 2\cdot \left(\frac b {0.4 \delta \lmin}\right)^2 (2+2\log d)^{25b/\lmin-1}  \\ & && +\qquad \left(\frac b {0.4 \delta \lmin} \right)^2(2+2\log d)^{25b/\lmin-1}(2+2\log d-2),\\
     & && \frac {b}{0.4 \delta \lmin }(2+2\log d)^{25b/\lmin} + \left(\frac {b-0.4\delta \lmin}{0.4 \delta \lmin}\right)\left(\frac {b \lmin}{0.4 \delta \lmin}\right)(2+2\log d)^{25b/\lmin}\}\\
     &=&&  \left(\frac b {0.4 \delta \lmin}\right)^2 (2+2\log d)^{25b/\lmin}.
\end{alignat*}

Since for the set $C$ on which the algorithm was run we know that $b\le 6\lmax$, $c\le \lmax$, $d\le \card {C}$ and $\lmax \le 1.2 \lmin$, we obtain the desired bound $O(\delta^{-2}\log^{O(1)}\card C)$ on the number of clusters, which concludes the proof of \pref{fragprop2}.

\subsection{Running Time analysis}
\begin{lemma}
    The running time of \Cref{alg:fragmenting} is $\tilde O( \delta^{-2}\beta^2(\nu \lmax)^{O(\beta)})$.
\end{lemma}
\begin{proof}

    \sref{frag1} of \Cref{alg:fragmenting} makes $O(\lmax)$ many requests to LocalKCut. 
    By \Cref{thm:deterministic-dynamic}, this takes $\tilde O( \beta ^2(\nu \lmax)^{O(\beta)})$ time.
    In particular, we know that we have at most $\tilde O( \beta ^2(\nu \lmax)^{O(\beta)})$ many cuts in $\C$ in \sref{frag2} of \Cref{alg:fragmenting}, each of volume $\nu$.

    It remains to bound the time for \sref{frag3} of \Cref{alg:fragmenting}. Since we proved in the correctness proof that the algorithm added at most $O(\delta^{-2} \log^{O(1)} \card C)$ many clusters to $\mathcal{P}$ at termination, the recursion tree is binary  and each recursive call increases the number of clusters in $\mathcal{P}$ by 1, we know that \textsc{Trim} is run at most $O(\delta^{-2} \log^{O(1)} \card C)$ times. 
    It thus suffices to show that a run of \textsc{Trim} takes $\tilde O( \beta ^2(\nu \lmax)^{O(\beta)})$ time.

    In \sref{trim0}, for every cut in $\C'$, it takes $O(\nu)$ time to check whether a cut is boundary sparse by \Cref{thm:datastructure}, and thus $\tilde O( \beta ^2(\nu \lmax)^{O(\beta)})$ to check all cuts.
    Similarly, in \sref{trim1}, it takes $O(\nu)$ time for each cut to decide if a cut satisfies the if condition, $\tilde O( \beta ^2(\nu \lmax)^{O(\beta)})$ overall. 
    Removing $C'$ and adding $D$ and $C'\setminus D$ takes $O(\nu)$ time, as this is encoded by simply changing the labels of the boundary edges to inter-cluster. 
    Creating $\C' \inter D$ and $\C' \inter (C'\setminus D)$ takes $\tilde O( \beta ^2(\nu \lmax)^{O(\beta)})$ time, as there are $\tilde O( \beta ^2(\nu \lmax)^{O(\beta)})$ cuts of volume $\nu$ each to intersect with another set.

    Similarly, in \sref{trim2}, it takes $O(\nu)$ time for each cut to find its cardinality and the one of its complement, thus $\tilde O( \beta ^2(\nu \lmax)^{O(\beta)})$ overall to find the minimum. 
    Removing $C'$ and adding $D$ and $C'\setminus D$ takes $O(\nu)$ time, as this is encoded by simply changing the labels of the boundary edges to inter-cluster. 
    Creating $\C' \inter D$ and $\C' \inter (C'\setminus D)$ takes $\tilde O( \beta ^2(\nu \lmax)^{O(\beta)})$ time, as there are $\tilde O( \beta ^2(\nu \lmax)^{O(\beta)})$ cuts of volume $\nu$ each to intersect with another set.    
\end{proof}

\section{The Cluster Decomposition and the Cluster Hierarchy}
\label{sec:hierarchy}

In this section, we lay out the main definitions and data structures that we maintain throughout our algorithm. 
The central object we maintain is the cluster hierarchy, a hierarchy of cluster decompositions.

\paragraph{Choice of parameter values.}
We remind the reader that we have $c>0$, $\delta = 2^{O(\log^{3/4 -c } n)}\le 0.04$. We set $\lmin=2^{\Theta (\log^{3/4-c}n)}$ and $\lmax\le 1.2 \lmin$, $H=38^h = 2^{-O(\log^{1/2}n)}$, $\phi =\frac {\phi'} {38^h}= 2^{-\Theta (\log^{3/4}n)}$, $\rho= 2^{\Theta(\log ^{1/2}n)}$ and $ \alpha=\frac 1 {\poly \log n}.$

\subsection{Cluster Decomposition}
\begin{definition}[Pre-Cluster Decomposition]\label{def:clusterdecomposition}
 Let $\alpha, \phi, \delta, H \in \R$, $\lmax, \in \N$.
 An $(\alpha, \phi, \lmax, H, \delta)$-\emph{ pre-cluster decomposition} of a graph $G$ is a partition $\{C_j\}_{j \in [\ell]}$ of the vertex set of $G$, for some $\ell \in \N_+$. 
 Moreover, that partition satisfies, for each $j \in [\ell]$:
    \begin{enumerate}[label=\roman*., noitemsep]
    \item $C_j$ is connected.
        \item \label{inanexpander} $C_j$ is contained in an $(\alpha, \phi'=\phi H)$-boundary linked expander.
        \item $C_j$ satisfies either:
        \begin{enumerate}
            \item \label{easycondition} it contains no $(1-\delta)$-boundary sparse cut $S$ with $\boundary_G(S) \le \lmax$ and $\vol(S) \le \frac \lmax \phi$, or
            \item \label{hardcondition} it satisfies $\boundary C_j \le 6\lmax$.
        \end{enumerate}
    \end{enumerate}
\end{definition}

\begin{definition}[Cluster Decomposition]
    Let $\alpha, \phi, \delta \in \R$, $\lmax, \lmin \in \N$.
 An $(\alpha, \phi, \lmax, H, \delta)$-\emph{cluster decomposition} of a graph $G$ is a partition $\{C_j\}_{j \in [\ell]}$ of the vertex set of $G$, for some $\ell \in \N_+$, obtained from an $(\alpha, \phi, \lmax, H, \delta)$-pre-cluster decomposition, where each cluster $C_j$ that does not satisfy Condition~\ref{easycondition} is fragmented using \Cref{alg:fragmenting} (see \Cref{thm:fragmenting}).
\end{definition}
As $\alpha$, $\phi$, $\lmax$, $H$ and $\delta$ do not change, we simply use \emph{cluster decomposition} to denote an $(\alpha, \phi, \lmax,  H, \delta)$-cluster decomposition in the following.

The main point of the cluster decomposition is that it can approximate most minimum \proper cuts, that is, in most cases, if $S$ is a minimum \proper cut in $G$ of cut-size $\lambda$ with $\lmin \le \lambda \le \lmax$, then we can uncross $S$ by a cut $S'$ that consists of the union of some clusters, and that is an approximate minimum \proper cut.

Let us now talk about the corner cases. 
In fact, given a cluster decomposition we can always uncross a cut $S$ into a set $S'$.
However, $S'$ is not guaranteed to be a proper cut, that is, we could have $\boundary S'=0$.
In this case, we show $S$ is a \emph{local cut} in a \emph{mirror cluster}, defined next.

\begin{definition}
    A \emph{local cut} is a cut of volume at most $4\frac \lmax \phi$.
\end{definition}

Local cuts are central in our analysis, as they are easier to find and maintain than non-local cuts.
To deal with the local cuts, we introduce \emph{mirror clusters}, one for each cluster.

\begin{restatable}[Mirror Clusters, Graphs and Cuts]{definition}{mirrorclusters}\label{def:mirrorclusters}
    The {\em mirror cluster} $C'$ of a cluster $C$ in graph $G$ is the graph $G/(V\setminus C)$, that is, the graph where all nodes outside of $C$ have been contracted into one node.

    A {\em mirror graph} $G_{\P_1}$ of $G$ according to partition $\P_1$ is the graph that is the collection of all mirror clusters of the clusters in $\P_1$.

    A {\em mirror cut} $S_v$ of a vertex $v \in V$ of a graph $G$, given a partition $\P_1$, a volume $\nu$ and a maximum cut value $\lmax$, is defined as follows. Let $C_v \in \P_1$ be the cluster of $G$ such that $v \in C_v$. 
    Then $S_v$ is a proper cut of smallest cut-value among proper local cuts $S \subsetneq C_v$ of cut-value at most $\lmax$ and that contain $v$. Note that many cuts can be mirror cuts of $v$, and that such a cut may also not exist.
\end{restatable}

Mirror clusters have the nice property that all local cuts can be found by running LocalKCut on them.
We can thus state the main property of a cluster decomposition:

\begin{restatable}[Cluster Decomposition Uncrossing]{proposition}{clusterdecomposition}\label{lem:clusterdecomposition}
    Let $\{C_j\}_{j \in [\ell]}$ be an $(\alpha, \phi, \lmax, H, \delta)$-cluster decomposition of a graph $G$.
    If $\lmin \le {\lambda}\le \lmax$, then every minimum \proper cut $S$ of value ${\lambda}$ can be $(1+2\delta)$-approximated by a cut $S'$, that can be one of two cases:
    \begin{enumerate}[noitemsep]
        \item $S'$ is a cut of $G$ and $S'$ crosses no cluster.
        \item There exists a cluster $C_j$ such that $S'$ is a  local cut in the mirror cluster of $C_j$.
    \end{enumerate}
\end{restatable}
We prove this proposition in the rest of this section.

We let $\{C_j\}_{j \in [\ell]}$ be a cluster decomposition as defined in \Cref{lem:clusterdecomposition}, and first show that any minimum cut can cross at most two clusters using the following lemmata.

\begin{lemma}\label{lem:smallvolume}
  Let $\{C_j\}_{j \in [\ell]}$ be an $(\alpha, \phi, \lmax, H, \delta)$-cluster decomposition of a graph $G$ with $\alpha > \phi$.
    Let $S$ be a cut of $G$, and let $C_j$ be a cluster that $S$ crosses.
    Let $X$ equal $ S\inter C_j$ if $\vol(S\inter C_j) \le \frac {\vol (C_j)} 2 $, and let $X$ equal $C_j \setminus S$ otherwise.
    We have that $\vol(X) \le \frac {\boundary S} \phi$.

    In particular, if $S$ is a minimum proper cut of cut-size at most $\lmax$, we have that $\vol(X) \le \frac {\lmax} \phi$.
\end{lemma}

\begin{proof}
    As $C_j$ is a cluster in the expander decomposition, $C_j$ is contained (or equal to) a cluster in the pre-cluster decomposition, which itself is contained in an $(\alpha,\phi')$-boundary linked expander as per Condition~\ref{inanexpander} of \Cref{def:clusterdecomposition}. Let $C$ be this expander.
    All volumes in this proof are the volumes of the sets in $G[C]^{\frac \alpha {\phi'}}$, which are an upper bound of the volumes in $G$. 
    We use $G[C]^{\frac \alpha {\phi'}}$ as this is the graph on which we have the guarantee of being a $\phi$-expander as per \Cref{thm:expanderdecomposition}.
    
    Let $X' = S\inter C$ if $\vol(S\inter C) \le \frac {\vol (C)} 2 $, and $X'= C \setminus S$ otherwise.
    Since $G[C]^{\frac \alpha {\phi'}}$ is a $\phi$-expander, we have that $\vol(X') \le \frac {w(X, C\setminus X)} \phi \le \frac {\boundary S} \phi$.
    Therefore, $\vol(X' \inter C_j) \le \frac {\boundary S} \phi$.
    Recall that $\vol(X) = \min\{\vol(S \inter C_j), \vol(C_j \setminus S)\}$.
    If $X' = S\inter C$ then
    $\vol(X) \le \vol(S \cap C_j) \le \vol(S\inter C) = \vol(X') \le \frac {\boundary S} \phi. $
    If $X' = C \setminus S$ then
    $\vol(X) \le \vol(C_j \setminus S ) \le \vol(C \setminus S) = \vol(X') \le \frac {\boundary S} \phi. $
    Thus, the result follows.
\end{proof}

For the rest of this paper, we will thus simply call our $(\alpha, \phi')$-boundary linked expanders as $\phi$-expanders.

\begin{corollary}\label{cor:insidecut}
    Let $\P_1$ be a pre-cluster decomposition, and let $C_j$ be a cluster in that decomposition. 
    Let $S \subseteq C_j$ be a proper cut of $G$ of cut-size $\Lambda$. 
    Then there exists in the mirror cluster of $C_j$ a cut of cut-size $\boundary S$, of volume at most $3\frac \Lambda \phi$ (in $G$).
\end{corollary}

\begin{proof}
    By \Cref{lem:smallvolume}, we have that either $S$ or $C_j \setminus S$ has volume at most $\frac \Lambda \phi$. In the case where $\vol(S) \le \frac \Lambda \phi$, the claim is trivial. 
    In the case where $\vol(C_j\setminus S)\le \frac \Lambda \phi$, note that in the mirror cluster, we have that $X=(C_j\setminus S)\union \{v\}$ has boundary-size equal to $\boundary S = \Lambda$, where $v$ is the vertex representing the contracted $G\setminus C_j$. 
    It remains to show that $\vol(X)\le 3\frac \Lambda \phi$. Note that $\vol(v) \le \vol(C_j\setminus S) + \boundary S \le 2\frac \Lambda \phi$, and thus $\vol(X) \le \vol(v)+\vol(C_j\setminus S) \le 3 \frac \Lambda \phi$.
\end{proof}

\begin{lemma}\label{lem:atmosttwo}
    Let $\P_1$ be an $(\alpha, \phi, \lmax, H, \delta)$-pre cluster decomposition, and $\P_2$ be its associated cluster decomposition. Let $f: \P_2\mapsto \P_1$ be the mapping that maps each cluster $C \in \P_2$ to the cluster it is contained in in $\P_1$.
    Let $S$ be a minimum \proper cut $S$ of cut-size at most $\lmax$. 
    Let $C'_1, \dots, C'_k\in \P_2$ be the clusters that $S$ crosses.
    Then $\card{\{f(C'_1), \dots, f(C'_k)\}} \le 2$.
\end{lemma}

\begin{proof}
    Recall that $\delta \le \frac 1 3$.
    Assume by contradiction that $\card{\{f(C'_1), \dots, f(C'_k)\}}\ge 3$, and wlog assume the clusters $C_1$, $C_2$ and $C_3$ belong to $\{f(C'_1), \dots, f(C'_k)\}$.
    Since $S$ is a minimum \proper cut, we know that $w(S, V\setminus S)\le {\lmax}$.
    As the clusters are vertex disjoint, $w(S\inter C_1, C_1 \setminus S)+w(S\inter C_2, C_2 \setminus S)+w(S\inter C_3, C_3 \setminus S) \le w(S, V\setminus S) $ and thus there exists a $C_i$ such that $w(S\inter C_i, C_i \setminus S)\le \frac {w(S, V\setminus S)} 3$.

We have two cases. Either $C_i$ satisfies Condition~\ref{easycondition} of \Cref{def:clusterdecomposition}, which means $C_i$ contains no $(1-\delta)$-boundary sparse cut of volume at most $\frac \lmax \phi$ and of cut-size at most $\lmax$, or it satisfies Condition~\ref{hardcondition}. 
    
    \begin{enumerate}
        \item In the first case, since $C_i$ contains no $(1-\delta)$-boundary-sparse cuts of cut-size at most ${\lmax}$ and of volume at most $\frac \lmax \phi$, and as by \Cref{lem:smallvolume}, $S\inter C_i$ has volume at most $\frac \lmax \phi$, we have that $S\inter C_i$ is not $(1-\delta)$-boundary-sparse:
    $$
    \frac {w(S, V\setminus S)} 3 \ge w(S\inter C_i, C_i\setminus S) \ge (1-\delta) \min\{w(S\inter C_i, V\setminus C_i), w(V \setminus (S\inter C_i), V\setminus C_i)\}
    $$
    And hence either $w(S\inter C_i, V\setminus C_i) \le \frac {w(S, V\setminus S)} {3(1-\delta)} \le \frac {w(S, V\setminus S)} 2 $ and $S\inter C_i$ is a cut of cut-size at most $w(S\inter C_i, C_i\setminus S)+w(S\inter C_i, V\setminus C_i) \le \frac {w(S, V\setminus S)} 3 + \frac {w(S, V\setminus S)} 2$, a contradiction to $S$ being a minimum \proper cut (remember that $C_i$ is connected by definition of a pre-cluster decomposition and thus $S\inter C_i$ is a proper cut),
    or $w(C_i\setminus S, V\setminus C_i) \le \frac {w(S, V\setminus S)} {3(1-\delta)} \le \frac {w(S, V\setminus S)} 2 $ and $C_i\setminus S$ is a cut of cut-size at most $\frac {w(S, V\setminus S)} 3 + \frac {w(S, V\setminus S)} 2$, also a contradiction to $S$ being a minimum \proper cut.
    \item If $C_i$ satisfies Condition~\ref{hardcondition}, then let $\{A_1, \dots, A_\kappa\}\subsetneq \P_2$ be the partition of $C_i$ obtained by fragmenting $C_i$, and let $A_j\in \{C'_1, \dots, C'_k\} \inter \{A_1, \dots, A_\kappa\}$ be one of the fragmented clusters that $S$ crosses.
    Then $S\inter A_j$ is a proper cut.
    By \Cref{thm:fragmenting}, there is a subset of indices $Q\subseteq [\kappa]$ such that $A=\union_{q \in Q} A_q$ satisfies that $S\inter A$ is not $(1-\delta)$-boundary sparse in $A$.
    Similarly to the first case, we then have:
    $$
    \frac {w(S, V\setminus S)} 3 \ge w(S\inter C_i, C_i\setminus S)\ge w(S\inter A, A\setminus S)  \ge (1-\delta) \min\{w(S\inter A, V\setminus A), w(V \setminus (S\inter A), V\setminus A)\}
    $$
    And hence either $w(S\inter A, V\setminus A) \le \frac {w(S, V\setminus S)} {3(1-\delta)} \le \frac {w(S, V\setminus S)} 2 $ and $S\inter A$ is a cut of cut-size at most $w(S\inter A, A\setminus S)+w(S\inter A, V\setminus A) \le \frac {w(S, V\setminus S)} 3 + \frac {w(S, V\setminus S)} 2$, a contradiction to $S$ being a minimum \proper cut ($S\inter A$ must be a proper cut of $A$ as otherwise it cannot be non $(1-\delta)$-boundary-sparse in $A$),
    or $w(A\setminus S, V\setminus A) \le \frac {w(S, V\setminus S)} {3(1-\delta)} \le \frac {w(S, V\setminus S)} 2 $ and $A\setminus S$ is a cut of cut-size at most $\frac {w(S, V\setminus S)} 3 + \frac {w(S, V\setminus S)} 2$, also a contradiction to $S$ being a minimum \proper cut.    
    \end{enumerate}
\end{proof}

\begin{corollary}\label{cor:1}
 Let $\P_1$ be an $(\alpha, \phi, \lmax, H, \delta)$-pre-cluster decomposition, and $\P_2$ be its associated cluster decomposition. Let $f: \P_2\mapsto \P_1$ be the mapping that maps each cluster $C \in \P_2$ to the cluster that contains it in $\P_1$.
    Let $S$ be a minimum \proper cut $S$ of cut-size at most $\lmax$. 
     Let $C'_1, \dots, C'_k\in \P_2$ be the clusters $S$ crosses, if any.
    We have 3 cases:
    \begin{enumerate}[noitemsep, label=(\arabic*)]
    \item $\card{\{f(C'_1), \dots, f(C'_k)\}} =0 $.
    \item $\card{\{f(C'_1), \dots, f(C'_k)\}} =1$ \label{case1}
    \item $\card{\{f(C'_1), \dots, f(C'_k)\}} = 2$.  \label{case2}
    \end{enumerate}
\end{corollary}

If we are in Case~(1) of \Cref{cor:1}, then we trivially are in Case~1. of \Cref{lem:clusterdecomposition} by simply setting $S'=S$, as $\{f(C'_1), \dots, f(C'_k)\}$ being empty can only mean that $S$ crosses no cluster. If on the other hand we are in either case (2) or (3), we need to uncross $S$ to get $S'$. Next we show that in those two cases there exists  a cut $S'$ of cut cut-size at most $(1+2\delta) {\lambda} $ that crosses no cluster. The proofs exploit the definition of $(1-\delta)$-boundary sparseness and \Cref{lem:smallvolume}.

\begin{lemma}\label{lem:3.6}
    If we are in Case (2) of \Cref{cor:1}, then there exists a cut $S'$ that is a $(1+2 \delta)$-approximation of $S$, and satisfies either:
    \begin{enumerate}[label=\roman*., noitemsep]
        \item $\vol_{G/(V\setminus C_1)}(S') \le 4\lmax$, where $C_1$ is a cluster of $\P_1$, and $S'$ is a cut of $G/(V\setminus C_1)$, or \label{casei}
        \item $S' \subsetneq V$ crosses no cluster and $S'$ is a cut of $G$. \label{caseii}
    \end{enumerate}
\end{lemma}

\begin{proof}
    Wlog, assume that $S$ is connected. Also, wlog, let $C_1$ be the unique element of $ \{f(C'_1), \dots, f(C'_k)\}$.
    We have two cases. Either $C_1$ satisfies Condition~\ref{easycondition}, which means $C_1$ contains no $(1-\delta)$-boundary sparse cut of volume at most $\frac \lmax \phi$ and of cut-size at most $\lmax$, or it satisfies Condition~\ref{hardcondition}.
\begin{enumerate}
    \item 
   Since $C_1$ contains no $(1-\delta)$-boundary sparse cut of volume at most $\frac \lmax \phi$ and of cut-size at most $\lmax$, and $S\inter C_1$ has volume at most $\frac \lmax \phi$ by \Cref{lem:smallvolume}, we have that $w(S\inter C_1, C_1 \setminus S) \ge (1-\delta)\min\{w(C_1 \inter S, V\setminus C_1), w(C_1 \setminus S, V\setminus C_1)\}$.
    Let $X$ be the set among $C_1 \inter S$, $C_1 \setminus S$ that satisfies 
    $w(X, V\setminus C_1) = \min\{w(C_1 \inter S, V\setminus C_1), w(C_1 \setminus S, V\setminus C_1)\}$.
    Note that $w(X, V\setminus C_1) \leq w(S\cap C_1, C_1\setminus S)/(1-\delta) = w(X,C_1 \setminus X)/(1-\delta)$, which implies that
    $w(X, V\setminus C_1) - w(X,C_1 \setminus X) \le \delta \cdot w(X, V\setminus C_1) \le (\delta/(1-\delta)) w(S\cap C_1, C_1\setminus S) \le (3\delta/2) \cdot w(S\cap C_1, C_1\setminus S) $.

    Let's now temporarily set $S' = S \union X$ if $X \not\subset S$, or $S' = S \setminus X$ otherwise, i.e., we ``uncross'' $S$.
    Then we have that $w(S', V\setminus S') \le w(S, V\setminus S) - w(X, C_1 \setminus X) + w(X, V\setminus C_1) \le 
    w(S, V\setminus S) +  3\delta/2 \cdot w(S\cap C_1, C_1\setminus S)  \le
    (1+3\delta/2) w(S, V\setminus S)$.
    
    If $S'$ is a \proper cut ($\boundary S' \neq 0$), we are in Case~\ref{caseii}

    If $\boundary S'=0$, let $K$ be the connected component of $G$ that contains $S$. Then either $S=X$ (and $S'=\emptyset$) or $S=K \setminus X$ (and $S'= K$).
    We conclude by \Cref{cor:insidecut} that we are in Case~\ref{casei}.

    \item If $C_1$ satisfies Condition~\ref{hardcondition}, then let $\{A_1, \dots, A_\kappa\}\subsetneq \P_2$ be the partition of $C_i$ obtained by fragmenting $C_i$, and let $A_j\in \{C'_1, \dots, C'_k\} \inter \{A_1, \dots, A_\kappa\}$ be one of the fragmented clusters that $S$ crosses.
    Then $S\inter A_j$ is a proper cut.
    By \Cref{thm:fragmenting}, there is a subset of indices $Q\subseteq [\kappa]$ such that $A=\union_{q \in Q} A_q$ satisfies that $S\inter A$ is not $(1-\delta)$-boundary sparse in $A$.

    Similarly to the first case, we have that $w(S\inter A, A \setminus S) \ge (1-\delta)\min\{w(A \inter S, V\setminus A), w(A \setminus S, V\setminus A)\}$.
    Let $X$ be the set among $A \inter S$, $A \setminus S$ that satisfies 
    $w(X, V\setminus A) = \min\{w(A \inter S, V\setminus A), w(A \setminus S, V\setminus A)\}$.
    Note that $w(X, V\setminus A) \leq w(S\cap A, A\setminus S)/(1-\delta) = w(X,A \setminus X)/(1-\delta)$, which implies that
    $w(X, V\setminus A) - w(X,A \setminus X) \le \delta \cdot w(X, V\setminus A) \le (\delta/(1-\delta)) w(S\cap A, A\setminus S) \le (3\delta/2) \cdot w(S\cap A, A\setminus S) $.

    Let's now temporarily set $S' = S \union X$ if $X \not\subset S$, or $S' = S \setminus X$ otherwise, i.e., we ``uncross'' $S$.
    Then we have that $w(S', V\setminus S') \le w(S, V\setminus S) - w(X, A \setminus X) + w(X, V\setminus A) \le 
    w(S, V\setminus S) +  3\delta/2 \cdot w(S\cap A, A\setminus S)  \le
    (1+3\delta/2) w(S, V\setminus S)$.
    
    If $S'$ is a \proper cut ($\boundary S' \neq 0$), we are in Case~\ref{caseii}

    If $\boundary S'=0$, let $K$ be the connected component of $G$ that contains $S$. Then either $S=X$ (and $S'=\emptyset$) or $S=K \setminus X$ (and $S'= K$).
    
    We conclude by \Cref{cor:insidecut} that we are in Case~\ref{casei}.

    \end{enumerate}
\end{proof}

\begin{lemma}\label{lem:3.7}
    If we are in Case 3 of~\Cref{cor:1}, then there exists a cut $S'$ of either $G/(V\setminus C_1)$ or $G/(V\setminus C_2)$ ($\{C_1, C_2\} =\{f(C'_1), \dots, f(C'_k)\}$) that is a $(1+\frac 3 4 \delta)$-approximation of $S$, and satisfies $\vol_{G/(V\setminus C_i)} (S') \le 4\lmax$ (where $j \in\{1,2\}$ depending on where $S'$ exists).
\end{lemma}

\begin{proof}
Let ${\lambda}=w(S, V\setminus S)$.
    We have that $w(C_1\inter S, C_1 \setminus S) + w(C_2\inter S, C_2 \setminus S) \le w(S, V\setminus S) = {{\lambda}}$.
    Assume thus wlog that $w(C_1\inter S, C_1 \setminus S) \le \frac {{\lambda}} 2$.

    Again, we have two cases:
    \begin{enumerate}
        \item If $C_1$ satisfies Condition~\ref{easycondition}, it contains no $(1-\delta)$-boundary sparse cuts of volume at most $\frac \lmax \phi$ of cut-size at most $\lmax$, and thus we have that $S\inter C_1$ is not $(1-\delta)$-boundary sparse in $C_1$. Thus $w(S\inter C_1, C_1 \setminus S) \ge (1-\delta)\min\{w(C_1 \inter S, V\setminus C_1), w(C_1 \setminus S, V\setminus C_1)\}$.
    Let $X$ be the set among $C_1 \inter S$, $C_1 \setminus S$ that satisfies 
     $w(X, V\setminus C_1) = \min\{w(C_1 \inter S, V\setminus C_1), w(C_1 \setminus S, V\setminus C_1)\}$.
    It follows that 
    $w(X, C_1 \setminus X) = w(S\inter C_1, C_1 \setminus S) \ge (1-\delta)w(X, V\setminus C_1)$.
    We thus have that $w(X, V\setminus X) = w(X, V\setminus C_1) + w(X, C_1\setminus X) \le ((1/(1-\delta) + 1) w(X, C_1\setminus X) \le ({{2-\delta}\over {2(1-\delta)}}) {{\lambda}}  = (1+ {\delta \over 2-2\delta}) {\lambda} \le (1+3\delta/4){\lambda}$. 
    
    We conclude by \Cref{cor:insidecut}.

    \item  If $C_1$ satisfies Condition~\ref{hardcondition}, then let $\{A_1, \dots, A_\kappa\}\subsetneq \P_2$ be the partition of $C_1$ obtained by fragmenting $C_1$, and let $A_j\in \{C'_1, \dots, C'_k\} \inter \{A_1, \dots, A_\kappa\}$ be one of the fragmented clusters that $S$ crosses.
    Then $S\inter A_j$ is a proper cut.
    By \Cref{thm:fragmenting}, there is a subset of indices $Q\subseteq [\kappa]$ such that $A=\union_{q \in Q} A_q$ satisfies that $S\inter A$ is not $(1-\delta)$-boundary sparse in $A$, and $A_j \subseteq A$.

    Similarly to the first case, we have that $S\inter A$ is not $(1-\delta)$-boundary sparse in $A$. Thus $w(S\inter A, A \setminus S) \ge (1-\delta)\min\{w(A \inter S, V\setminus A), w(A \setminus S, V\setminus A)\}$.
    Let $X$ be the set among $A \inter S$, $A \setminus S$ that satisfies 
     $w(X, V\setminus A) = \min\{w(A \inter S, V\setminus A), w(A \setminus S, V\setminus A)\}$.
    It follows that 
    $w(X, A \setminus X) = w(S\inter A, A \setminus S) \ge (1-\delta)w(X, V\setminus A)$.
    We thus have that $w(X, V\setminus X) = w(X, V\setminus A) + w(X, A\setminus X) \le ((1/(1-\delta) + 1) w(X, A\setminus X) \le ({{2-\delta}\over {2(1-\delta)}}) {{\lambda}}  = (1+ {\delta \over 2-2\delta}) {\lambda} \le (1+3\delta/4){\lambda}$. 
    We conclude by \Cref{cor:insidecut} that we are in Case~\ref{casei}.

    \end{enumerate}
\end{proof}

Note that both lemmata show the existence of a \emph{\proper cut} $S'$, i.e., $S'$ is required to be both nonempty and different from $V$.
This concludes the proof of \Cref{lem:clusterdecomposition}.

\subsection{Cluster Hierarchy}

The $(\alpha, \phi, \lmax, H, \delta)$-cluster decomposition and the mirror clusters are  the building blocks of the cluster hierarchy, defined as follows:

\begin{definition}[Cluster Hierarchy]\label{def:clusterhierarchy}
    Let $\hbar$ be a positive integer. An $h$-cluster hierarchy is a set of graphs $G_1, \dots, G_\hbar$ with $G_1=G$
    with $\vol({G}_\hbar) \le \frac \lmax \phi$, together with $\hbar$ many $(\alpha, \phi, \lmax, H, \delta)$-cluster decompositions, one for each ${G}_j, j\in [\hbar]$, and a mirror cluster for each cluster of the pre-cluster decomposition of each graph. 
    Moreover, each ${G}_j, j>1$ is obtained from ${G}_{j-1}$ by contracting each cluster in the cluster decomposition.
\end{definition}

The main point of the $\hbar$-cluster hierarchy is to ensure that we only have to maintain local cuts and their cut-sizes, in other words, to transform any minimum cut into a local cut, if not on this level of the hierarchy then further down in the hierarchy.

We will use the $\hbar$-cluster hierarchy as follows: Let $S$ be a minimum proper cut of cut-size in $[\lmin, \lmax]$. 
Assume for the moment, that we have access to all local cuts.
By \Cref{lem:clusterdecomposition}  there exists in ${G}_1$ an approximate cut $S'_1$ that is either local or that crosses no cluster.
In the former case, we are done. In the latter, $S_1'$ induces a cut in ${G}_2$. 
Again by \Cref{lem:clusterdecomposition} there exists in ${G}_2$ an approximation $S'_2$ of $S_1'$ that is either local or that crosses no cluster.
Again, if it is local, we are done. If it is not, then we move on to ${G}_3$. 
Following this logic, we arrive the following proposition:

\begin{proposition}\label{prop:clusterhierarchy}
     Let $\hbar$ be a positive integer. In an $\hbar$-cluster hierarchy of $G$, there exists a local cut $S'$ in a mirror cluster $C$ of the pre-cluster decomposition of a collapsed graph $G_j$ or in the final collapsed graph, that is a $(1+2\delta)^\hbar$ approximation of the minimum proper cut $S$.
    Moreover, there exists a corresponding cut $S$ in $G$ such that $S$ and $S'$ have the same cut-size.
    If, in addition, $\delta < \frac 1 {2\boundary S}$, then $\boundary S = \boundary S'$, that is, $S'$ is a minumum proper cut.
\end{proposition}

\begin{proof}
    It suffices to apply \Cref{lem:clusterdecomposition} iteratively starting at $G$ until we end up in Case 2.\ of that proposition. 
    If we never arrive in Case 2, it follows that the recursion reaches the graph ${G}_\hbar$, which has volume $\frac\lmax \phi$. Thus, all its cuts are local.
    Once a local cut is found, to find the corresponding cut in $G$, since all graphs (including mirror clusters) in the hierarchy are contractions of graphs from the level above, it suffices to uncontract the graph.
    Uncontracting all the way to the top gives the result.
    The approximation value comes from the fact that on each level, the cut-size of the cut increases by a factor of at most $(1+2\delta)$.
\end{proof}
The goal of further sections will thus be to compute and maintain such a cluster hierarchy.

\section{The Mirror Cuts Data Structure}
\label{sec:mirrorcuts}

In this section, we present the Mirror Cuts Data Structure, which maintains the mirror graph and mirror cuts under a sequence of batch edge updates.
A batch update is an update that includes many edge insertions and deletions at once.

It is necessary to consider batch updates, as the data structure is designed to handle the refinement of clusters, that is, replacing a cluster $C$ by $S\subsetneq C$ and $C\setminus S$. 
Typically, this is done by deleting all of the edges between $S$ and $C\setminus S$. 
Doing so one by one proves to not work, as our data structure relies on the LocalKCut algorithm, which is only guaranteed to be correct if at every request, one can ensure that the minimum cut of the cluster is large enough.
However, by deleting the edges one by one, the two clusters are only considered two different ones by the LocalKCut when they are not connected anymore, and thus $C$ is still a cluster, with minimum cut exactly $1$ just before the deletion of the last edge. This breaks our data structure.

Hence, we introduce batch updates, so all of the edge deletions can be forwarded to the LocaKCut algorithm \emph{before} any request to it is made.

The results in this section rely on the results of \Cref{sec:localkcut}, and we assume that the data structure we present in \Cref{sec:datastructure} supports our cluster decomposition.
First, we remind the reader about mirror clusters, mirror graphs and mirror cuts:

\mirrorclusters*

\begin{theorem}[Mirror Cuts]\label{thm:mirrorcuts}
    Let $G_{\P_1}=(V,E)$ be a fully-dynamic mirror graph under batch updates (updates with up to $m$ many insertions or deletions) of a graph $G$ and partition ${\P_1}$ and $\lmax, \nu, \beta \in \N_+$ satisfying $\beta\le \lmax \le \nu$, and no mirror cluster $C'\subseteq G_{\P_1}$ strictly contains a cut $S$ with $\boundary_{C'}S\le \frac \lmax \beta$. 

    There exists a data structure that maintains, for each vertex $v \in V$, one of its mirror cuts $S_v$.  
    It can return in constant time the value of the smallest mirror cut $S_{v_{\min}}$, a pointer to $v_{\min}$
    and in $\card{S_{v_{\min}}}$ pointers to the vertices of $S_{v_{\min}}$.
    
   The algorithm runs in $\Tilde O ((m+n)\beta^2(\lmax \nu)^{O(\beta)}) $ preprocessing time, $\Tilde O ( \beta^4(\lmax \nu)^{O(\beta)})$ update time per edge update, $\Tilde O (K\cdot \beta^4(\lmax \nu)^{O(\beta)})$ update time for an update with $K$ edge updates, and $O(1)$ request time. 
   
\end{theorem}

The idea is as follows. We run an instance of LocalKCut (\Cref{thm:deterministic-dynamic}) on $G_{\P_1}$, with parameters $\lmax, \nu, \beta$, updating it as updates arrive. As no mirror cluster $C'\subseteq G_{\P_1}$ strictly contains a cut $S$ with $\boundary_{C'}S\le \frac \lmax \beta$, any request on a node $v \in C'_v$ yields all cuts $S \subsetneq C'_v$ with $\vol(S) \le \nu$, $\boundary_{C'_v} S \le \lmax$ and $v \in S$.

When an edge is deleted, what can happen is that this edge used to contribute to a cut that was too large to be the only mirror cut of a node $v$, and thus was not necessarily the one maintained by the data structure, but it now needs to be one. This cut is easily found by formulating a request on both endpoints of that edge, and then can be stored at any vertex it includes for which it is a mirror cut.

When an edge is inserted, it might contribute to a cut that was a mirror cut, and thus maintained, but is not anymore after the insertion. In fact, there might be many such cuts, but we show below that there cannot be too many of them. We must then formulate a request to LocalKCut from all those vertices whose maintained cut has changed due to that edge insertion. 

Finally, we deal with set deletions by simply deleting all the edges in the set in a batch update.

\subsection{Data Structure and Algorithm Description}
We store these cuts using the following data structure: 
\begin{itemize}
    \item For each mirror cut, we keep pointers to all the vertices for which it is the mirror cut. We also keep pointers to all the vertices it contains. We moreover keep the cut-size of the cut.
    We delete the cut once it is not a mirror cut anymore, as described in \Cref{alg:deletions}.
    \item For each vertex in a cluster $C$ we keep a pointer to {\em one of its} mirror cuts and its corresponding cut-size. We also keep pointers to all the mirror cuts of other vertices that it is contained in.
    \item We also maintain a heap of the mirror cuts, prioritized by their cut-sizes.
\end{itemize}

\begin{algo}[Deleting a mirror cut]\label{alg:deletions}
    Whenever the mirror cut of a node $u$ changes from $S$ to $S'$, we remove $u$ from the list of nodes stored at $S$, for which $S$ is the mirror cut. 
    We check if this list becomes empty. 
    If it is the case, we delete all pointers from and to the nodes it contains and  then delete the mirror cut and update the heap accordingly.
\end{algo}

\begin{lemma}\label{lem:deletion}
    \Cref{alg:deletions} takes $O(\nu)$ time as a mirror cuts has volume $O(\nu)$.
\end{lemma}
The high-level description of the Mirror Cuts Data Structure is as follows:
During preprocessing of this data structure, we  build this data structure.
To process an edge insertion, we note that the edge insertion might increase the cut-size of the mirror cuts it is incident to.
We must therefore request LocalKCut on all the vertices whose stored cut has been affected by this insertion, as the stored cut might not be the mirror cut anymore.
Note that we can ensure that not many such nodes exist.
To process an edge deletion, we note that the deletion can decrease the cut-size of a cut that was not a mirror cut before the deletion, but becomes one after this deletion. 
We can find all such cuts by running LocalKCut from the endpoints of the deleted edge.

We next give the details of this algorithm. %
\begin{algo}[Mirror Cuts Data Structure]\label{alg:maintainmirrors}

   \item  Input:\begin{enumerate}
        \item A mirror graph $G_{\P_1}$ of graph $G$, where edges can be inserted and deleted in batches
        \item A maximum cut value $\lmax$
        \item A maximum volume value $\nu$
        \item An approximation ratio $\beta$
    \end{enumerate} 
    
    \item Data Structure: 
    \begin{enumerate}
        \item For every node, we maintain its Mirror Cut and the corresponding value.
        \item For every Mirror Cut, we maintain pointers from and to the vertices it contains in a sorted list.
        \item A LocalKCut data structure for $G_{\P_1}$ with parameters $\lmax, \nu, \beta$
        \item A heap on the vertices prioritized by the cut-size of the associated mirror cut
    \end{enumerate}

    \item Preprocessing: 
    \begin{enumerate}
        \item \label{cutspre1} Instantiate the LocalKCut data structure
        \item \label{cutspre2} Process every node using \Cref{alg:processing}.
        \item \label{cutspre3} Instantiate the heap
    \end{enumerate}

\justaflag

    \item Handling updates:
    \begin{enumerate}
    \item \label{jhkstep1} Update the LocalKCut data structure.
        \item \label{jhkstep2} For every insertion of an edge $(u,v)$:
 \begin{enumerate}
     \item \label{cutsins1} Update the value of all maintained mirror cuts that contain exactly one of its endpoints $u$ or $v$. Update the heap accordingly.
     \item \label{cutsins2} Mark for processing all the nodes $u$ whose maintained mirror cut has seen its value increase because of the update. 
 \end{enumerate}

 \item \label{jhkstep3} For every deletion of an edge $(u,v)$:
 \begin{enumerate}
     \item \label{cutsdel1} Update the value of all maintained mirror cuts that contain either $u$ or $v$ but not both. Update the heap accordingly.
     \item \label{cutsdel2} Mark for processing the endpoints of the deleted edge.
 \end{enumerate}
 \item \label{cutsprocess} Call \Cref{alg:processing} for all the nodes marked for processing.

    \end{enumerate}
 
\end{algo}

\begin{algo}[Processing]\label{alg:processing}
    Input: a vertex $v$.
    \begin{enumerate}
        \item Issue a LocalKCut query for $v$, which returns a set of cuts $\mathcal S$.
        \item  For every cut $S \in \mathcal S$ with $\boundary_{G_{\P_1}}S \neq 0$, for every $u \in S$, if $S$ is a smaller cut than the one maintained for $u$:
         \begin{enumerate}[label=(\alph*), noitemsep]
        \item Build a data structure for $S$ if $S$ was no mirror cut before the update. 
        \item Store $S$ as mirror cut of $u$.
        \item Remove the old mirror cut of $u$
    \end{enumerate}
    \end{enumerate}
\end{algo}

We first show the correctness of the above algorithm and then analyze its running time.

\subsection{Correctness}
\begin{lemma}\label{lem:correctmaintainmirrorcuts}
    The Mirror Cuts Data Structure (\Cref{alg:maintainmirrors}) maintains a mirror cut of $G$ for every vertex $v \in G$ according to partition $\P_1$.  
\end{lemma}

\begin{proof}
    We will prove this by induction on the number $t$ of batch updates. The base case is clear, as after preprocessing, after all nodes are processed, the minimum local cut around each node is found, as guaranteed by~\Cref{thm:deterministic-dynamic} and then stored.
    For the induction step, assume that after update $t$, for $t \ge 0$, we have that all cuts stored are mirror cuts.

   Let $u$ be a node and let $S$ be its maintained mirror cut after the batch of updates $t$. 
   Note that $S$ could still be a mirror cut of $u$ at time $t+1$, in which case there is nothing to do, and no processing can change $S$ into another cut. 
   Indeed, \Cref{alg:processing} will only overwrite it if it finds a smaller local cut, which is impossible.
   Assume next that there exists a local cut $S'$ such that $u \in S'$ and $\boundary _{t+1} S' < \boundary _{t+1} S$.
   By definition, we have that $\boundary _{t} S \le \boundary _t S'$.
   Then either we have that the cut value of $S$ has increased, or the one of $S'$ has decreased.
   In the first case, an edge insertion $(u,v)$ with $u \in S'$, $v \notin S'$ forces the processing of $u$. $S'$ is then found by processing $u$, as $S'$ has small volume and thus is found by the request to LocalKCut. In the second case, an edge deletion $(u,v)$ with $u \in S'$, $v \notin S'$ forces the processing of $u$, which finds $S'$.
   Both processing steps will update the cut associated with $u$.
\end{proof}

\subsection{Running time analysis}
\begin{lemma}\label{alg:vertexprocessing}
    Processing a vertex with \Cref{alg:processing} takes $\Tilde O (\beta^2(\lmax \nu)^{O(\beta)}) $ time.
\end{lemma}

\begin{proof}
    Each time we process a node, we make a request to LocalKCut.
    By \Cref{thm:deterministic-dynamic}, each request takes $\Tilde O (\beta^2(\lmax \nu)^{O(\beta)}) $ time. It outputs at most $\Tilde O (\beta^2(\lmax \nu)^{O(\beta)}) $ many cuts, and each cut can update at most $\nu$ many vertices, which in turn could each  enforce the deletion of one maintained cut, which takes  $ O(\nu)$ time by \Cref{lem:deletion}. Hence, the deletions of maintained cuts can take time $\Tilde O (\beta^2(\lmax \nu)^{O(\beta)}) $.
    
    Moreover, each cut found by LocalKCut might need to be stored. In the worst case, each cut needs to be created in the Mirror Cuts data structure, and be linked to all its vertices, this takes time $O(\nu)$.
    Overall, the running time is $\Tilde O (\beta^2(\lmax \nu)^{O(\beta)}) $.
\end{proof}

\begin{corollary}
    Preprocessing of the Mirror Cuts Data Structure (\Cref{alg:maintainmirrors}) on a mirror graph $G_{\P_1}$ with $m$ edges and $n$ nodes takes $\Tilde O ((m+n)\beta^2(\lmax \nu)^{O(\beta)}) $ time.
\end{corollary}

\begin{proof}
    By~\Cref{thm:deterministic-dynamic}, Step~\ref{cutspre1} of preprocessing takes $\tilde O( m \beta^2 (\lambda_{\max}\nu)^{O(\beta)})$ time. 
    By~\Cref{alg:vertexprocessing}, Step~\ref{cutspre2} takes $\Tilde O (n\beta^2(\lmax \nu)^{O(\beta)}) $ time. Step~\ref{cutspre3} takes $\tilde O(n)$ time.
\end{proof}

\begin{corollary}
    Handling an edge deletion in Steps~\ref{jhkstep3} and~\ref{cutsprocess} in the Mirror Cuts Data Structure (\Cref{alg:maintainmirrors}) takes $\Tilde O (\beta^2(\lmax \nu)^{O(\beta)})$ time.
\end{corollary}

\begin{proof}
Let $e = (u,v)$ be a deleted edge. 
We start by arguing that there are not many mirror cuts that need to be updated in Step~\ref{cutsdel1}.
Consider a mirror cut $S$ containing node $u$, that is, there exists a node $u'$ such that $S$ is the mirror cut of $u'$, and $u \in S$.
All of these cuts are found by a request on node $u$ to LocalKCut, which takes $\Tilde O (\beta^2(\lmax \nu)^{O(\beta)}) $ time, and thus there are at most $\Tilde O (\beta^2(\lmax \nu)^{O(\beta)}) $ such cuts.
Hence, deleting an edge can affect at most $\Tilde O (\beta^2(\lmax \nu)^{O(\beta)}) $ many cuts whose cut-size have to be updated. 
Finding those cuts in the data structure is trivial as one can access all cuts including $u$ or $v$ from the endpoints of the deleted edge and check in $\tilde O(1)$ whether they include both or only one vertex by looking at the sorted list of nodes in each cut, and thus Step~\ref{cutsdel1} takes $\Tilde O (\beta^2(\lmax \nu)^{O(\beta)}) $ time.

For Step~\ref{cutsdel2}, the deleted edge marks two vertices for processing in Step~\ref{cutsprocess}. Processing both endpoints of the deleted edge takes $\Tilde O (\beta^2(\lmax \nu)^{O(\beta)}) $ time by \Cref{alg:vertexprocessing}, hence the result follows.
\end{proof}

\begin{corollary}
    Handling an edge insertion in Steps~\ref{jhkstep2} and~\ref{cutsprocess}  in the Mirror Cuts Data Structure (\Cref{alg:maintainmirrors}) takes $\Tilde O (\beta^4(\lmax \nu)^{O(\beta)})$ time.
\end{corollary}

\begin{proof}
Let $e = (u,v)$ be an inserted edge.
We start by arguing that there are not many mirror cuts that need to be updated in Step~\ref{cutsins1}.
As in the case for edge deletions, consider a mirror cut $S$ containing node $u$, that is, there exists a node $u'$ such that $S$ is the mirror cut of $u'$, and $u \in S$.
All of these cuts would be found by a request on node $u$ to LocalKCut, which takes $\Tilde O (\beta^2(\lmax \nu)^{O(\beta)}) $ time, and thus there are at most $\Tilde O (\beta^2(\lmax \nu)^{O(\beta)}) $ such cuts.
Hence, inserting an edge can affect at most $\Tilde O (\beta^2(\lmax \nu)^{O(\beta)}) $ many cuts whose cut-size have to be updated. 
Finding those cuts in the data structure is trivial as one can access all cuts including $u$ or $v$ from the endpoints of the deleted edge and check in $\tilde O(1)$ whether they include both or only one vertex by looking at the sorted list of nodes in each cut, and thus Step~\ref{cutsins1} takes $\Tilde O (\beta^2(\lmax \nu)^{O(\beta)}) $ time.

Each affected cut in Step~\ref{cutsins1} induces the marking in Step~\ref{cutsins2} for processing in Step~\ref{cutsprocess} of all of the nodes it is the mirror cut of, which there might be $O(\nu)$ of. Hence, inserting an edge can affect at most $\Tilde O (\beta^2(\lmax \nu)^{O(\beta)}) $ many nodes that all need to be processed, which takes $\Tilde O (\beta^2(\lmax \nu)^{O(\beta)}) $ time each by \Cref{alg:vertexprocessing}, hence the result follows.
\end{proof}

\begin{lemma}
    Handling a batch of updates with $K$ edge updates takes  $\tilde{O}\left(K\cdot\beta^4( \lambda_{\max}\nu)^{O(\beta)}\right)$ amortized update time.
\end{lemma}

\begin{proof}
    Step~\ref{jhkstep1} takes $\tilde{O}\left(\beta^4( \lambda_{\max}\nu)^{O(\beta)}\right)$ per edge by \Cref{thm:deterministic-dynamic}. 
    The times for Step~\ref{jhkstep2} and Step~\ref{jhkstep3} are discussed in the two corrolaries just above.
\end{proof}

\newcommand{\AAA}[0]{\textsc{ApproxMinCut} }
\newcommand{\UP}[0]{\textsc{UpdatePartition} }

\section{Static $(1 + \delta$)-Approximate Minimum Proper Cut Algorithm for a Restricted Minimum Cut Range}\label{sec:static}

For this section, we assume that the data structure we present in \Cref{sec:datastructure} supports our cluster decomposition.

Before diving into the dynamic algorithm, we present a simple static algorithm that we will dynamize in later sections.

\begin{theorem}\label{thm:static}
    Let $G$ be an unweighted undirected graph. Let $c>0$.
    Let $\lmin$ and $\lmax$ be parameters with $\lmax \le 1.2 \lmin = 2^{O(\log ^{3/4-c}n)}$, and assume that the minimum \proper cut of $G$ is at least $\lmin$.
    Let $\phi = 2^{-\Theta(\log ^{3/4}n)}$ and $\delta = 2^{-\Theta(\log^{3/4-c} n)}$.
    Then, there is an algorithm that, if $\vol (G) \neq O(\frac \lmax \phi)$:
\begin{enumerate}
    \item \label{item1} Computes a pre-cluster decomposition $\P_1$ of $G$.
    \item \label{item2} Computes a corresponding cluster decomposition $\P_2$
    \item \label{item3} Creates the mirror graph $G_{\P_1}$.
    \item \label{item7} Calls itself recursively on the contracted graph $G/\P_2$.
    \item \label{item4} If the minimum proper cut of $G$ is at most $\lmax$, returns a $(1+\delta)$-approximation of the minimum proper cut value.
    \item \label{item5} If the minimum proper cut of $G$ is at most $\lmax$, maintains pointers to the cut that achieves the minimum proper cut, and can return on request the nodes of the side $S$ of smaller volume in time $\tilde O(\card S)$ or the cut-edges in time $\lmax n^{o(1)}$.
    \item \label{item6} If the value of the minimum proper cut is strictly larger than $\lmax$, report ``$\lambda > \lmax$''.
\end{enumerate}
If $\vol(G)=O(\frac \lmax \phi)$, the algorithm simply returns the minimum proper cut.

This algorithm runs in $m^{1+o(1)}$ time.
\end{theorem}

The main idea is as follows: we first decompose the graph into expanders, then decompose the expanders further into smaller connected graphs, called \emph{clusters}, so that every minimum \proper cut $S$ in the graph can be uncrossed for each cluster that it cuts. 
This step is decomposed into two substeps, one is to compute the pre-cluster decomposition $\P_1$, and then compute the cluster decomposition $\P_2$. 
In the end, $S$ can be approximated either by a cut that consists of a union of (full) clusters or by a local cut.
We then contract the clusters and recurse on the resulting graph, as clustering only once reduces the number of edges by an $n^{o(1)}$ factor only, and thus we need to recurse multiple times before we get a graph of volume $n^{o(1)}$ on which we can run any static algorithm.

We describe below first the data structure that the algorithm maintains during its computation on the static graph. We then describe the algorithm, called \AAA algorithm, which basically instantiates the data structures it uses to compute the partition $\P_2$ and then calls itself recursively on the contracted graph.
The main work in initializing and updating these data structures is done in the subroutine \UP, which is presented afterwards. Note that we maintain one \AAA data structure for each $\lmax$, i.e., for each $i$.

Recall that $\alpha = \frac{1}{\poly \log n}$, $\phi = 2^{-\Theta(\log^{3/4} n)}$ and that $\delta \le \frac 1 {2\lmax}$ with $ \delta \le 0.04$, which holds by our above choice of $c>0$, $\delta = 2^{O(\log^{3/4 -c } n)}\le 0.04$. We set $\lmin=2^{\Theta (\log^{3/4-c}n)}$ and $\lmax\le 1.2 \lmin$, $H=38^h = 2^{-O(\log^{1/2}n)}$, $\phi =\frac {\phi'} {38^h}= 2^{-\Theta (\log^{3/4}n)}$, $\rho= 2^{\Theta(\log ^{1/2}n)}$ and $ \alpha=\frac 1 {\poly \log n}.$

Note that $\frac \alpha \phi \ge 1$. We also set $\gamma = 6$.

\begin{algo}[Static Algorithm]\label{alg:StaticAlgorithm}
The static algorithm works as follows:
\item Input:
\begin{enumerate}
    \item An unweighted undirected graph $G$ with possibly parallel edges.
    \item An integer $\lmin$ such that the minimum \proper cut of $G$ is no less than $\lmin$.
    \item An integer $\lmax$ with $\lmax \le 1.2 \lmin = 2^{\Theta (\log^{3/4-c}n)}$ with $c>0$.
\end{enumerate}

\item \AAA Data Structure:

If $G$ has volume $\Omega(\frac \lmax \phi)$:
\begin{enumerate}
   \item \label{staticapprox} An approximation value $ \beta = 8$.
   \item \label{staticexpansion} An expansion parameter $\phi=\frac {\phi'} {38^h}= = 2^{-\Theta (\log ^{3/4} n)}$
    \item \label {staticpartition} A two-level partition $\P_1, \P_2$ of $G$ defined by a labeling of the edges as ``intercluster", ``intracluster'' or ``fragmented''.
    \item \label{staticlabels} A labeling of each cluster $C$ as either ``well-connected'', ``poorly-connected'', or ``fragmented''.
    \item A labeling of each node as ``checked'' or ``unchecked''.
    \item \label{staticdatastructure} An instance of \Cref{alg:datastructure} which maintains the mirror graph $G_{\P_1}$ and the contracted graph $G/{\P_2}$. (See \Cref{thm:datastructure}).
    \item \label{staticlocakcut} An instance of LocalKCut (\Cref{alg:detlocalkcut}) on $G$ where ``intercluster'' edges have been removed, and $\lmax, \nu=4\frac \lmax \phi, \beta$. (See \Cref{thm:deterministic-dynamic})
    \item \label{staticmirrors} An instance of \Cref{alg:maintainmirrors}, the Mirror Cuts Data Structure, on $G_{\P_1}$ and $\lmax, \nu=4\frac \lmax \phi, \beta$. (See \Cref{thm:mirrorcuts})
    \item \label{staticboringrecursivecall} A recursive call of the dynamic algorithm on $G$ where ``intercluster'' edges have been removed, with $\lmin = \ceil{1.2^i}$ and $\lmax = \floor{1.2^{i+1}}$ for all $i\in \N_0$ satisfying $1.2^{i} \le  \frac \lmax 8$.
    \item \label{staticrecursivecontract} A recursive call of \Cref{alg:StaticAlgorithm} on the contracted graph $G/{\P_2}$.
\end{enumerate}

\item Rules for cluster classifications:
\begin{enumerate}
    \item Initially, every cluster is ``well-connected''. 
    \item Any cluster $C$ that satisfies $\boundary C \le 3\lmax$ becomes ``poorly-connected''
    \item In $\P_2$, any cluster that is created by  the fragmenting algorithm (\Cref{alg:fragmenting}) becomes ``fragmented'' (note that a ``poorly-connected'' cluster that goes through the fragmenting algorithm untouched still becomes ``fragmented''). This overrides other rules.
\end{enumerate}
   
\item Processing:
\begin{enumerate}[label=(\arabic*),noitemsep]
    \item  If the graph has volume $O(\frac \lmax \phi)$, run a connectivity algorithm then a static minimum cut algorithm (\cite{DBLP:conf/soda/HenzingerLRW24}) on all connected components. If the graph has no edges, output ``$\lambda > \lmax$''. In any case, stop this call to the algorithm.

    \item \label{step:static1} Compute an $(\alpha, \phi')$-boundary-linked expander decomposition of $G$, using the algorithm in~\cite{expanderhierarchy}.

    \item \label{staticupdatepartition} Run \textsc{UpdatePartition} (given below) to instantiate the \AAA  data structure.

    \item \label{step:static2} Run \Cref{alg:decomposingexpanders} on each expander, which computes a pre-cluster decomposition $\P_1$.
    Every time \Cref{alg:decomposingexpanders} splits a cluster it calls \UP 
    to guarantee that the \AAA data structure reflects the current $\P_1$.

    \item \label{step:staticfragmenting} Run the fragmenting algorithm (\Cref{alg:fragmenting}) on each ``poorly connected'' cluster from $\P_1$ to get a cluster decomposition $\P_2$. Update the data structure accordingly.

     \item \label{step:static4}Call \Cref{alg:StaticAlgorithm} recursively on the contracted graph $G/{\P_2}$. 

    \item \label{step:static6}
    Return the smallest cut value found among the recursive call in Step~\ref{step:static4} and the request on the Mirror Cuts Data Structure (\Cref{alg:maintainmirrors}), if this value is smaller than $\lmax$.
    Store the pointer  to the cut that achieved this value (if the algorithm is required to return cuts).
    
       If either no cut is found or the smallest found cut is strictly larger than $\lmax$, then report ``$\lambda > \lmax$''. 
\end{enumerate}
\item \textsc{UpdatePartition:}
\begin{enumerate}
    \item If not yet instantiated, instantiate $\P_1$ by setting the label of every edge to ``intracluster".
    \item Update $\P_1$, by changing the labels of the intercluster edges to ``intercluster''. Update or instantiate  Items~\ref{staticlabels} and~\ref{staticdatastructure}. Mark the vertices adjacent to a newly ``intercluster'' edge as ``unchecked''.
    \item Update or instantiate the recursive calls of the \AAA algorithm that maintains  \Cref{staticboringrecursivecall} of the \AAA data structure in increasing order of $i$, i.e. for all suitable $i$ starting from 0.
    
    Stop this call to \textsc{UpdatePartition} whenever a recursive call returns a cut, cut the partition along this cut and run \textsc{UpdatePartition} from scratch.
    \item Update or instantiate every other structure in the \AAA data structure.
\end{enumerate}

\end{algo}

\subsection{Correctness of Algorithm \ref{alg:StaticAlgorithm}}

First, let's argue about the base case.
For a graph with volume $O(\frac \lmax \phi)$, the result holds by correctness of the algorithms we call. This trivially satisfies the requirements of our algorithm.

Let us now argue in the general case.
Note that if $\lambda > \lmax$, the algorithm cannot find a cut smaller than $\lmax$ (as such a cut does not exist), and will report that $\lambda > \lmax$.

If $\lambda \le \lmax$, then we need to show that for any minimum proper cut $S$ in $G$, we can find a $(1+2\delta)$-approximate mincut $S'$  in $G$ with either the calls to the Mirror Cuts Data Structure (\Cref{alg:maintainmirrors}) in Step~\ref{staticmirrors} or the recursive call of Step~\ref{step:static4}. For this we rely on the following lemma.

\begin{lemma}
    In \Cref{alg:StaticAlgorithm}, the contracted graph of~\Cref{step:static4} and the mirror clusters of \Cref{staticmirrors} preserve the cut values of the corresponding non-contracted cuts of $G$ and given a cut in one of them, one can reconstruct the corresponding cut in $G$.
\end{lemma}

\begin{proof}
    This is immediate as they are contractions of the original graph, and thus preserve cuts.
\end{proof}

As we prove in  \Cref{cor:3.15} in Section~\ref{subsec:clusterdecomposition}, the graph is decomposed into a valid cluster pre-decomposition $\P_1$.
As we show in Section~\ref{sec:mirrorcuts}, \Cref{alg:maintainmirrors} in Item~\ref{staticmirrors} of our data structure correctly finds all local cuts
in the mirror graph $G_{\P_1}$ at this level of the cluster hierarchy. 
Step~\ref{step:static4} of the \AAA algorithm calls itself recursively, which ensures the full cluster hierarchy is built until only one or two nodes are left. In the case where $\lmin\le \lambda \le \lmax$, by \Cref{prop:clusterhierarchy}, the minimum proper cut can be approximated by a local cut in a mirror graph or in the most-collapsed graph of the cluster hierarchy. Since all such local cuts are computed in Item~\ref{staticmirrors}, and their minimum is returned in Step~\ref{step:static6}, this completes the correctness proof if $\lambda \le \lmax$.
Note that \AAA algorithm only returns a value in Step~\ref{step:static6} and that it only returns the value of an existing cut in $G$.
In the case where $\lambda > \lmax$, since the algorithm only returns the value of an existing  cut, it either outputs a value strictly larger than $\lmax$, or reports that $\lambda > \lmax$. This concludes correctness.

\subsection{Decomposing Expanders}\label{subsec:clusterdecomposition}
In this subsection, we present an algorithm, called \Cref{alg:decomposingexpanders}, that implements Step~\ref{step:static2} in Algorithm~\ref{alg:StaticAlgorithm}. 
It takes as input an expander decomposition, and outputs a decomposition such that 
every cluster $C$ with $\boundary C \ge 3\lmax$ contains no $(1-\delta)$-boundary-sparse cut $S$  with $\boundary_G(S) \le \lmax$ and $\vol(S) \le \lmax/\phi$.
It also maintains or instantiates an instance of the Mirror Cuts Data Structure (\Cref{alg:maintainmirrors}) for $G_{\P_1}$ by calling \UP when needed.

This algorithm is very similar to the one in~\cite{DBLP:conf/soda/El-HayekH025}, but works with a different data structure, described in \Cref{sec:datastructure}. 
We include the new analysis here for completeness.

Algorithm~\ref{alg:decomposingexpanders} takes as input an expander decomposition, initializes the data structures, and then repeatedly calls Algorithm~\ref{alg:subroutine} on each cluster until no cluster with an unchecked vertex exists. Algorithm~\ref{alg:subroutine} performs the actual cut finding and graph cutting.

Each node stores a boolean value, which represents whether it  is \emph{checked} or \emph{unchecked}. Intuitively, node $v$ being checked means that LocalKCut has already been executed at $v$ and the cut-size of the returned cuts has already been taken into account, while being unchecked means that LocalKCut needs to be executed at $v$.
Algorithm~\ref{alg:decomposingexpanders} runs as subroutine Algorithm~\ref{alg:subroutine}  that takes a cluster decomposition, a cluster $C$ of that decomposition together with a set of unchecked vertices in $C$ as input and
(a) either cuts the cluster along a $(1-\delta)$-boundary sparse cut or (b) certifies that no cut is $(1-\delta)$-boundary sparse.
Algorithm~\ref{alg:decomposingexpanders} repeats this procedure until no cluster with minimum cut at least $\lmin$ containing a $(1-\delta)$-boundary-sparse cut exists. 
Our algorithm will maintain the following invariant:

\begin{restatable}{invariant}{cuttingisgood}\label{inv:cuttingisgood}
For every cluster $C$ 
    and every cut $S \subsetneq C$ such that $S$ is $(1-\delta)$-boundary-sparse in $C$, $\vol (S) \le \frac \lmax \phi$, and $\lmin \le w(S, C\setminus S) \le \lmax$, there exists a node $v \in S$ such that $v$ is unchecked.
\end{restatable}

The goal of the algorithm is not to  miss any $(1-\delta)$-boundary-sparse cut. 
We use procedure LocalKCut to find such cuts. \emph{Checking node $v$} refers to executing LocalKCut with $v$ as starting point.
As we will show, it suffices to make sure that for every cut that might be boundary-sparse there is at least one node in the cut on which the subroutine LocalKCut is requested to guarantee that the cut is found.

Note that if a set $S \subseteq C$ has no edges leaving $C$, then $w(S, V \setminus C) = 0 \le w(S, C\setminus S)$ and, thus, $S$ cannot be $(1-\delta)$-boundary-sparse in $C$. Thus, vertices in $C$ that are not incident to a boundary edge of $C$ do not need to be checked in $C$, i.e., we do not need to mark them as unchecked. Hence, only vertices that are incident to a boundary edge are marked as unchecked, all others are marked as checked.

We call the decomposition of the vertex set after the algorithm terminates the \emph{pre-cluster decomposition}.

\begin{algo}[Decomposing Expanders]\label{alg:decomposingexpanders}
    Input: 
    \begin{itemize}
        \item an $(\alpha, \phi')$-boundary-linked expander decomposition $V_1, \dots, V_\ell$, specified by edge labels.
        \item Access to a LocalKCut instance on the graph where intercluster edges have been removed.
        \item The parameters $\lmax$, $\lmin$, $\nu$, $\delta$.
    \end{itemize}

\begin{enumerate}[noitemsep]
    \item set $\C \leftarrow \{V_1, \dots, V_\ell\}$.
    \item \label{step:de2}   Mark every node incident to an inter-expander edge as unchecked, and all other nodes as checked.
\item While there exists in $\C$ a cluster $C$ satisfying $\boundary C \ge 3\lmax$ and containing an unchecked vertex, we run the Find and Cut subroutine (\Cref{alg:subroutine}) on $C$, which might change $\C$.
\end{enumerate}
\end{algo}

\begin{algo}[Find and Cut Subroutine]\label{alg:subroutine}
Input: 
\begin{itemize}
    \item a pointer to a graph $G$, with a partition $\P_1$ of $G$ defined by a labeling of the edges.
    \item a pointer to a set $C\in \P_1$
    \item a set of unchecked vertices satisfying \Cref{inv:cuttingisgood}
    \item an instance of LocalKCut on $G$ where intercluster edges have been removed.
\end{itemize}

    \item\label{step:fc2}  For every unchecked vertex $v$ in $C$:
\begin{enumerate}[noitemsep]
    \item \label{step:fcl1} Make a request to LocalKCut  on $v$. Let $\mathcal{S}$ be the set of outputs from all LocalKCuts.
    \item For every $S \in \mathcal{S}$:

    \label{step:fcl2b}
    Check if $S$ is $(1-\delta)$-boundary-sparse.
    If it is the case, cut the cluster along $S$,
    and run \textsc{UpdatePartition} thereby removing $C$ and creating two new clusters $S$ and $C \setminus S$ in $\P_1$, and mark any node incident to a newly cut edge as unchecked. The node $v$ remains unchecked. End the current call to \Cref{alg:subroutine}.
    \item \label{step:fcl3} 
    Mark the node $v$ as checked and move on to the next unchecked vertex in $C$.

\end{enumerate}
\end{algo}

\subsubsection{Correctness}

We need to show that 
 Invariant~\ref{inv:cuttingisgood} holds at the termination of \Cref{alg:decomposingexpanders}.
For that, we start by showing that the invariant holds before the while loop in Algorithm~\ref{alg:decomposingexpanders}, and that each step of Algorithm~\ref{alg:subroutine} which is called in the while loop of Algorithm~\ref{alg:decomposingexpanders} maintains   Invariant~\ref{inv:cuttingisgood}. 
Thus, the invariant holds after Algorithm~\ref{alg:decomposingexpanders} has terminated. As there are no unchecked nodes left in clusters $C$ with $\boundary C \le 3\lmax$ at that point in time,  it follows that every cluster $C$ with $\boundary C \le 3\lmax$ contains no $(1-\delta)$-boundary-sparse cut $S$  with $\boundary_G(S) \le \lmax$ and $\vol(S) \le \lmax/\phi$.

\begin{lemma}\label{lem:3.12}
    Right after Step~\ref{step:de2} in \Cref{alg:decomposingexpanders},  \Cref{inv:cuttingisgood} holds.
\end{lemma}

\begin{proof}
    Right  after Step~\ref{step:de2} every cluster is an expander in the expander decomposition. Thus,
    any cut $S$ in a cluster $C$  can only be $(1-\delta)$-boundary-sparse if it contains a node $v$ that is incident to an inter-expander edge, that is, if $w(S, V\setminus V_i) \neq 0$ as we have $(1-\delta) w(S, V\setminus V_i) > w(S, V_i\setminus S) \ge 0$, by the definition of boundary-sparsity.
     As all the nodes incident to all inter-expander edges were marked as unchecked in Step~\ref{step:de2}, \Cref{inv:cuttingisgood} holds.
\end{proof}

\begin{lemma}\label{lem:bsgood}
    Cutting a cluster $C$ along a $(1-\delta)$-boundary-sparse cut $C'$, then marking every node adjacent to an edge in $E(C', C\setminus C')$ as ``unchecked'' maintains \Cref{inv:cuttingisgood}.
\end{lemma}
\begin{proof}
     To show \Cref{inv:cuttingisgood} we need to show that 
every $(1-\delta)$-boundary sparse cut $S$ in $C'$  with 
 $\vol(S) \le \frac \lmax \phi$ and $w(S, C'\setminus S) \le \lmax$ contains an unchecked vertex. 
 By symmetry the same claim follows for $C \setminus C'$.
 
    Assume by contradiction that there  exists a $(1-\delta)$-boundary sparse cut $S$  with 
 $\vol(S) \le \frac \lmax \phi$ and $w(S, C'\setminus S) \le \lmax$ such that no vertex in  $S$ is unchecked
   after marking every node adjacent to an edge in $E(C', C\setminus C')$ as ``unchecked'' , i.e. every vertex in $S$ is checked. As cutting does not change the marking of the vertices, this implies that (a) every vertex in $S$ was also checked in $C$, and (b) no edge incident  to $S$ was cut, i.e., became a new inter-cluster edge as the endpoints of such edges are marked ``unchecked'' in Step 2 of \Cref{alg:subroutine}. Specifically, there exists no edge from a node in $S$ to a node in $C\setminus C'$.
    To achieve a contradiction we will show that $S$ must have been (i) a $(1-\delta)$-boundary-sparse cut in $C$, (ii) that $S$ has volume at most $\lmax/\phi$, and (iii)  $\lmin \le w(S, C\setminus S) \le \lmax$. As, by assumption, \Cref{inv:cuttingisgood} was satisfied in $C$ before Step~\ref{step:fcl2b}, it follows that $S$ must have contained an unchecked vertex in $C$, leading to a contradiction. 
    
    It remains to show (i) to (iii).
    We first show (i). Recall that $C' \subset C$.
    Then we know that $w(S, C'\setminus S) = w(S, C\setminus S)$, as otherwise there would be an edge from a node in $S$ to a node in $C\setminus C'$. Thus $w(S, C \setminus S) = w(S, C' \setminus S)$ and $w(S, V\setminus C') = w(S, V\setminus C)$.
    Since $C'$ is a $(1-\delta)$-boundary-sparse cut in $C$, we have that $w(C', C\setminus C') \le (1-\delta) w(C\setminus C', V\setminus C) < w(C\setminus C', V\setminus C)$. 
    Also since $S$ is $(1-\delta)$-boundary sparse in $C'$, it holds that  $w(S, C'\setminus S)\le(1-\delta) w(C'\setminus S, V\setminus C')$.
    
    Therefore, 
    \begin{align*}
        w(S, C\setminus S) = w(S, C'\setminus S)&\le(1-\delta) w(C'\setminus S, V\setminus C')\\
        & \le (1-\delta)\PAR{w(C'\setminus S, V\setminus C) + w(C'\setminus S, C\setminus C')}\\
        &\le(1-\delta)\PAR{w(C'\setminus S, V\setminus C) + w(C', C\setminus C') } \\
        &< (1-\delta)\PAR{w(C'\setminus S, V\setminus C) + w(C\setminus C', V\setminus C)}\\
        &=(1-\delta)w(C\setminus S, V\setminus C) 
    \end{align*}
    Furthermore, since $S$ is $(1-\delta)$-boundary sparse in $C'$, it also holds that
    \begin{align*}
        w(S, C\setminus S) = w(S, C'\setminus S)&\le(1-\delta) w(S, V\setminus C') = w(S, V\setminus C)
    \end{align*}
    Thus $S$ is $(1-\delta)$-boundary sparse in $C$.

    We next show (ii). As $S \subset C' \subset C$, it follows that the volume of $S$ is the same, and thus, it is at most $\lmax/\phi$. 

    Finally we show (iii). As there exists no edge from a node in $S$ to a node in $C \setminus C'$ it follows that $w(S, C'\setminus S) = w(S, C \setminus S)$ and, thus,
    $w(S, C\setminus S) =w(S, C'\setminus S)\le \lmax$.

    Thus (i) to (iii) hold, completing the proof of the lemma.
\end{proof}

\begin{lemma}\label{lem:updateissparse}
    If \textsc{UpdatePartition} cuts along a cut $C'$ in a cluster $C$, $C'$ is $(1-\delta)$-boundary sparse in $C$.
\end{lemma}

\begin{proof}
    \textsc{UpdatePartition} cuts along cuts that are returned by a recursive call from \Cref{staticboringrecursivecall}, that is, where the cut-size is at most $\frac \lmax 8$. 
    More specifically, in a cluster $C$, it cuts along a cut $C'$ which satisfies $w(C', C\setminus C') \le \frac \lmax 8$. However, $w(C', C\setminus C') + w(C', V\setminus C) \ge \lmin$ by the assumptions of \Cref{thm:static}, and therefore, $w(C', V\setminus C) \ge \lmin - \frac \lmax 8 \ge 0.7\lmax$ (where we use that $\lmax \le 1.2\lmin$).
    Similarly, $w(C\setminus C', V\setminus C) \ge 0.7 \lmax$. 
    This proves that $C'$ is $(1-\delta)$-boundary-sparse in $C$ for $\delta \le 0.82$, which holds for our choice of $\delta$.
\end{proof}

\begin{lemma}\label{lem:3.13}
    Each execution of~\Cref{alg:subroutine}
    maintains \Cref{inv:cuttingisgood} if  there are no cuts of cut-size at most $\frac \lmax 8$ in the cluster $C$ given as second parameter.
\end{lemma}

\begin{proof}

\textbf{Step~\ref{step:fcl1}:}
Step~\ref{step:fcl1} of \Cref{alg:subroutine} neither affects the marking of vertices nor the cluster decomposition and, thus, does not affect \Cref{inv:cuttingisgood}.

\textbf{Step~\ref{step:fcl2b}:}
We will next show that the invariant is maintained in Step~\ref{step:fcl2b}. 
In Step~\ref{step:fcl2b} 
a cluster $C$ is cut along a  $(1-\delta)$-boundary-sparse cut $C' \subsetneq C$ into a new cluster $C'$ and into $C \setminus C'$ and the endpoints of new inter-cluster edges are marked as unchecked, whether this happens in \Cref{alg:subroutine} itself or the call to \textsc{UpdatePartition} (\Cref{lem:updateissparse}), which by \Cref{lem:bsgood} maintains the invariant.

\textbf{Step~\ref{step:fcl3}:}
It remains to show that the invariant is maintained in Step~\ref{step:fcl3}. If this step is executed, (a) the status of only one node can change, namely the one of $v$ and (b) the execution of the for-loop for $v$ modified neither the cluster decomposition nor the status of any other unchecked node. Thus~\Cref{inv:cuttingisgood} continues to hold for all clusters that do not contain $v$. We need to argue that it also holds for the (unique) cluster $C_v$ that contains $v$.
For that, we will start by making the following claim, which is crucial to ensure that LocalKCut finds all the cuts necessary for our analysis:

\begin{claim}\label{claim:4.1}
    For any cluster $C$ and any subset $\varnothing \subsetneq S\subsetneq C$, at any point outside of \textsc{UpdatePartition}, we have that $\boundary_CS \ge \frac \lmax 8$.
\end{claim}

\begin{proof}
    Note that apart from inside \UP, the only two steps in the algorithm where $\P_1$ changes is in Steps~\ref{step:static1} and~\ref{step:static2}. 
    Any change is immediately followed by a call to \UP, which then runs until $\boundary_CS \ge \frac \lmax 8$ for all $C\in \P_1$, $S\subseteq C$, hence the claim.
\end{proof}

We now continue arguing about Step~\ref{step:fcl3}.
Assume by contradiction that a cut $S$ exists in  cluster $C_v$ at the end of Step~\ref{step:fcl3}
such that $S$
(a) is $(1-\delta)$-boundary sparse in $C_v$,
(b) has volume at most $\lmax/\phi$,
(c) satisfies $w(S, C_v\setminus S) \le \lmax$, and (d)
 contains only checked vertices. As Invariant~\ref{inv:cuttingisgood} held at the beginning of 
Step~\ref{step:fcl3}, $S$ must have had an unchecked vertex then and as $v$ is the only vertex whose status changed since then, $v$ must belong to $S$.

By \Cref{claim:4.1}, in $C$ there are no cuts of cut-size at most $\frac \lmax 8$,  and we have that $\boundary_CS \le \lmax$.
Hence $S$ is a $8$-approximate minimum cut in $C$, and a call to LocalKCut with vertex $v$ and $\beta=8$ is guaranteed to find $S$.
and that LocalKCut is run with approximation parameter $\beta = 8$ by~\Cref{thm:deterministic-dynamic}.
But then, as $S$ is found, \Cref{alg:subroutine} would have cut it instead in Step 2 instead of proceeding to Step 3 and of unchecking $v$, and thus we have a contradiction.
\end{proof}

Now we can finish the correctness proof.
\begin{lemma}\label{cor:3.15}
 It holds that
 (a) at the termination of \Cref{alg:decomposingexpanders} \Cref{inv:cuttingisgood} holds; 
 (b) the partition $\P_1$ computed by \Cref{alg:decomposingexpanders} satisfies that no  cluster $C$ with $\boundary C \ge 3\lmax$ contains a $(1-\delta)$-boundary-sparse cut of cut-size at most $\lmax$ and volume at most $\frac \lmax \phi$, and thus $\P_1$ is a pre-cluster decomposition.
\end{lemma}

\begin{proof}

    Since we only uncheck a vertex if it is contained in no cuts that are $(1-\delta)$-boundary-sparse with cut value at most $\lmax$ and volume at most $\frac \lmax \phi$, this lemma relies on \Cref{inv:cuttingisgood}.
    By \Cref{lem:3.12}, the invariant holds at the beginning of the while loop in \Cref{alg:decomposingexpanders}.

     We now have to show that the invariant holds until the end of the while loop.
     This is now immediate as the while loop is simply a call to \Cref{alg:subroutine}, and \Cref{alg:subroutine} does not affect the invariant by \Cref{lem:3.13}.
     
       The fact that the while loops in \Cref{alg:decomposingexpanders} ends when each cluster $C$ either satisfies $\boundary C <3 \lmax$ or contains no unchecked vertex, together with \Cref{inv:cuttingisgood}, concludes the proof.
      \end{proof}

\subsubsection{Bounding the number of inter-cluster edges}
This subsection is dedicated to the analysis of the number of inter-cluster edges in our final cluster decomposition, i.e., it proves \Cref{lem:numedges}.

\begin{proposition}
    \label{lem:numedges}
    Let $M$ be the number of inter-expander edges after the expander decomposition in \Cref{alg:StaticAlgorithm}.
    Then there are $\tilde O\PAR{\frac M {\delta^3}}$ many inter-cluster edges in the cluster decomposition after the execution of Step~\ref{step:staticfragmenting} of \Cref{alg:decomposingexpanders}.
\end{proposition}

    Consider the following rooted forest $F$ in which each node represents a cluster of $G$ at \emph{some point} during the computation: the roots are the expanders output by the expander decomposition.
    For each node representing a cluster, its two children (if any) are the clusters formed by splitting the corresponding cluster. By slight abuse of notation we do not distinguish between clusters of the cluster decomposition and nodes in the forest.

    Note that a cluster is always split along a $(1-\delta)$-boundary-sparse cut, whether it is in \textsc{UpdatePartition}, \Cref{alg:subroutine} or \Cref{alg:fragmenting}.

    We first show that the boundary-size of a node is always smaller than the one of its parent by at least $\delta {\lmin}/2$.

\begin{lemma}\label{lem:boundariessmaller}
    Let $C$ be a node in $F$ and let $S, C\setminus S$ be its children.
    Then we have that $\boundary C \ge \max\{\boundary S, \boundary (C\setminus S)\}+(\delta \lambda_{\min})/2$.
\end{lemma}

\begin{proof}
    Since $S$ is a $(1-\delta)$-boundary-sparse cut, we have that $$w(S, C\setminus S) \le (1-\delta)\min\{ w(S, V \setminus C), w(C\setminus S, V\setminus C)\},$$
    which implies that $w( C\setminus S, V\setminus C) > w( C\setminus S, S)$.
    As it holds that $w( C\setminus S, V\setminus C) + w( C\setminus S, S) = \boundary (C\setminus S) \ge \lambda_{\min}$, it follows that $w( C\setminus S, V\setminus C) \ge \frac {\lambda_{\min}} 2 $.
    
    The boundary-sparseness of $S$ also implies that $$\boundary S = w(S, V\setminus C) + w(S, C\setminus S) \le w(S, V\setminus C) + (1-\delta) w( C\setminus S, V\setminus C) = \boundary C - \delta w( C\setminus S, V\setminus C)$$ i.e., $$\boundary C \ge \boundary S +  \delta w( C\setminus S, V\setminus C) \ge \boundary S + \delta \lmin/2.$$

    The same bound holds by symmetry for $\boundary (C\setminus S)$ and the result follows.    
\end{proof}
Nonte that the lemma implies that $\boundary C > \boundary S$ and $\boundary(C \setminus S)$.
Thus, along any path in $F$ from a root to a leaf the boundary-size of the clusters strictly decreases. Specifically, if a cluster has boundary-size less than some value, let's say $3\lmax$, then all its descendants in $F$ has boundary-size less than $3 \lmax$.
We now look only at the subforest $F'$ of $F$ created from $F$ by deleting from $F$ all internal nodes with boundary at most $3\lmax$. More formally, $F'$ is defined as follows: All roots of $F$ are also in $F'$, and moreover, for every $C\in F'$ with $\boundary C \ge 3\lmax$, all its children in $F$ are also in $F'$. However, for $C \in F'$ with $\boundary C < 3 \lmax$, their children do not belong to $F'$, i.e., $C$ is a leaf in $F'$.
As discussed, $F \setminus F'$ contains no node of boundary larger than $3\lmax$.

\begin{lemma}
    The total number of inter-cluster edges of the leaves of $F'$ is at most $O(\frac M \delta)$, where $M$ is the number of inter-expander edges output by the expander decomposition.
\end{lemma}

\begin{proof}
We count the number of inter-cluster edges by bounding the sum of the boundary-sizes for every cluster that is a descendant of a root cluster $V_i$. 
Any inter-cluster edge that is created at some descendant of $V_i$  is also an inter-cluster edge at a leaf descendant of $V_i$ in $F'$ since clusters are never merged. Thus, it suffices to bound the sum of the boundary-sizes of the leaf descendants of $V_i$ in $F'$. Let $\C_f$  contain exactly these clusters.
We bound this sum by splitting $\C_f$ into two sums that we will analyze separately:
 \begin{align*}
        \sum_{C \in \C_f} \boundary C &= \sum_{\substack{C \in \C_f, \boundary C \le 2.2\lmax}} \boundary C+\sum_{\substack{C \in \C_f, \boundary C > 2.2\lmax}} \boundary C
        \intertext{and we define} S_1 :&= \sum_{\substack{C \in \C_f, \boundary C \le 2.2\lmax}} \boundary C \intertext{ and } S_2 :&= \sum_{\substack{C \in \C_f, \boundary C > 2.2\lmax}} \boundary C.
\end{align*}
To proceed we  will analyze the potential function $\Phi(\C) := \sum_{C \in \C} \Phi(C)$ were $\Phi(C) = \max\{ 0 , \boundary C - 2.1 \lmax\}$. 
Right after the computation of the expander decomposition the cluster decomposition of $V_i$ only consists of  $V_i$ itself (that is, an expander from the expander decomposition). While computing the cluster decomposition, the algorithm might partition a descendant  $C$ of $V_i$, replacing it by two new clusters $C_1$ and $C_2$ and making $C_1$ and $C_2$ its children in $F'$. We show next that the sum of the potentials of the two children clusters is at least $(\delta \lmin)/2$ smaller than the potential of $C$.
To do so we analyze three cases: 
    \begin{enumerate}
        \item If both $C_1$ and $C_2$ satisfy $\boundary C_1, \boundary C_2 \ge 2.1\lmax$, then the total potential changes by:
        \begin{multline*}
        \Phi(C_1) + \Phi (C_2) - \Phi (C) = \boundary C_1 - 2.1 \lmax + \boundary C_2 - 2.1 \lmax -\boundary C + 2.1\lmax \\= -2.1 \lmax + 2 w(C_1, C_2) \le -2.1 \lmax + 2 \lmax \le -0.1 \lmax
        \end{multline*}

        \item If wlog $C_1$ and $C_2$ satisfy $\boundary C_1\le 2.1\lmax \le\boundary C_2 $, then the total potential changes by:
        \begin{multline*}
        \Phi(C_1) + \Phi (C_2) - \Phi (C) = 0 + \boundary C_2 - 2.1 \lmax -\boundary C + 2.1\lmax = \boundary C_2 - \boundary C \le -(\delta \lambda_{\min})/2
        \end{multline*}
        where the last inequality follows from \Cref{lem:boundariessmaller}.
        \item If both $C_1$ and $C_2$ satisfy $\boundary C_1, \boundary C_2 \le 2.1\lmax$, then the total potential changes by:
        $$
        \Phi(C_1) + \Phi (C_2) - \Phi (C) = 0 + 0 -\boundary C + 2.1\lmax \le -0.9 \lmax
        $$
    \end{enumerate}

    It follows that the potential drops by at least $(\delta \lambda _{\min})/2$ each time when a cluster is split into two smaller clusters. 
    
    We now bound the number of clusters in $\C_f$ as follows: Since the potential drops by $\delta \lambda_{\min}/2$ each time a cluster is split, and splitting a cluster only increases the number of clusters by one, we know that the total number of clusters in $\C_f$ is at most $\frac {2\Phi(\{V_i\})} {\delta \lmin} \le \frac{2\boundary V_i}{\delta \lmin}$.
    As every cluster contributing to $S_1$ contributes at most $2.2\lmax$, it follows that $S_1$ is at most $\lmax \frac {4.4 \boundary V_i} {\delta \lmin} = O\PAR{\frac{\boundary V_i \lmax}{\delta \lmin}}.$
   
    For the second sum, $S_2$, note that if $\boundary C \ge 2.2\lmax$, then $\boundary C = O(\Phi(C))$. 
    Recall that the sum of potentials of the two children of a node in $F'$ is less than the potential of the parent. It follows by induction that for any set of descendants of $V_i$ such that none is an ancestor of the other, the sum of their potentials is less than the potential of $V_i$. As the clusters in $C_f$ fulfill this condition, it follows that
    $S_2 = \sum_{\substack{C \in \C_f\\\boundary C > 2.2\lmax}} O(\Phi( C)) = O(\Phi(\{V_i\})) = O\PAR{\frac{\boundary V_i \lmax}{\delta \lmin}}$.    

    Summing over all expanders, and using that $\lmax = O(\lmin)$, we get the result.
\end{proof}

We now look at the forest $F''=(F\setminus F')\union \mathrm{leaves}(F')$. 
The leaves of $F'$ are the roots of $F''$.
By the lemma above, we know that the sum of the boundaries of the roots of $F''$ is at most $O\PAR{\frac M \delta}$.
We moreover know that any cluster $C$ with children in $F''$ satisfies $\boundary C < 3\lmax$, and thus the decomposition is made by fragmenting (\Cref{alg:fragmenting}). This gives the following lemma, which is a direct consequence of \Cref{thm:fragmenting}:

\begin{lemma}\label{lem:inter-clusteredges}
    The sum of the boundaries of the leaves of $F''$ is at most $\tilde O\PAR{\frac M {\delta^3}}$, where $M$ is the number of inter-expander edges output by the expander decomposition.
\end{lemma}

\subsubsection{Running Time Analysis}
Now we are ready to analyze the running time of the static algorithm.
\begin{definition}
    Let $u$ be a vertex. 
    We say that $u$ is \emph{responsible} for a cluster $S$ if, in \Cref{alg:subroutine}, $S$ was cut from its parent cluster after being found by LocalKCut requested on vertex $u$.
\end{definition}

\begin{lemma}\label{lem:nottoomanykargers}
    Let $u$ be a vertex. Vertex $u$ is responsible for at most $O(\frac 1 \phi)$ many clusters in $F$.
\end{lemma} 

\begin{proof}
    The first cut $u$ is responsible for will reduce the volume of the cluster $u$ is in to at most $\frac \lmax \phi$.
    Then every further cut will reduce this volume by at least $\lmin$ as it cuts a at least $\lmin$ edges out of the cluster.
    As $\lmax/\lmin = O(1)$ the lemma follows.
\end{proof}

This lemma is important, as after checking a node, it is not guaranteed to become unchecked, as if a boundary-sparse cut of cut-size at least $\lmin$ is found, a cut is made, but the node stays unchecked. 
Hence this bounds how many times a vertex is being checked, that is, how many times we need to run LocalKCut on any individual vertex. It follows that we can bound the total number of call to LocalKCut.

\begin{lemma}\label{lem:nooflocalkcuts}
    The total number of calls to LocalKCut in \Cref{alg:decomposingexpanders} is $\Tilde O\PAR{\frac {M}{\delta^3} \frac 1 {\phi}}$, where $M=O(m \phi)$ is the number of inter-expander edges.
\end{lemma}

\begin{proof}
    A vertex only becomes unchecked when it is incident to an inter-cluster edge. By \Cref{lem:inter-clusteredges} we have $O\PAR{\frac M {\delta^3}}$ inter-cluster edges, and, thus, this bounds the total number of times that a vertex becomes unchecked. From each unchecked vertex $u$
    we run a call to LocalKCut requested on $u$ in \Cref{step:fcl1} of \Cref{alg:subroutine}. 
    As \Cref{lem:nottoomanykargers} shows, after $O(\frac 1 \phi)$ such executions of \Cref{step:fcl1} vertex $u$ must become checked.
    Thus the total number of calls to LocalKCut is as stated in the lemma.
\end{proof}

We remind the reader of the following constants: $c>0$, $ \lmax \le 1.2\lmin = 2^{O(\log ^{3/4-c}n)}$, $\delta \le \frac 1 {2\lmax}$ with $\delta \le  0.04.$, $\phi = 2^{-\Theta (\log^{3/4}n)}$ and $ \alpha=\frac 1 {\poly \log n}.$

Let us now discuss the number of recursive calls. Remember that $G$ is the original graph with $m$ edges. 
    We say that an instance of the algorithm is on level $\ell$ if it is run on a graph $G'$ that is obtained from $G$ by contracting it $\ell$ times in \Cref{staticrecursivecontract}.
    For example, the instance  $A_0$ of original call contracts the graph and calls recursively the instance $A_1$, which calls instance $A_2$ on the same contracted graph with smaller $\lmin$ and $\lmax$ in \Cref{staticboringrecursivecall}. This instance then calls instance $A_3$ on another contracted graph, and $A_3$ calls $A_4$ in \Cref{staticboringrecursivecall}. $A_0$ is on level $0$, $A_1$ and $A_2$ are on level $1$, $A_3$ and $A_4$ are on level $2$.

We also assign label $i$ to each node in that tree where its $\lmax$ parameter is $\floor{1.2^{i+1}}$.

\begin{lemma}\label{lem:norecursivecallse}
    The number of levels of \Cref{alg:StaticAlgorithm} is $O(\log ^{1/4} n)$.
\end{lemma}

\begin{proof}
Let $A$ be the number of edges input to an instance of the algorithm.
Step~\ref{step:static1} of the algorithm creates $M=\tilde O(A\phi)$ many interexpander edges, as proven in~\cite{expanderhierarchy}.
    \Cref{lem:inter-clusteredges} shows that then the number of intercluster edges, which is the number of edges in the input graph of the $\ell+1$-th recursive call is $\tilde O(A \phi \cdot \frac 1 {\delta^3})\le \frac A {2^{a\log ^{3/4}n}}$ for some $a \in \R_+$. %
Hence, the number of edges at the $r$-th level of the cluster decomposition is at most $\frac m {2^{ar \log ^{3/4}n}}$.
This shows that after $r = \frac 1 a \log^{1/4}n $ levels, the number of edges drops to less than $\tilde O(1)$.
\end{proof}

\begin{lemma}\label{lem:countingcalls}
    Let $\alg_i$ be an algorithm indexed on $i$ that calls itself recursively, where $\alg_i$ calls $\alg_j$ each at most once for every $0\le j<i$. 
    Overall, a call to $\alg_i$ runs at most $2^{i}$ many instances of the algorithm.
\end{lemma}

\begin{proof}
    We prove the statement by induction on $i$. The case $i=0$ is straightforward. For the inductive step, assume that a call to $\alg_j$ runs $2^{j}$ many instances of the algorithm for all $j \le i$, and let us prove the statement for the case $i+1$.

    Let us consider a call to $\alg_{i+1}$. 
    We thus have a call to all of $\alg_j$ for $0\le j <i+1$.
    By the induction hypothesis, we have that the total number of instances called recursively is at most $\sum_{0\le j < i+1} 2^j = 2^{i+1}-1$. 
    Adding the instance $\alg_{i+1}$ to that total concludes the proof.
\end{proof}

\begin{lemma}\label{lem:samelevel}
    Let us consider an instance labeled $i$ on level $\ell$. This instance calls overall at most $2^i$ many instances on level $\ell$.
\end{lemma}

\begin{proof}
    Remember that a call on such an instance calls recursively on the same graph with $\lmin = \ceil{1.2^i}$ and $\lmax=\floor{1.2^{i+1}}$ for all $i \in \N_0$ with $1.2^{i} \le \lmax/8$. 
    We can ignore all recursive calls on the contracted graphs, as those will only create calls on higher levels.

    Let us look at the tree of recursive calls, 
    Note that any children labeled $j$ of a node labeled $i$ must satisfy $j+10 \le i$ as we have $1.2^j \le \frac {\floor {1.2^{i+1}}}{8} \le \frac { {1.2^{i+1}}}{8} \le  \frac { {1.2^{i+1}}}{1.2^{11}} $. 
    We relax this condition and note that $j+1 \le i$. Then, by \Cref{lem:countingcalls}, this creates at most $2^i$ many recursive calls.   
\end{proof}

\begin{lemma}\label{lem:norecursivecalls}
    The total number of recursive calls is $n^{o(1)}$.
\end{lemma}

\begin{proof}
    Note that every instance on level $\ell$ calls immediately exactly one instance on level $\ell +1$ in the limit of $\ell = O(\sqrt[4]{\log n})$ many levels, by \Cref{lem:norecursivecallse}.
    Moreover, every call on level $\ell$ and index $i$ generates overall $2^i$ many instances on level $\ell$, by \Cref{lem:samelevel}.

    Therefore, the initial call on level $0$ and label $i=O(\log \lambda)$ creates at most $2^i$ many instances on level $0$. By induction, it is straightforward to see that there are at most $2^{i(\ell +1)}$ many instances on level $\ell$. 
    Plugging in $1.2^i = O(\lmax)$ we have that the total number of recursive calls is equal to $\lmax^{O(\sqrt[4]{\log n})} = 2^{O(\log^{1-c} n)}$.
\end{proof}

\begin{theorem}
    The total running time of \Cref{alg:StaticAlgorithm} is $O(m^{1+o(1)})$.
\end{theorem}

\begin{proof}
    Let us first analyze one level of the recursion: Step~\ref{step:static1} takes $O(m^{1+o(1)})$ as proven in~\cite{expanderhierarchy}.
    It creates $M=m \phi$ inter-expander edges.

Step~\ref{staticupdatepartition} takes $O(m)$ time.

    By \Cref{lem:nooflocalkcuts}, Step~\ref{step:static2} runs  LocalKCut a total of $\Tilde O\PAR{\frac {M}{\delta^3} \frac 1 {\phi}}$ times, each instance taking $n^{o(1)}$, for a total of $O(m^{1+o(1)})$ time overall.

    By \Cref{thm:fragmenting}, Step~\ref{step:staticfragmenting} takes $O(n^{o(1)})$ per intercluster edge at most.
    Steps~\ref{step:static4}, and~\ref{step:static6} take $O(m)$ time (excluding the recursive call).

By \Cref{thm:datastructure}, the rest of the data structure can be maintained in $O(m^{1+o(1)})$, as the number of updates is upper-bounded by $m$ and the volume is $n^{o(1)}$

   By \Cref{lem:norecursivecallse}, the number of levels is $n^{o(1)}$. Hence the claim follows.
\end{proof}

\section{The Dynamic Algorithm}\label{sec:dynamic}
The dynamic algorithm works as follows: at the beginning, the algorithm runs the static algorithm, which sets up the data structure. In particular, on each layer, it provides a pre-cluster decomposition $\P_1$ and a corresponding cluster decomposition $\P_2$.

We then maintain this data structure on each layer, by realizing that each intracluster edge insertion and intercluster edge deletion cannot create any $(1-\delta)$-boundary-sparse cut, while for any other update, only a limited number of calls to LocalKCut suffice to find all newly created $(1-\delta)$-boundary-sparse cuts.

In fact, we will simply mark a few nodes after each update as unchecked, so as to make sure that \Cref{inv:cuttingisgood} holds, and then run the corresponding part of the algorithm again.
The goal is to have after each update a cluster decomposition as in \Cref{lem:clusterdecomposition}.

Remember that we set $c>0$, $ \lmax \le 1.2\lmin = 2^{O (\log^{3/4-c}n)}$, $\delta \le \frac 1 {2\lmax}$ with $\delta \le 0.04.$ $\phi = 2^{-\Theta (\log^{3/4}n)}$ and $ \alpha=\frac 1 {\poly \log n}.$
Note that $\frac \alpha \phi \ge 1$ and that we assume $\delta \le 0.04$.
We further set $\psi=2^{\Theta(\log^{1/2} n)}, h= \Theta(\log^{1/2}m),$ and $\rho = 2^{\Theta(\log^{1/2} n)}$.

For this section, we assume that the data structure we present in \Cref{sec:datastructure} supports our cluster decomposition.

We describe below the main structure of the dynamic algorithm:

\begin{algo}\label{alg:dynamicestimate}
Data Structure: The data structure is the same as the one of the static algorithm (\Cref{alg:StaticAlgorithm}). 

Preprocessing:
\begin{enumerate}[noitemsep]
    \item Run \Cref{alg:StaticAlgorithm}
\end{enumerate}

Processing updates:

\begin{enumerate}[noitemsep]
  \item If the graph has $O(\frac \lmax \phi)$ nodes, run a connectivity algorithm then a static algorithm (\cite{DBLP:conf/soda/HenzingerLRW24}) on all connected components, and return the minimum cut found. If the graph has no edges, output ``$\lambda > \lmax$''. In any case, stop this call to the algorithm.
    \item  Run \Cref{alg:dynclusterdecompose}.
    \Cref{alg:dynclusterdecompose} is the dynamic algorithm that maintains the pre-cluster decomposition and the cluster decomposition
    \item  Maintain the contracted graph where each cluster output by \Cref{alg:dynclusterdecompose} is contracted into a node, and call recursively \Cref{alg:dynamicestimate} on that contracted graph (as described in \Cref{sec:datastructure})
    \item Output the value of the minimum \proper cut found among the cut returned by the recursive call and the mirror clusters, if that value falls in $[\lmin, \lmax]$.  Store pointers to the corresponding cut. If the smallest cut found is larger than $\lmax$ (or no cut is found), report ``$\lambda > \lmax$''.
\end{enumerate}

\end{algo}

Since the cluster decomposition satisfies the conditions of \Cref{lem:clusterdecomposition} at every point, this proves correctness.

The subsections below will prove the following lemma:

\begin{restatable}{lemma}{lemrecourse}\label{lem:clusterdec}
In \Cref{alg:dynclusterdecompose}:
\begin{enumerate}[noitemsep]
    \item  The amortized recourse\footnote{the recourse is the number of modifications to the edges of the contracted graph} per updated edge is $\Tilde O\PAR{\frac \rho {\delta^3} \lmax}$
    \item At every time step, and every level, if $\tilde m$ is the number of edges input on that level, the number of intercluster edges in $\P_2$ is at most $\Tilde{O} \PAR{ \frac {\tilde m\phi} {\delta^3}} $, and thus the recursive call will have as input at most $\Tilde{O} \PAR{ \frac {\tilde m\phi} {\delta^3}} $ edges.
    \item The amortized update time is $n^{o(1)}$ on each level -- excluding recursive calls.
\end{enumerate}
\end{restatable}

Let us now discuss the number of recursive calls. Remember that $G$ is the original graph with $m$ edges. 
As for the static algorithm, we say that an instance of the algorithm is on level $\ell$ if it is run on a graph $G'$ that is obtained from $G$ by contracting it $\ell$ times in \Cref{staticrecursivecontract} of \Cref{alg:StaticAlgorithm}.
    For example, the instance  $A_0$ of original call contracts the graph and calls recursively the instance $A_1$, which calls instance $A_2$ on the same contracted graph with smaller $\lmin$ and $\lmax$ in \Cref{staticboringrecursivecall}. This instance then calls instance $A_3$ on another contracted graph, and $A_3$ calls $A_4$ in \Cref{staticboringrecursivecall}. $A_0$ is on level $0$, $A_1$ and $A_2$ are on level $1$, $A_3$ and $A_4$ are on level $2$.

We also assign label $i$ to each node in that tree where its $\lmax$ parameter is $\floor{1.2^{i+1}}$.

\begin{lemma}\label{lem:norecursivecalls2e}
    The number of levels of \Cref{alg:dynamicestimate} is $O(\log ^{1/4} n)$.
\end{lemma}

\begin{proof}
Let $A$ be the number of edges input to an instance of the algorithm.
By \Cref{lem:clusterdec}, the number of edges on the next level is $\tilde O\PAR{\frac{\tilde m \phi} {\delta^3}}$. %
Hence, the number of edges at the $r$-th level of the cluster decomposition is at most $\frac m {2^{ar \log ^{3/4}n}}$ for some $a \in \R$.
This shows that after $r = \frac 1 a \log^{1/4}n $ levels, the number of edges drops to less than $\tilde O(1)$.
\end{proof}

\begin{lemma}\label{lem:samelevel1}
    Let us consider an instance labeled $i$ on level $\ell$. This instance calls overall at most $2^i$ many instances on level $\ell$.
\end{lemma}

\begin{proof}
    Remember that a call on such an instance calls recursively on the same graph with $\lmin = \ceil{1.2^i}$ and $\lmax=\floor{1.2^{i+1}}$ for all $i \in \N_0$ with $1.2^{i} \le \lmax/8$. 
    We can ignore all recursive calls on the contracted graphs, as those will only create calls on higher levels.

    Let us look at the tree of recursive calls, 
    Note that any children labeled $j$ of a node labeled $i$ must satisfy $j+10 \le i$ as we have $1.2^j \le \frac {\floor {1.2^{i+1}}}{8} \le \frac { {1.2^{i+1}}}{8} \le  \frac { {1.2^{i+1}}}{1.2^{11}} $. 
    We relax this condition and note that $j+1 \le i$. Then, by \Cref{lem:countingcalls}, this creates at most $2^i$ many recursive calls.   
\end{proof}

\begin{lemma}\label{lem:norecursivecalls2}
    The total number of recursive calls is $n^{o(1)}$.
\end{lemma}

\begin{proof}
    Note that every instance on level $\ell$ calls immediately exactly one instance on level $\ell +1$ in the limit of $\ell = O(\sqrt[4]{\log n})$ many levels, by \Cref{lem:norecursivecalls2e}.
    Moreover, every call on level $\ell$ and index $i$ generates overall $2^i$ many instances on level $\ell$, by \Cref{lem:samelevel1}.

    Therefore, the initial call on level $0$ and label $i=O(\log \lambda)$ creates at most $2^i$ many instances on level $0$. By induction, it is straightforward to see that there are at most $2^{i(\ell +1)}$ many instances on level $\ell$. 
    Plugging in $1.2^i = O(\lmax)$ we have that the total number of recursive calls is equal to $\lmax^{O(\sqrt[4]{\log n})} = 2^{O(\log^{1-c} n)}$.
\end{proof}

\begin{lemma}
    The amortized recourse computed over all recursion levels is $n^{o(1)}$.
\end{lemma}

\begin{proof}
    The total recourse on the $i$-th level is $\Tilde O\PAR{(\frac \rho {\delta^3}\lmax)^{i-1}}$, as per \Cref{lem:clusterdec}.
    As $i= O(\log^{\frac 1 4} n)$ by the lemma above, we get that $\Tilde O\PAR{(\frac \rho {\delta^3}\lmax)^{i-1}} = 2^{O(\log^{1-c}n)}=n^{o(1)}$.
    Summing over all levels of recursion gives us the desired result.
\end{proof}

\begin{lemma}
    The amortized running time over all levels is $n^{o(1)}$ in \Cref{alg:dynclusterdecompose}.
\end{lemma}

\begin{proof}
    Each recursion level receives an amortized number of $n^{o(1)}$ updates as per the previous lemma, and requires $n^{o(1)}$ amortized time to process the update, by \Cref{lem:clusterdec}.
    Hence, on each level, after one update on the top level, the time to process that update is $n^{o(1)}$.
    Summing over all levels gives us the desired result.
\end{proof}

\begin{lemma}
    In \Cref{alg:dynclusterdecompose}, the total approximation ratio after all calls is $1$.
\end{lemma}

\begin{proof}
    Each recursive call degrades the quality of the solution by a factor $(1+2 \delta)$ at most, by \Cref{lem:clusterdecomposition}. 

    Note that since we chose $\delta < \frac 1 {2\lmax}$, all cuts smaller than $\lmax$ are thus preserved by our algorithm.
    This shows that our algorithm is exact on every level, and thus is exact overall.
\end{proof}

\subsection{Proof of \Cref{lem:clusterdec}}
In this subsection, we are given a graph with updates, and our goal is to maintain a cluster decomposition of this graph such that no cluster contains a cut of cut-size at most $\lmax$ that is $(1-\delta)$-boundary sparse.
As in the static setting the main idea is to maintain a dynamic expander decomposition which we refine further. 

\begin{definition}
    We say that an edge $e=(u,v)$ is \emph{incident} to a cluster $C$ if either $u$ or $v$ (or both) are in $C$.
\end{definition}

\begin{algo}[Maintaining the cluster decomposition]\label{alg:dynclusterdecompose}
    Preprocessing: We run the static algorithm on the initial graph, with the preprocessing of \Cref{thm:expanderdecomposition} for the expander decomposition.

    Handling updates: 
    We give the next $O(\frac{m\phi}{\rho})$ updates to the dynamic expander decomposition algorithm.
    Every change output by that algorithm is then processed in the following way:
    \begin{enumerate}[noitemsep]
    \item Run \textsc{UpdatePartition} (from \Cref{alg:StaticAlgorithm}), which updates the partition $\P_1$, and the classification of the clusters as ``well-connected'', ``poorly-connected'' or ``fragmented''.
        \item Mark the endpoints of the updated edges as unchecked.
        \item  Then, for each affected ``well-connected'' cluster $C$ (i.e. containing an endpoint of an affected edge), find $2\lmax$ intercluster edges leaving that cluster whose endpoint in $C$ is checked, and mark their incident vertices in the cluster as unchecked.
        \item While there exists in $\C$ (the cluster decomposition) a ``well-connected'' cluster $C$ with an unchecked vertex, we run the Find and Cut subroutine (\Cref{alg:subroutine}) on $C$, which might change $\C$.
        \item For every affected ``poorly-connected'' or ``fragmented'' cluster $C$, revert any ``fragmented' edges it contains to ``intracluster'', and run \Cref{alg:fragmenting} on $C$.
    \end{enumerate}

    Every $O(\frac{m\phi}{\rho})$ updates, we restart from scratch, i.e., we run the preprocessing step on the full graph.
\item Rules for Cluster Classification:
\begin{enumerate}
    \item By default, every cluster is ``well-connected''
    \item Any cluster $C$ that satisfies $\boundary C \le 3\lmax$ becomes ``poorly-connected''
    \item Any cluster $C$ that satisfies $\boundary C \ge 6\lmax$ becomes ``well-connected''
    \item In $\P_2$, any cluster that is the result of the fragmenting algorithm becomes ``fragmented'' (note that a ``poorly-connected'' cluster that goes through the fragmenting algorithm untouched still becomes ``fragmented''). This overrides other rules.
\end{enumerate}
\end{algo}

We start by showing the correctness of the algorithm.

\subsubsection{Correctness}

\begin{lemma}\label{lem:keepinvariant}
    Let $C\in \P_1$ be a cluster that satisfies \Cref{inv:cuttingisgood} at time $t$.
    Assume that between times $t$ and $t'$, $C$ is subject to updates, and that the Find and Cut subroutine (\Cref{alg:subroutine}) is not run in between $t$ and $t'$ on the cluster\footnote{This can either be because the cluster is poorly-connected, fragmented, or because $t'$ is the first time after $t$ where an update is incident to $C$}.
    If we mark every node incident to an edge insertion or deletion as well as the nodes incident to $2\lmax$ intercluster edges incident to $C$ apart from the edges inserted between $t$ and $t'$  as unchecked, then $C$ satisfies \Cref{inv:cuttingisgood} at time $t'$.
\end{lemma}
\begin{proof}
    Let $S$ be, at time $t'$, a $(1-\delta)$-boundary-sparse cut in $C$, with $w(S, C\setminus S) \le \lmax$, and $\vol (S) \le \frac \lmax \phi$.
    We need to show that $S$ contains an unchecked vertex at time $t'$. 
    We have multiple cases:
    
    Case 1: $S$ contains an unchecked vertex at time $t$:

    In that case, the same vertex is still unchecked at time $t'$.
    Indeed, since the Find and Cut subroutine was not run between $t$ and $t'$, no unchecked vertex becomes checked between the two times.

    Case 2: There is an edge with at least one endpoint in $S$ that was inserted or deleted.
    That update causes the corresponding endpoint to become unchecked. 
    Since the Find and Cut subroutine was not run between $t$ and $t'$, this endpoint is still unchecked at time $t'$.

    Case 3: If all vertices of $S$ were checked at time $t$ and no edge update was incident to $S$ between $t$ and $t'$:
    It follows by \Cref{inv:cuttingisgood} that $S$ was either:
    \begin{enumerate}[label=(\alph*), noitemsep]
        \item Not $(1-\delta)$-boundary sparse at time $t$, or
        \item $w(S, C\setminus S) > \lmax$ at time $t$, or
        \item $\vol (S) > \frac \lmax \phi$ at time $t$.
    \end{enumerate}
    As at time $t'$ we have that $w(S, C\setminus S)$ and no edge update was incident to $S$ between $t$ and $t'$, case (b) is impossible.

    As at time $t'$ we have that $\vol (S) \le \frac \lmax \phi$ and no edge update was incident to $S$ between $t$ and $t'$, case (c) is impossible.

    Let us thus look at case (a): At time $t$, $w(S, C\setminus S) \ge (1-\delta)\cdot \min\{w(S,V\setminus C), w(C\setminus S, V\setminus C)\}$.

    Since no edge update was incident to $S$ between $t$ and $t'$, we have that $w(S, V\setminus C)$ and $w(S, C\setminus S)$ remain unchanged between $t$ and $t'$.

    At time $t'$, by boundary sparseness, we have that $w(S, C\setminus S) < (1-\delta) w(S, V\setminus C)$.
    This inequality also holds at time $t$, and since $S$ is not $(1-\delta)$-boundary sparse at time $t$, we must have that at time $t$: $w(S,C\setminus S) \ge (1-\delta) w(C\setminus S, V\setminus C)$.

    But, as we are not in case (b), at time $t$, we have that $\lmax \ge w(S, C\setminus S) \ge (1-\delta) w(C\setminus S, V\setminus C)$.
    Therefore $w(C\setminus S, V\setminus C) < \frac 1 {1-\delta} \lmax<2\lmax$ as $\delta \in [0, \frac 1 2]$. 

    We have two cases: (i): $E(C, V\setminus C)$ contains at least $2\lmax$ many edges.
    As at least $2\lmax$ many edges in $E(C, V\setminus C)$ have their endpoints in $C$ unchecked, and there are strictly less than $2\lmax$ in $E(C\setminus S, V\setminus C)$, at least one edge in $E(S, V\setminus C)$ has its endpoint in $C$ unchecked. 
    That endpoint is in $S$ which concludes the proof.

    (ii): $E(C, V\setminus C)$ contains at most $2\lmax$ many edges.
    Since all of those edges have their endpoints unchecked, it remains to show that one of them has an endpoint in $S$.
    This is straightforward as $S$ is boundary sparse at time $t'$, which implies that $w(S, V\setminus C) > \frac 1 {1-\delta} w(S, C\setminus S) \ge 0$.    
\end{proof}

\begin{remark}
    One can apply \Cref{lem:keepinvariant} for exactly one update between $t$ and $t'$, which is what we do when the cluster is ``well-connected'' both before and after the update.
\end{remark}

\begin{corollary}\label{dynproba}
Assume that \Cref{inv:cuttingisgood} holds at some time $t$.
Assume the graph is subject to $t'-t$ updates.
   Then, in the partition $\P_1$ maintained by \Cref{alg:dynclusterdecompose} at time $t'$ no  ``well-connected'' cluster $C$ contains a cut $S$ that is $(1-\delta)$-boundary-sparse and such that $w(S, C\setminus S) \le \lmax$ and $\vol (S) \le \frac \lmax \phi$. In other words, $\P_1$ is a pre-cluster decomposition.
\end{corollary}
    
\begin{proof}
This is a direct consequence of \Cref{lem:3.13} and \Cref{lem:keepinvariant}, as since there are no cuts at time $t'$ in $G$ that are strictly smaller than $\lmin$, and \Cref{alg:dynclusterdecompose} ensures no unchecked vertex remains in any ``well-connected'' cluster.
 \end{proof}

 \begin{lemma}
     The partition $\P_2$ maintained by the algorithm is a cluster decomposition associated to the pre-cluster decomposition $\P_1$.
 \end{lemma}

 \begin{proof}
     By \Cref{dynproba}, $\P_1$ is a pre-cluster decomposition. 
     As we rerun the fragmenting algorithm on any cluster that has been modified in $\P_1$, this ensures $\P_2$ is the cluster decomposition associated to $\P_1$.
 \end{proof}

This, together with \Cref{lem:clusterdecomposition} shows correctness.

\subsubsection{Bounding the number of intercluster edges}

We now take a look at the number of intercluster edges just before a new rebuild, and aim to bound that number.
For that, as for the static algorithm we build below a forest of all the clusters that existed since the last rebuild, and analyze carefully the number of intercluster edges at each node of this forest.

Formally, let $M$ be the initial number of interexpander edges (after running the first static expander decomposition) and $R$ (for recourse) be the number of changes output by the dynamic expander decomposition.
We aim to bound the number of intercluster edges after all of the updates before the next rebuild.

For that, let $T$ be a time after an update. We build a forest $F$ of all the clusters ever created between the last rebuild before $T$ and time $T$.
The parent of each cluster is the cluster it was cut off from when it was created.
The roots are the expanders as output by the static expander decomposition.
We will denote by $t(C)$ the \emph{time} of cluster $C$, that is, the update at which $C$ was split (if it is an internal node), or $T$ otherwise, and $\boundary C(t)$ the boundary of cluster $C$ after update $t$.
Let $w_t(A,B)$ be the number of edges between $A$ and $B$ at time $t$.
Here, we assume that each update corresponds to one change output by the dynamic expander decomposition.

For each cluster $C$ let $L(C)$ be its boundary-size when it was split into two (if it is an internal cluster) or its boundary-size at time $T$ (if it is a leaf). 

We will then, for the analysis, add marks to each cluster in $\P_1$ as follows: A cluster $C$ that is a root gets mark $\oplus$ if $L(C)\ge 3\lmax$ and $\odot$ otherwise.
Then, in BFS fashion, each cluster $C$ whose parent is marked $\oplus$ gets marked $\oplus$ if $L(C) \ge 3\lmax$, and $\odot$ otherwise. If its parent is marked $\odot$, it gets marked $\oplus$ if $L(C) \ge 6\lmax$, and $\odot$ otherwise

Intuitively, a $\oplus$ labeled node is a cluster that is ``well-connected'' at the time it is split in two or at time $T$, and a $\odot$ is a ``poorly-connected'' cluster at time $T$.

The goal is to show that the sum of the labels of the leaves is at most $O(\frac {M+R} \delta)$.
As $M= O(m\phi)$ and $R=O(m\phi)$, this will yield at most $O(\frac {m\phi} \delta)$ many intercluster edges, which is crucial for bounding the recourse, as well as the running time (which depends on the number of intercluster edges).

We look at the following potential:
$$
\Phi(C) = \max\{0, \boundary C(t(C)) - 2.1 \lmax\}
$$
and the total potential:
$$
\Phi\{\C\} = \sum_{c\in\C} \Phi(C)
$$

Let $C_0$ be a node of $F$ that is marked $\oplus$ that is either a root of $F$ or has a parent labeled $\odot$.
We start with $\C=\{C_0\}$ and at each step, we will replace one element from $\C$ marked $\oplus$ with its two children until all elements of $\C=\C_f$ are either leaves of $F$ or marked $\odot$.

\begin{lemma}\label{lem:dynboundariessmaller}
    Let $C$ be a node in $F$ marked $\oplus$ and $S, C\setminus S$ its children, where $S$ is $(1-\delta)$-boundary sparse in $C$.
    Then we have that $\boundary C(t(C)) \ge \max\{\boundary S (t(C)), \boundary (C\setminus S)(t(C))\}+\frac{\delta \lmin} 2$.
\end{lemma}

\begin{proof}
    In this proof, all quantities are considered at time $t(C)$.
    
    Since $S$ is a $(1-\delta)$-boundary-sparse cut, we have that $w(S, C\setminus S) \le (1-\delta)\min\{ w(S, V \setminus C), w(C\setminus S, V\setminus C)\}$.
    Therefore, $\boundary S = w(S, V\setminus C) + w(S, C\setminus S) \le w(S, V\setminus C) + (1-\delta) w( C\setminus S, V\setminus C) \le \boundary C - \delta w( C\setminus S, V\setminus C)$.

    But $w( C\setminus S, V\setminus C) + w( C\setminus S, S) = \boundary (C\setminus S) \ge \lambda_{\min}$.
    Since $w( C\setminus S, V\setminus C) > w( C\setminus S, S)$ (by boundary-sparsity of $S$), we have that $w( C\setminus S, V\setminus C) \ge \frac {\lambda_{\min}} 2 $ and the result follows by symmetry for $\boundary (C\setminus S)$.    
\end{proof}

\begin{lemma}\label{lem:lemtwo}
    Let $C \in F$ be a set in $\C$ replaced with $S$ and $C\setminus S$.
    Let $u(C)$, $u(S)$ and $u(C\setminus S)$ be the number of updates to $S$ and $C\setminus S$ respectively for updates between $t(C)$ (excluded) and $t(S), t(C\setminus S)$ (included).

    Then $\Phi(S) + \Phi(C\setminus S) \le \Phi(C) - \frac{\delta \lmin}{2} + u(S)+u(C\setminus S)+2u(C)$.
\end{lemma}

\begin{proof}
    Recall that $\boundary C(t(C)) \ge 3\lmax$, which imply that $\Phi(C) \ge 0.9\lmax$.
    We then have five cases:
    \begin{itemize}
        \item if both $\boundary S(t(S)), \boundary (C\setminus S)(t(C\setminus S))\le 2.1\lmax$:
        $$
        \Phi(S) + \Phi(C\setminus S) = 0 \le \Phi(C)  - \frac {\delta \lmax} {2}
        $$
        as $\Phi(C) - \frac {\delta \lmax}{2} \ge 0.9\lmax - \frac {\delta \lmax}{2} >0$.
        \item if both $\boundary S(t(S)), \boundary (C\setminus S)(t(C\setminus S))> 2.1\lmax$, and $S$ is a cut that stems from the recourse of the expander decomposition or other updates to the partition:
        \begin{multline*}
        \Phi(S) + \Phi(C\setminus S) = \boundary S(t(S)) +\boundary (C\setminus S)(t(C\setminus S)) - 2\times 2.1\lmax\\
        \overset{(i)}\le \boundary S(t(C))+\boundary (C\setminus S)(t(C)) -2\times 2.1\lmax + u(S)+u(C\setminus S)\\
        \le \boundary C(t(C)) +2w_{t(C)}(S, C\setminus S) -2\times 2.1\lmax + u(S) + u(C\setminus S)\\
        \le \Phi(C)  +2\lmax + 2u(C) -2.1\lmax +u(S)+u(C\setminus S)
        \le \Phi(C) - \frac {\delta \lmax}{2} +u(S)+
        u(C\setminus S)+2u(C)
        \end{multline*}
        Where (i) stems from the fact that between times $t(C)$ and $t(S)$, at most $u(S)$ many edges got added to $S$ and therefore $\boundary S$.
         
        \item if both $\boundary S(t(S)), \boundary (C\setminus S)(t(C\setminus S))> 2.1\lmax$, and $S$ is a $(1-\delta)$-boundary-sparse cut of cut-size at most $\lmax$:
        \begin{multline*}
        \Phi(S) + \Phi(C\setminus S) = \boundary S(t(S)) +\boundary (C\setminus S)(t(C\setminus S)) - 2\times 2.1\lmax\\
        \overset{(i)}\le \boundary S(t(C))+\boundary (C\setminus S)(t(C)) -2\times 2.1\lmax + u(S)+u(C\setminus S)\\
        \le \boundary C(t(C)) +2w_{t(C)}(S, C\setminus S) -2\times 2.1\lmax + u(S) + u(C\setminus S)\\
        \le \Phi(C)  +2\lmax -2.1\lmax +u(S)+u(C\setminus S)
        \le \Phi(C) - \frac {\delta \lmax}{2} +u(S)+
        u(C\setminus S)
        \end{multline*}
        Where (i) stems from the fact that between times $t(C)$ and $t(S)$, at most $u(S)$ many edges got added to $S$ and therefore $\boundary S$.
        \item if wlog $\boundary S(t(S))\le 2.1\lmax$ and $\boundary (C\setminus S)(t(C\setminus S))> 2.1\lmax$, and $S$ is a $(1-\delta)$-boundary-sparse cut:

    \begin{align*}
        \Phi(S)+\Phi(C\setminus S)&= 0 + \boundary(C\setminus S)(t(C\setminus S))-2.1\lmax\\
        &\le 0 + \boundary(C\setminus S)(t(C))+u(C\setminus S)-2.1\lmax\\
        &\le \boundary C(t(C))-w_{t(C)}(S, V\setminus C) + w_{t(C)}(S, C\setminus S)+u(C\setminus S)-2.1\lmax\\
        &\overset{(ii)}\le \Phi(C)-\delta w_{t(C)}(S, V\setminus C) +u(C\setminus S)\\
    \end{align*}
    Where $(ii)$ follows from the $(1-\delta)$ boundary-sparsity of $S$ at time $t(C)$.
     We conclude using the following claim:
    \begin{claim}
        For any cluster $C$ and $(1-\delta)$-boundary sparse cut $S\subsetneq C$, we have that $w(S, V\setminus C) \ge \frac \lmin 2$.
    \end{claim}
    \begin{proof}
        Assume by contradiction that we have $w(S, V\setminus C) < \frac \lmin 2$.
        We are going to show that in that case,  $\boundary S < \lmin$ holds, a contradiction.
        By boundary sparsity, we have that $w(S, C\setminus S)<(1-\delta) w(S, V\setminus C)<(1-\delta)\frac \lmin 2 $, and therefore, $\boundary S = w(S, C\setminus S) + w(S, V\setminus C)<\lmin$.       
    \end{proof}

        \item if wlog $\boundary S(t(S))\le 2.1\lmax$ and $\boundary (C\setminus S)(t(C\setminus S))> 2.1\lmax$, and $S$ is a cut that stems from the recourse of the expander decomposition or other updates to the partition, that is not boundary-sparse:

    \begin{align*}
        \Phi(S)+\Phi(C\setminus S)&= 0 + \boundary(C\setminus S)(t(C\setminus S))-2.1\lmax\\
        &\le 0 + \boundary(C\setminus S)(t(C))+u(C\setminus S)-2.1\lmax\\
        &\le \boundary C(t(C))-w_{t(C)}(S, V\setminus C) + w_{t(C)}(S, C\setminus S)+u(C\setminus S)-2.1\lmax\\
        &\overset{(ii)}\le \Phi(C)-\delta \lmax+2u(C) +u(C\setminus S)\\
    \end{align*}
    Where $(ii)$ follows from the the fact that $S$ is non-boundary sparse and thus $u(C) \ge  w_{t(C)}(S, C\setminus S) \ge (1-\delta)w_{t(C)}(S, V\setminus C)$ and thus $u(C) \ge  \frac \lmin 2$.

    \end{itemize}
\end{proof}

\begin{lemma}\label{lem:sumfprime}
    The sum of all the labels of the leaves of $F$ does not exceed $O\PAR{\frac {M+R} {\delta}}$.
\end{lemma}

\begin{proof}
    Let us consider one cluster $C_0$ be a node of $F$ that is marked $\oplus$ that is either a root of $F$ or has a parent labeled $\odot$, and we consider all elements of $F$ that are reachable form $C_0$ only through paths with innner nodes marked $\oplus$. We start with $\Phi(C_0)=O(\boundary C_0(t_0)+u(C_0))$, where $t_0$ is the time of the last rebuild.
    By \Cref{lem:lemtwo}, the potential decreases by at least $\frac {\delta\lmin} 2 -u(S) - u(C\setminus S)+2u(C)$ every time we replace a cluster $C$ by its children $S$ and $C\setminus S$.
    Therefore, after $X$ many replacements each of which increases the number of clusters by 1, the potential is at most $\Phi(C_0) -X\frac {\delta \lmin} 2 + 3\sum_{C \in F(C_0)} u(C) \le \Phi(C_0)  +U(C_0)-X\frac {\delta \lmin} 2$, where $F(C_0)$ is the component of $C_0$ with inner nodes marked $\oplus$ in $F$ and $U(C_0)=3\sum_{C \in F(C_0)} u(C)$.
    The potential being positive, we have that $X=O\PAR{\frac{\Phi(C_0) + U(C_0)} {\delta \lmin}}$, and thus we have $O\PAR{\frac{\Phi(C_0) + U(C_0)} {\delta \lmin}}$ many clusters descendents of $C_0$ in $F$ as described above.

    Let $\C_f(C_0)$ be the set of all leaves in the component of $C_0$ in $F$.
    Recall that $L(C) = \boundary C(t(C)$.
    We then have:
    {\renewcommand{\boundary}[1]{L(#1)}
    \begin{align*}
        \sum_{C \in \C_f(C_0)} \boundary C &= \sum_{\substack{C \in \C_f(C_0)\\\boundary C \le 2.2\lmax}} \boundary C+\sum_{\substack{C \in \C_f(C_0)\\\boundary C > 2.2\lmax}} \boundary C\\
        \intertext{In the first sum, we know the that the total number of clusters is at most $O\PAR{\frac{\Phi(\C_0) + U(C_0)} {\delta \lmin}}$. For the second sum, note that if $\boundary C \ge 2.2\lmax$, then $\boundary C = O(\Phi(C))$.}
        &\le  O\PAR{\frac{\Phi(\C_0) + U(C_0)} {\delta \lmin}} 2.2\lmax+\sum_{\substack{C \in \C_f(C_0)\\\boundary C > 2.2\lmax}} O(\Phi( C))\\
        &\le  O\PAR{\frac{\Phi(\C_0) + U(C_0)} {\delta }}+ O(\Phi(\C_f(C_0))) = O\PAR{\frac{\Phi(\C_0) + U(C_0)} {\delta }}\\
    \end{align*}}

    For the last equality, recall that the sum of potentials of the two children of a node in $F$ is less than the potential of the parent to which we add the number of updates made to the children. It follows by induction that for any set of descendants of $C_0$ such that none is an ancestor of the other, the sum of their potentials is less than the potential of $C_0+U(C_0)$.
    As the clusters in $C_f(C_0)$ fulfill this condition, it follows that
    $\sum_{\substack{C \in \C_f(C_0)\\\boundary C > 2.2\lmax}} O(\Phi( C)) = O(\Phi(\{C_0\}+U(C_0))$.

    Summing over all such choices of $C_0$ in $F$, we get that the sum of the boundaries of all leaves of $F$ with a parent marked $\oplus$ does not exceed $O(\Phi(\{\C_0\}+U(C_0))$, where $\C_0$ is the set of all $C_0$ as described in the beginning of this proof.

    Let us now estimate $\Phi(\C_0)$.
    In $\C_0$, we have the roots in $F$ that are an output from the static expander decomposition, as well as nodes marked $\oplus$ whose parents are marked $\odot$.
    The sum of the boundaries of the expanders is $M$ initially, and does not exceed $M+R$ when they are cut.
    The sum of the boudaries of the nodes marked $\oplus$ whose parents are marked $\odot$ is at most $O(R)$, as they have ancestors with boundary-size at most $3\lmax$ while they have boundary-size at least $6\lmax$. This boundary increase can only be due to the recourse $R$.

    Therefore $\Phi(\C_0) = O(M+R)$.

    Finally, notice that all other leaves of $F$ have boundary at most $O(M+R)$ as well, as all nodes that are marked $\odot$ who are roots or have a node marked $\oplus$ have boundary at most $O(M+R)$ by the above analysis.
    The boundary of their children can only increase due to the partition changing, which is due to recourse, and thus the boundary-size of their descendents is at most $O(M+R)$ as well.
    \end{proof}

    \begin{lemma}\label{lem:boundaries}
        The number of intercluster edges in $\P_1$ is $O\PAR{\frac {M+R} \delta}$.

        The number of intercluster edges in $\P_2$ (that is, edges labeled ``intercluster'' or ``fragmented'' in our algorithm) is $\tilde O\PAR{\frac {M+R} {\delta^3}}$.
    \end{lemma}

    \begin{proof}
        The first statement is immediate from \Cref{lem:sumfprime}.

        For the second statement, by \Cref{thm:fragmenting}, fragmenting a cluster increases its boundary-size by at most a $\frac 1 {\delta^2} \poly \log n$ factor. 
        Applying that on all ``poorly-connected'' clusters in $\P_1$ gives the result.
    \end{proof}

\subsubsection{Running time}

\begin{lemma}
    The total number of times, over the $O(\frac{m\phi}\rho)$ updates to the input of \Cref{alg:dynclusterdecompose}, we process an unchecked vertex is $O(\frac {M+R} {\delta}  \frac \lmax \phi)$.
\end{lemma}

\begin{proof}
    Let us first estimate the number of nodes that get marked as unchecked overall (counting a node that has been unchecked multiple times with multiplicity):

    We divide the nodes that were unchecked at some point in two categories: nodes that were adjacent to an edge update, and nodes that we adjacent to an edge that became an intercluster edge.
    Any edge that was an intercluster edge and then deleted belongs to the first category.

    Nodes that were adjacent to an edge update are at most $R$ many.
    Nodes that were adjacent to an edge that was cut are at most $O(\frac {M+R}{\delta})$ many, as there are $O(\frac {M+R}{\delta})$ many such edges in the decomposition $\P_1$ by \Cref{lem:boundaries}, and any intercluster edge in a prior decomposition is either a deleted edge or an intercluster edge in $\P_1$, as we do not merge clusters.

    In the worst case, all updates are made to a ``well-connected'' cluster in $G$, and thus we need to add $2\lmax$ unchecked nodes.
    Moreover, while processing an unchecked node, we might find a cut that needs to be split from the original cluster.
    In that case, we have to process the node again, but this can happen at most $\frac 1 \phi$ times as per \Cref{lem:nottoomanykargers}.
\end{proof}

\begin{lemma}
    The total number of times, over the $O(\frac{m\phi}\rho)$ updates to the input of \Cref{alg:dynclusterdecompose}, we run the fragmenting algorithm is $O(\frac {M+R} {\delta})$.
\end{lemma}

\begin{proof}
    Each intercluster edge in $\P_1$ can result in at most two runs of the fragmenting algorithm (one the cluster at each endpoint). 
    Moreover, any edge update can also result in at most two runs of the fragmenting algorithm.
    Since there are at most $O(\frac {M+R} {\delta})$ of these edges, the result follows.
\end{proof}

\begin{corollary}\label{cor:recourse}
    The total recourse of the algorithm over the $O(\frac {m\phi} \rho)$ updates to the input of \Cref{alg:dynclusterdecompose} is $\tilde O\PAR{\frac {m\phi} {\delta^3}\lmax}$.
\end{corollary}

\begin{proof}
    At any given time, there are at most $O(\frac {m\phi} \delta)$ many intercluster edges in $\P_1$, by \Cref{lem:boundaries}.
    Since $\P_1$ only gets refined over time, this yields $O(\frac {m\phi} \delta)$ recourse.

    Moreover, there are at most $O(\frac {M+R} {\delta})$ calls to the fragmenting algorithm, and as many reversions of whatever changes these calls made. 
    Each of these calls make at most $\tilde O(\frac \lmax {\delta^3})$ changes by \Cref{thm:fragmenting}, hence the result.    
\end{proof}

\begin{corollary}\label{cor:time}
    As $\lmax = n^{o(1)}$ and $\phi = n^{-o(1)}$, the total running time of \Cref{alg:dynclusterdecompose} between two rebuilds, excluding recursive calls, is $(M+R)\cdot n^{o(1)}=m^{1+o(1)}$.
\end{corollary}

\begin{proof}
    As per \Cref{thm:deterministic-dynamic}, amortized update time and request time to LocalKCut is $\tilde{O}\left(\beta^4( \lambda_{\max}\nu)^{O(\beta)}\right)$ with $\beta = O(1)$, and therefore all calls to LocalKCut -- outside of the fragmenting algorithm -- need $\Tilde O(\frac {M+R} \delta \beta^4( \lambda_{\max}\nu)^{O(\beta)}) = (M+R)\cdot n^{o(1)}$ time. 
    This is also an upper bound to the number of nodes that need to be unchecked, as each LocalKCut is ran on an unchecked vertex.
    Furthermore, we need to update the data structure, which takes amortized time $n^{o(1)}$ per update, by \Cref{thm:datastructure}, as the graph has initial volume $m$ and is subject to updates of volume $m\phi$ between two rebuilds.
\end{proof}

This proves the following lemma:

\lemrecourse*

\begin{proof}
    Item 1. is immediate from \Cref{cor:recourse} as we can amortize the recourse over all updates between two rebuilds.

    Item 2. is a direct application of  \Cref{lem:boundaries}, where we note that between two rebuilds, the decomposition is only refined, and thus the number of intercluster edges just before the next rebuild is an upper bound on the number of intercluster edges at any point since the last rebuild.
    
    Item 3. is proven in \Cref{cor:time}.
\end{proof}

\section{Data Structure}\label{sec:datastructure}

This section is dedicated to proving the following theorem:

\begin{theorem}\label{thm:datastructure}
    Consider a fully-dynamic graph $G$, two (dynamic) partitions $\P_1$ and $\P_2$ of $G$. 
    $G$ can be subject to edge insertions or deletions, and $\P_1$ can only be subject to splits (that is, replace a cluster $C \in \P_1$ by $C\setminus S$ and $S$ for some $S \subsetneq C$), while in $\P_2$, in addition to splits, clusters of volume at most $\nu$ (for some parameter $\nu \in \N_+$) are allowed to merge back to larger clusters.
    $\P_1$ and $\P_2$ are encoded by a labeling of the edges as ``intracluster'', ``intercluster'', and ``fragmented'', that is, $\P_1$ is formed by the connected components formed by looking only at the edges labeled ``intracluster'' and ``fragmented'', while $\P_2$, a refinement of $\P_1$, is formed by the connected components formed by looking only at the edges labeled ``intracluster''. Each vertex, if it is a vertex obtained by the contraction of other vertices, has a label representing the inner volume of the contracted cluster it represents.

    There exists a data structure that maintains the mirror graph $G_{\P_1}$ and the contracted graph $G/{\P_2}$ of $G$. $G/{\P_2}$ should have, as vertex labels, the inner volume of the cluster the vertex represents. 
    Over $U$ many edge updates, it has $\tilde O(m+U\nu)$ running time (preprocessing + update handling). 
    Moreover, for any request on a set $S\subseteq C$ for $C \in \P_2$, it can output in time $\tilde O(\vol(S))$ whether or not $S$ is $(1-\delta)$-boundary-sparse.
    It can also store for each node, edge and cluster a set of constant size of parameters accessible in constant time.
\end{theorem}

We need a data structure that enables us to efficiently compute all the values and get the correct pointers at all times.
This includes, for every cluster $C$, access to its boundary edges, and its boundary-size. 
We should also be able to compute efficiently, for any cut $S$ of volume at most $O(\nu)$ inside $C$, whether or not it is boundary sparse, that is we need $w(S, C\setminus S)$ and $w(S, V\setminus C)$.
We also need for each node to know in which cluster it is, and whether it is checked or unchecked.

For that, we have the following algorithm:

\begin{algo}\label{alg:datastructure}
\item Input:
    A fully-dynamic graph $G$ under edge insertions, deletions, update of vertex and edge labels.
    \item Maintained variables: 
    \begin{itemize}
    \item for each node:
    \begin{itemize}
        \item a pointer to the cluster it is currently in
        \item a pointer to each of the edges it is incident to
        \item other node parameters
    \end{itemize}
    \item for each edge:
    \begin{itemize}
        \item a pointer to its endpoints
        \item Whether it is an ``intercluster'', ``intracluster'' or ``fragmented'' edge
        \item other edge parameters
    \end{itemize}
    \item for each cluster $C$ in $\P_1$ or $\P_2$:
    \begin{itemize}
        \item a list of the nodes that constitute this cluster
        \item the boundary-size of the cluster
        \item the inner volume of the cluster
        \item a list of all edges on its boundary
        \item other cluster parameters
    \end{itemize}
    \item for each cluster $C$ in $\P_1$:
    \begin{itemize}
        \item a mirror cluster $C'=G/(G\setminus C)$
    \end{itemize}
     \item an instance of Minimum Spanning Forest (\Cref{thm:dynforest}) on intracluster and fragmented edges which can output whether or not two nodes are in the same $\P_1$ component in $\tilde O(1)$ time.
    \item an instance of Minimum Spanning Forest (\Cref{thm:dynforest}) on intracluster edges which can output whether or not two nodes are in the same $\P_2$ component in $\tilde O(1)$ time.
    \item a contracted graph ${G/\P_2}$
\end{itemize}
\item Preprocessing:
\begin{enumerate}
\item Instantiate the Minimum Spanning Forests
    \item Compute the connected components (or clusters) of the graph $G$ where only edges labeled ``intracluster'' or ``fragmented'' are considered, to get $\P_1$.
    \item Compute the connected components (or clusters) of the graph $G$ where only edges labeled ``intracluster'' are considered, to get $\P_2$.
    \item Traverse each cluster in $\P_1$, creating its mirror cluster.
    \item Create a node to represent each cluster in $\P_1$ or $\P_2$ and add pointers from and to all the nodes it contains.
    \item Traverse each cluster in $\P_2$ to create the contracted graph. 
\end{enumerate}
\item Handling a split (in either $\P_1$ or $P_2$ or both) specified by a set of edges changing their labels: %
\begin{enumerate}
\item \label{split0} Update the Minimum Spanning Forests
    \item \label{split1} Choose an updated edge, explore in lockstep the connected component on each side of that edge (for both/either $\P_1$ and/or $\P_2$). Once one connected component $S$ is completely uncovered, stop the exploration.
    \item \label{split3} Create a new node for the new (smaller) cluster, and update all the pointers to the nodes it contains to the new cluster, as well as the edges on its boundary. This also updates the old cluster into the new bigger cluster.
    \item \label{split4} Update the edges in the mirror graph if $\P_1$ changes.
    \item \label{split5} Update the edges in the contracted graph if $\P_2$ changes.
    \item \label{split6} Compute the boundary and inner volume of the smaller cluster, store them and use it to update the boundary and inner volume of the bigger cluster.
\end{enumerate}
\item Handling an edge deletion:
\begin{enumerate}
    \item Update the inner volume of the cluster it is in, and the corresponding label in the contracted graph.
    \item Update the mirror cluster in which the edge is, if necessary
    \item Update the Minimum Spanning forest
    \item If the two endpoints are now in different connected components of the minimum spanning forest, run the steps in ``Handling a split specified by a set of edges changing their label''.
\end{enumerate}
\item Handling an ``intracluster'' or ``fragmented'' edge insertion in $\P_1$, or an ``intracluster'' edge insertion in $\P_2$:
\begin{enumerate}
    \item In $\P_1$, if one of its endpoints is an isolated vertex, add it to the cluster of the other endpoint, updating the pointers accordingly
    \item In $\P_2$, look if the two endpoints are in the same connected component. If they are not, explore in lockstep the connected components of both endpoints. When one component is completely explored, delete the corresponding cluster and add the vertices to the cluster of the other endpoint. Update the pointers accordingly.
    \item Update the inner volume of the corresponding cluster and that volume in the contracted graph
    \item Update the mirror graph accordingly
    \item Update the Minimum Spanning Forest
\end{enumerate}
\item Handling any other update:
\begin{enumerate}
    \item Trivially update all data structures related to that update.
\end{enumerate}
\end{algo}

We start by arguing that only allowing cluster splits (ignoring merging of components in $\P_2$ for the moment) allows us to ensure the update time can be amortized to the preprocessing time.
We call the \emph{union graph between two updates} $t_1$ and $t_2$ the graph that is the union of all edges that appears at least once between updates $t_1$ and $t_2$, that is, an edge that either exists already at update $t_1$, or is inserted at any point between times $t_1$ and $t_2$.

One main argument that is useful throughout is the fact that if we maintain this data structure on a decomposition of the graph, such that the only allowed operations are to rebuild from scratch the whole decomposition or to split a cluster such that splitting a cluster only takes time proportional to the smaller side of the cut, then the data structure can be maintained efficiently:

\begin{lemma}\label{lem:datastructure}
    Let $\P$ be a (dynamic) partition of $G$, where $G$ has $U$ many updates and initially contains $m$ edges.
    Suppose that $\P$ starts as $\P=\{V\}$, and can be subject to splits during those updates, that is, an element $C$ in $\P$ can be replaced by $C\setminus S$ and $S$, for some $S\subsetneq C$.
    Let $\mathcal{D}$ be a data structure on $G$ and $\P$ with preprocessing time $mT$ such that it takes time $T\cdot\min\{\vol(S), \vol(C\setminus S)\}$ to handle a split, for some $T>0$.

    Then over all the updates, maintaining $\mathcal{D}$ takes $\tilde{O} ((m+U)T)$ time overall.
\end{lemma}

\begin{proof}
    Let us consider the union graph over the updates. 
    This graph has $m+U$ edges at most.
    For any set $X$, let $\vol^*(X)$ be the volume of $X$ in the union graph.
    Then, at any point in time, $\vol(X) \le \vol^*(X)$.

    Let us now consider a split of $C$ into $S$ and $C\setminus S$, and assume it happens after update $t$. Assume that, w.l.o.g.,  at time $t$, we have that $\vol(S) \le \vol(C\setminus S)$.
    Then, at time $t$, we have that $\vol(S)\le \vol^*(S)$ and $\vol(C\setminus S) \le \vol^*(C\setminus S)$.
    In particular, $\vol(S)\le \vol^*(C\setminus S)$ and thus $\vol(S) \le \min\{\vol^*(S),\vol^*(C\setminus S)\}$.

    Hence, if we look at the union graph, each split costs at most the side of the smallest volume (up to a factor $T$). 
    Whenever a split is made, assign to each edge on the smaller side of the cut a cost of $T$, where the sum of all the costs is thus higher than the time spent processing the split. Consider the half-edges (also called stubs), where each half-edge of an edge is assigned to one of the incident vertices.
    Then the volume of a node is the number of half-edges assigned to it, and the volume of a subset is the number of half-edges assigned to one of its vertices.
    Since each half-edge can only be on the smallest side of the split $\log (m+U)$ times, each half-edge is assigned at most cost $\log(m+U)T$. Summing over all half-edges and adding the preprocessing time concludes the proof.
\end{proof}

\begin{corollary}
    Let $\P$ be a dynamic partition of $G$, where $G$ has updates of total volume $U$ and initially contains $m$ edges. 
    Suppose that $\P$ starts as $\C=\{V\}$ and can be subject to splits during those updates, as well as merges where the volume of the smaller cluster being merged is at most $\nu$. 

    Let $\D$ be a data structure on $G$ and $\P$ that with preprocessing time $mT$ such that it takes time $T\cdot \min\{\vol(S) , \vol(C\setminus S)\}$ to handle a split, for some $T>0$.
    Then, over all updates, maintaining $\D$ takes $\tilde O((m+U\nu)T)$ time overall.
\end{corollary}

\begin{proof}
    Whenever a merge is observed, simply simulate it by deleting the smaller cluster to merge, and inserting it again inside the cluster it is merged to.
\end{proof}

\begin{lemma}
    \Cref{alg:datastructure} handles a split of a cluster $C$ into $S$ and $C\setminus S$ with $\vol(S) \le \vol(C\setminus S)$ in $\tilde O(\vol(S))$ time.
\end{lemma}

\begin{proof}

    Let us first consider that the split is encoded by a change of labels of the edges: 
    
    Step~\ref{split0} takes $O(\vol(S) \log^4 n)$ by~\Cref{thm:dynforest}, as there are at most $\vol(S)$ many edges deleted from the instance. Steps~\ref{split1},~\ref{split3} \ref{split4} and \ref{split5} take $O(\vol(S))$ time as well.
    In Step~\ref{split6}, computing the boundary and the inner volume of $S$ takes $O(\vol(S))$ time, and the inner volume and boundary-size of the new bigger cluster $C\setminus S$ can be computed in $O(\vol (S))$ time from the values for $C$ using $\invol(C) = \invol(C\setminus S) +\invol (S) + w(S, C\setminus S)$ and $\boundary C = \boundary S + \boundary C\setminus S - w(S, C\setminus S)$.

    If the split is encoded as an intracluster edge deletion, note that the first three steps take $\tilde O(1)$ time, and the fourth step is the same as the ones analyzed above.
\end{proof}

\begin{lemma}
    For $\P_1$, \Cref{alg:datastructure} handles any update other that a split in $\tilde O(1)$ time.
\end{lemma}

\begin{proof}
    This is immediate as every step takes $\tilde O(1)$ time.
\end{proof}

\begin{lemma}
    For $\P_2$, \Cref{alg:datastructure} handles any update other that a split in $\tilde O(\nu)$ time.
\end{lemma}

\begin{proof}
    This is immediate as every step takes $\tilde O(1)$ time, apart from when the update merges two clusters back together, which takes time $O(\bar \nu)$ where $\bar \nu$ is the volume of the smaller of the two clusters.
    as $\bar \nu \le \nu$, the result follows.
\end{proof}

These three lemmata together with \Cref{lem:datastructure} prove the running time claimed in~\Cref{thm:datastructure}. It remains to show that we can compute whether a cut is $(1-\delta)$boundary sparse efficiently.

\begin{lemma}
    With our data structure, checking whether a cut $S$ in cluster $C$ is $(1-\delta)$-boundary sparse takes $O(\vol(S))$ time.
\end{lemma}

\begin{proof}
    We need to compute three quantities: $w(S, C\setminus S)$, $w(S, V\setminus C)$ and $w(C\setminus S, V\setminus C)$.
    The first two ones can be trivially computed in $O(\vol(S))$. The third one can be computed using the value stored at $C$ as $\boundary C$: $\boundary C = w(C\setminus S,V\setminus C) + w(S, V\setminus C) $.
\end{proof}

\section{Dealing with Weighted Graphs}
\label{sec:weighted}
In this section, we will present a way to sample a graph where edges have integer weights.
For that, we extend the weight function described above so that $w(e)$ is the weight of edge $e$.
Again, we will rely on the edge-sampling approach of Karger~\cite{kargersparse} to solve the problem.

The idea is as follows: If  we estimate the minimum cut to be between $\lambda_i = 1.1^i$ and $\lambda_{i+1} = 1.1^{i+1}$, sampling for each edge $e$ a binomial random variable $X(e)$ with parameters $w(e)$ and ${p_i}=\frac{54 \ln n}{\epsilon^2 \lambda_{i}}$, and define $\tilde G_i$ to be the graph obtained from $G$ by replacing each edge with $X(e)$ parallel edges.
Then, with high probability, every cut in $\tilde G_i$ has a value that is at most a $(1-\epsilon)$ factor away from its expectation.

However, a binomial distribution cannot trivially be sampled efficiently.
The naive strategy, that is, to sample $w(e)$ many Binomial random variables with probability ${p_i}$, takes $O(w(e))$ time, which can be too large.
Instead, we note that we only need to preserve the cuts whose value is at most $\lmax$ in $ \tilde G_i$.
Indeed, for any cut that has value strictly larger than $\lmax$ in $\tilde G_i$, we do not mind that its value is a $(1-\epsilon)$ factor away from its expected value, we only need that it stays above $\lmax$ so that our algorithm discards it. 
Therefore, we only need to exactly compute $Y(e)$, the random variable defined as follows: $Y(e) = X(e) \1 [X(e) \le \lmax] + (\lmax +1) \1 [X_e > \lmax]$, and define $G_i$ to be the graph obtained from $G$ by replacing each edge with $Y(e)$ parallel edges.

The rest of this section formalizes these ideas.

\begin{definition}
    In the case of integer weights, the sampled graph $G_i$ is defined as follows. 
    Let ${p_i} =\frac{54 \ln n}{\epsilon^2 \lambda_{i}}$ and $\lmax = 1.1(1+\epsilon) \frac{54\ln n} {\epsilon^2}$. 
    For every edge $e \in E(G)$, let $X(e)$ be an independent Binomial random variable of parameters $w(e), {p_i}$.
    Define $Y(e)$ as follows:
    $$
    Y(e) = \begin{cases}
        X(e) & \text{if } X(e) \le \lmax\\
        \lmax +1 & X(e) > \lmax
    \end{cases}
    $$
    Let $\tilde G_i$ be the graph obtained from $G$ by replacing each edge with $X(e)$ parallel edges, and $G_i$ to be the graph obtained from $G$ by replacing each edge with $Y(e)$ parallel edges.
\end{definition}
Note that we will only compute $Y(e)$ for every edge $e$, and thus only compute $G_i$, not $\tilde G_i$. 
However, $\tilde G_i$ and $X$ will be useful for our analysis.

\begin{lemma}\label{lem:stillmincut}
    In $G_i$, let $S$ be a cut of value $c \le \lmax$.
    Then, with high probability, its value in $G$ is in $[(1-\epsilon) \frac c {p_i}, (1+\epsilon) \frac c {p_i}]$.
    Moreover, if $S$ also is the minimum cut of $G_i$, then with high probability, it is a minimum cut of $G$.
\end{lemma}

\begin{proof}
    If $c\le \lmax$, then for every cut-edge $e$, we have that $Y(e)\le \lmax$, and thus $X(e) \le \lmax$.
    Therefore, the value of the cut is the same in $G_i$ and $\tilde G_i$.
    By \Cref{thm:kargersparsify}, we thus have that if $c'$ is the value of the cut $S$ in $G$, then with high probability, $(1-\epsilon) {p_i}c'\le c\le (1+\epsilon) {p_i}c'$.
    This translates to $c'\le \frac c {{p_i}(1-\epsilon)}\le \frac c {p_i} (1+\epsilon)$ and $c'\ge \frac c {{p_i}(1+\epsilon)} \ge \frac c {p_i} (1-\epsilon) $.

    In the case where $S$ is also the minimum cut of $G_i$, note that any other cut $S'$ falls in one of two cases:
    Either all its cut edges $e$ satisfy $Y(e) \le \lmax$ and thus $\boundary_{\tilde G_i}S' = \boundary_{G_i} S' \ge \boundary_{G_i} S = \boundary_{\tilde G_i} S$.
    Or there exists a cut edge $e$ such that $Y(e) = \lmax+1$ and thus $ \boundary_{\tilde G_i} S' \ge X(e) > \lmax \ge \boundary_{G_i} S = \boundary_{\tilde G_i} S$.
    Therefore $S$ is also a minimum cut in $\tilde G_i$, and by \Cref{thm:kargersparsify}, it corresponds to the minimum cut in $G$.
\end{proof}

\begin{lemma}\label{lem:shorttime}
    We can sample $G_i$ in $O(m \lmax ^2)$ time.
\end{lemma}

\begin{proof}
    For each edge $e$ of weight $w(e)$, we will show that we can sample $Y(e)$ in $\lmax^2$ time.
    For that, note that: $$\Proba\PAR{Y(e) = k} = \begin{cases}
        {w(e) \choose k}{p_i}^k(1-{p_i})^k & \text{if } k \le \lmax\\
        \lmax +1 & \text{otherwise}
    \end{cases}$$

    Writing ${w(e) \choose k} = \frac{w(e)\times \dots \times (w(e)-k+1)}{k\times \dots \times 1}$, we can see that every  ${w(e) \choose k}{p_i}^k(1-{p_i})^k$ can be computed in $O(k) = O(\lmax)$ time for every $k \le \lmax$. 
    Hence, computing every probability $\Proba(Y(e) = k)$ for $k \le \lmax$ takes $O(\lmax^2)$ time, and $\Proba(Y(e) = \lmax +1)$ can be deduced in $O(\lmax)$ time.

    To sample $Y(e)$, all we have to do is thus sample uniformly at random $y(e)$ in $[0,1]$, and set $Y(e)$ to be the smallest integer $j$ such that $\sum_{i=0}^j \Proba(Y(e) = k) > y(e)$.
\end{proof}

We now discuss how to sample $G_i$ in the case where real weights are allowed, but no parallel edges are.
First note that any edge of weight at most $\frac {\epsilon \lambda_i} {n^2}$ can be discarded.
Indeed, for any cut $S$, at most $n^2$ edges cross the cut, and thus at most $\epsilon \lambda_i$ weight could have been discarded. 
This preserves the $(1+\epsilon)$-approximation.

We then replace each weight with the closest integer multiple of $\frac 1 {x_i}$, where ${x_i}= \ceil{\frac{n^2}{\epsilon \lambda_i}}$.
Indeed, this ensures that the difference between any old weight and newly assigned weight is at most $\frac {\epsilon \lambda_i}{n^2}$.
Similarly to before, this ensures that it only affects the cut-size of the cut by a $(1+\epsilon)$ factor.

Finally, we multiply each edge weight by ${x_i}$ to get integer weights, sample as above with ${p_i}=\frac{54\ln n}{\epsilon^2 \lambda_i {x_i}}$.
Notice how we reduce ${p_i}$ to correct for the fact that the weights got multiplied by ${x_i}$.

\begin{definition}\label{def:Gi}
    In the case of real weights, the sampled graph $G_i$ is defined as follows. 
    Define ${x_i}=\ceil{\frac{n^2}{\epsilon \lambda_i}}$ and $\hat w_i(e)\in \N$ such that $ \frac {\hat w_i(e)} {x_i} \le w(e)< \frac {\hat w_i(e)+1} {x_i}$.
    Let ${p_i} =\frac{54 \ln n}{\epsilon^2 \lambda_{i}{x_i}}$ and $\lmax = 1.1(1+\epsilon) \frac{54\ln n} {\epsilon^2}$. 
    For every edge $e \in E(G)$, let $X(e)$ be an independent Binomial random variable of parameters $\hat w_i(e), {p_i}$.
    Define $Y(e)$ as follows:
    $$
    Y(e) = \begin{cases}
        X(e) & \text{if } X(e) \le \lmax\\
        \lmax +1 & X(e) > \lmax
    \end{cases}
    $$
    Let $\hat G_i$ be $G$ with $\hat w_i$ as a weight function.
    Let $\tilde G_i$ be the graph obtained from $G$ by replacing each edge with $X(e)$ parallel edges, and $G_i$ to be the graph obtained from $G$ by replacing each edge with $Y(e)$ parallel edges.
\end{definition}

\begin{lemma}\label{lem:goodrange}
    In $G_i$, let $S$ be a cut of value $\lmin \le c \le \lmax$.
    Then, with high probability, its value in $G$ is in $[(1-\epsilon) \frac c {{p_i}{x_i}}, (1+\epsilon)^3 \frac c {{p_i}{x_i}}]$.
    Moreover, if $S$ also is the minimum cut of $G_i$, then with high probability, it is a $(1+\epsilon)^2$-approximation of the minimum cut of $G$.
\end{lemma}

\begin{proof}
    
    If $c\le \lmax$, then for every cut-edge $e$, we have that $Y(e)\le \lmax$, and thus $X(e) \le \lmax$.
    Therefore, the value of the cut is the same in $G_i$ and $\tilde G_i$.
    By \Cref{thm:kargersparsify}, we thus have that if $c'$ is the value of the cut $S$ in $\hat G_i$, then with high probability, $(1-\epsilon) {p_i}c'\le c\le (1+\epsilon) {p_i}c'$.
    This translates to $c'\le \frac c {{p_i}(1-\epsilon)}\le \frac c {p_i} (1+\epsilon)$ and $c'\ge \frac c {{p_i}(1+\epsilon)} \ge \frac c {p_i} (1-\epsilon) $.

    Let us now compare $c'$ with $c''$, the value of the cut $S$ in $G$. 
    There are at most $n^2$ cut edges in $S$, and thus:
    $$
    \frac {c'} {x_i} = \sum_{e \in E(S, V\setminus S)}\frac {\hat w_i(e)} {x_i} \le \sum_{e \in E(S, V\setminus S)}w(e)  = c'' \le  \sum_{e \in E(S, V\setminus S)}\frac {\hat w_i(e)+1} {x_i} \le \frac {c' + n^2} {x_i}\le \frac {c'}{{x_i}}+\epsilon \lambda_i
    $$
    However, $c' \ge \frac c {p_i} (1-\epsilon) \ge \frac {\lmin \epsilon ^2 \lambda_i {x_i} }{ 54 \ln n} \ge \frac {(1-\epsilon)54\ln n\epsilon ^2\lambda_i{x_i}}{\epsilon^254\ln n}\ge (1-\epsilon)\lambda_i {x_i}$ and thus $\frac {c'} {x_i} \ge (1-\epsilon )\lambda_i$.
    Hence $\frac {c'}{{x_i}} \le c'' \le \frac {c'} {x_i} (1+\frac \epsilon{1-\epsilon})$.
    We therefore have:
    $$
    \begin{cases}
        c''\ge \frac {c'} {x_i} \ge \frac c {{p_i}{x_i}} (1-\epsilon)\\
        c''\le \frac {c'} {x_i} (1+\frac \epsilon{1-\epsilon}) \le \frac c {{p_i}{x_i} (1-\epsilon)} (1+ \frac \epsilon {1-\epsilon})\le \frac c {{p_i}{x_i}} (1+\epsilon)^3
    \end{cases}
    $$

    In the case where $S$ is also the minimum cut of $G_i$, note that any other cut $S'$ falls in one of two cases:
    Either all its cut edges $e$ satisfy $Y(e) \le \lmax$ and thus $\boundary_{\tilde G_i}S' = \boundary_{G_i} S' \ge \boundary_{G_i} S = \boundary_{\tilde G_i} S$.
    Or there exists a cut edge $e$ such that $Y(e) = \lmax+1$ and thus $ \boundary_{\tilde G_i} S' \ge X(e) > \lmax \ge \boundary_{G_i} S = \boundary_{\tilde G_i} S$.
    Therefore $S$ is also a minimum cut in $\tilde G_i$, and by \Cref{thm:kargersparsify}, it corresponds to the minimum cut in $\hat G_i$.
    Let $a'$ be the value of the minimum cut $C$ in $\hat G_i$, $a''$ the value of $C$ in $G$, $c'$ the value of $S$ in $\hat G_i$, and $c''$ the value of $S$ in $G$.
    We then have that $c' \le a'$ with high probability as $S$ is the minimum cut of $\hat G_i$, $c'' \le \frac{c'}{{x_i}}(1+\frac{\epsilon}{1-\epsilon}) \le \frac {c'} {{x_i}}(1+\epsilon)^2\le \frac {a'}{{x_i}}(1+\epsilon)^3$.
    We also have that $a'' \ge \frac {a'}{{x_i}}$ as we round down the weights to get $\hat G_i$, which concludes the proof.
\end{proof}

\begin{observation}
    Note that \Cref{lem:shorttime} still holds in the weighted case, as we simply add a rounding step before applying the same algorithm, and where the value of $w(e)$ for each edge $e$ does not intervene in the proof.
\end{observation}

\begin{observation}
    Note that after sampling, our graph is unweighted, and thus, in all other sections, the graph can be assumed to be unweighted.
\end{observation}

We now discuss how to find a dynamic algorithm that finds an approximate minimum cut in weighted graphs, where $U$ and $L$ are respectively the maximum and minimum weight allowed for the algorithm:

\begin{algo}[Global Weighted Dynamic Algorithm]\label{alg:weighted}
    Define $O(\log_{1.1}(n^2\cdot \frac U L))$ many estimates of $\lambda$: $\lambda_0, \dots, \lambda_{\log_{1.1}(n^2\cdot \frac U L)}$ where $\lambda_i = L\cdot 1.1^i$ for every $i \in [\log_{1.1}(n^2\cdot \frac U L)]$.
    Define for every $i$, $G_i$, $p_i$ and $x_i$ as described in \Cref{def:Gi}.
    Run \Cref{alg:dynamicestimate} on every $G_i$ with $\lmin = (1-\epsilon)^2\frac{54\ln n} {\epsilon^2}, \lmax = 1.1(1+\epsilon) \frac{54\ln n} {\epsilon^2}$.
    
    For initialization, initialize the instances in increasing order of index $i$, and stop whenever an algorithm reports a value in their range.
    
    Feed every update to every instance in increasing order of index $i$, initializing new instances whenever needed. Stop whenever an algorithm reports a value in their range.
    
    Let $j$ be the smallest $i$ such that the return value of \Cref{alg:dynamicestimate} on $G_i$ is in $[\lmin, \lmax]$. 
    Return this value multiplied by $\frac 1 {p_jx_j}$.
\end{algo}

\begin{lemma}
   For every $i$ such that $\lambda_i \le \lambda$, if the return value of \Cref{alg:dynamicestimate} on $G_i$ is smaller than $\lmax$, then this value multiplied by $\frac 1 {p_i x_i}$ is at most a $(1+\epsilon)^4$ multiplicative factor away from the minimum cut value in $G$.
\end{lemma}

\begin{proof}
The argument is as follows: we will show that if $\lambda_i \le \lambda$ and the return value of \Cref{alg:dynamicestimate} is smaller than $\lmax$, then in $G_i$ the minimum cut is of value between $\lmin$ and $\lmax$.
This in turn ensures that the cut found by \Cref{alg:dynamicestimate} is a $(1+\epsilon)$-approximation of the minimum cut of $G_i$, and the fact that $\lambda_i \le \lambda$ ensures that $p_i$ is large enough -- that is, we sampled enough edges -- so that we can ensure with high probability that the found cut in an approximate minimum cut of $\hat G_i$, which in turn is an approximate minimum cut of $G$.
Let $\mu$ be the value of the minimum cut in $\hat G_i$.
We have that $\mu \le \lambda x_i \le \mu + n^2$ and thus $\mu \ge \lambda x_i - n^2 \ge \lambda x_i - \epsilon \lambda_i x_i\ge (1-\epsilon)\lambda_i x_i$.
    We have that $p_i = \frac {54 \ln n} {\epsilon^2 \lambda_ix_i} \ge \frac {54 \ln n} {\epsilon^2 \lambda x_i}$, and thus by \Cref{thm:kargersparsify}, with high probability the minimum cut in $G_i$ is at most a $(1+\epsilon)$ or $(1-\epsilon)$ factor away from $p_i \mu$, and thus is at least $(1-\epsilon) p_i \mu  \ge (1-\epsilon)^2\lambda x_i \frac { 54 \ln n} {\epsilon^2 \lambda_i x_i}\ge \lmin$. 

    If moreover \Cref{alg:dynamicestimate} finds a cut of value at most $\lmax$, this ensures that the minimum cut of $G_i$ is of cut-size at most $\lmax$, and thus ensures that \Cref{alg:dynamicestimate} found a $(1+\epsilon)$-aproximate minimum cut of $G_i$. 
    Multiplying the cut value by $\frac 1 {p_ix_i}$ gives the desired result, by \Cref{lem:goodrange}.
\end{proof}

\section{Acknowledgements}

Funded by the European union. Views and opinions expressed are however those of the author(s) only and do not necessarily reflect those of the European Union or the European Research Council Executive Agency. Neither the European Union nor the granting authority can be held responsible for them.

This project has received funding from the European Research Council (ERC) under the European Union's Horizon 2020 research and innovation programme (MoDynStruct, No. 101019564)  \includegraphics[width=1.6cm]{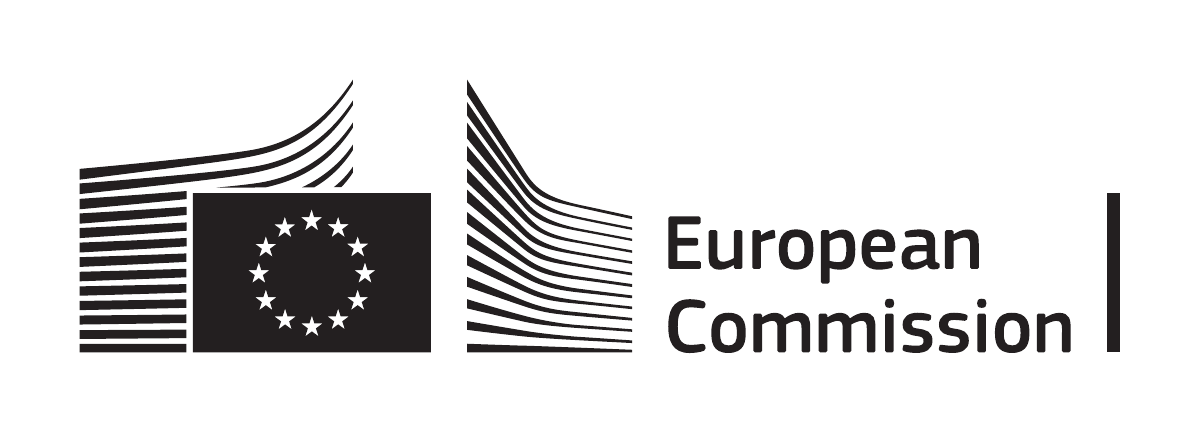} and the Austrian Science Fund (FWF) grant  \href{https://www.doi.org/10.55776/I5982}{DOI 10.55776/I5982}. 
For open access purposes, the author has applied a CC BY public copyright license to any author-accepted manuscript version arising from this submission.

\bibliographystyle{siamplain}
\bibliography{99-bibliography}

\end{document}